\documentclass[showpacs,aps,pra,longbibliography,onecolumn,notitlepage,nofootinbib]{revtex4-2}

\usepackage[dvips]{graphicx} 
\usepackage{amsmath,amssymb,amsthm,mathrsfs,amsfonts,dsfont}
\usepackage{enumerate}
\usepackage{epsfig}
\usepackage{subfigure}
\usepackage{xcolor}
\usepackage{amsthm}
\usepackage{braket}
\usepackage{multirow}    
\usepackage{diagbox}
\usepackage{float}
\usepackage{array}     
\usepackage{comment}
\usepackage{algorithm}
\usepackage{algpseudocode}
\usepackage{diagbox}
\usepackage{enumitem} 
\usepackage{booktabs,colortbl} 
\usepackage{rotating}
\usepackage[normalem]{ulem}
\usepackage[most]{tcolorbox}
\usepackage{enumitem} 
\usepackage[colorlinks = true]{hyperref}
\newcommand{\comments}[1]{}

\newcommand{\tabincell}[2]{\begin{tabular}{@{}#1@{}}#2\end{tabular}} 
\newcommand{\e}{\mathrm{e}}
\renewcommand{\i}{\mathrm{i}}

\tcbmaketheorem{Boxes}{Box}{
breakable,
fonttitle=\sffamily\bfseries\large, fontupper=\sffamily,fontlower=\sffamily,colframe=red!75!black,colback=white
}{boxes}{box}

\newtheorem{theorem}{Theorem}
\newtheorem{lemma}{Lemma}
\newtheorem{corollary}{Corollary}

\newtheorem{definition}{Definition}

\newtheorem{observation}{Observation}

\newcommand{\red}[1]{{\color{red} #1}}
\newcommand{\blue}[1]{{\color{blue} #1}}

\newcommand{\mbb}{\mathbb}
\newcommand{\mc}{\mathcal}

\begin{document}
\title{Quantum key distribution surpassing the repeaterless rate-transmittance bound without global phase locking}

\author{Pei Zeng}
\affiliation{Center for Quantum Information, Institute for Interdisciplinary Information Sciences, Tsinghua University, Beijing 100084, China}
\author{Hongyi Zhou}
\affiliation{Center for Quantum Information, Institute for Interdisciplinary Information Sciences, Tsinghua University, Beijing 100084, China}
\author{Weijie Wu}
\affiliation{Center for Quantum Information, Institute for Interdisciplinary Information Sciences, Tsinghua University, Beijing 100084, China}
\author{Xiongfeng Ma}
\affiliation{Center for Quantum Information, Institute for Interdisciplinary Information Sciences, Tsinghua University, Beijing 100084, China}

\begin{abstract}
Quantum key distribution --- the establishment of information-theoretically secure keys based on quantum physics --- is mainly limited by its practical performance, which is characterised by the dependence of the key rate on the channel transmittance $R(\eta)$.
Recently, schemes based on single-photon interference have been proposed to improve the key rate to $R=O(\sqrt{\eta})$ by overcoming the point-to-point secret key capacity bound with interferometers.
Unfortunately, all of these schemes require challenging global phase locking to realise a stable long-arm single-photon interferometer with a precision of approximately 100 nm over fibres that are hundreds of kilometres long.
Aiming to address this problem, we propose a mode-pairing measurement-device-independent quantum key distribution scheme in which the encoded key bits and bases are determined during data post-processing. Using conventional second-order interference, this scheme can achieve a key rate of $R=O(\sqrt{\eta})$ without global phase locking when the local phase fluctuation is mild. We expect this high-performance scheme to be ready-to-implement with off-the-shelf optical devices.

\end{abstract}

\maketitle


Quantum key distribution (QKD)~\cite{bennett1984quantum,ekert1991Quantum} is currently the most successful application of quantum information science and serves as the first stepping stone towards a future quantum communication network~\cite{Chen2021integrated}.
A core advantage of QKD compared to other quantum communication tasks is that it is ready to implement with current commercially available off-the-shelf optical devices.
However, two major characteristics of QKD --- its practical security and key-rate performance --- limit its real-life implementation. The key generation speed suffers heavily from transmission loss in the optical channel.
Fundamentally, the asymptotic key rate for point-to-point QKD schemes is upper bounded by the repeaterless rate-transmittance bound~\cite{takeoka2014fundamental,pirandola2017fundamental}, which is approximately a linear function of the transmittance, $R\le O(\eta)$. 
Quantum repeaters~\cite{EntSwap1993,Briegel1998Repeater,azuma2015all} have been proposed as a radical solution to this problem. 
Unfortunately, none of the quantum repeater proposals are easy to implement in the near term.

In real-life use, the deviation of the realistic behaviour of physical devices from their ideal ones gives rise to critical issues in practical security. There are many quantum attacks that can take advantage of the loopholes introduced by device imperfections~\cite{Xu2020Secure}. A typical QKD system can be divided into three parts: source, channel, and measurement. The security of the channel has been well addressed in the security proofs for QKD \cite{lo1999Unconditional,shor2000Simple,koashi2009simple}. The source is relatively simple and can be well characterised \cite{gottesman2004security}. In contrast, the measurement device, is complicated and difficult to calibrate. Moreover, an adversary could manipulate the measurement device by sending unexpected signals \cite{makarov2008effects,qi2007time}. To solve this implementation security problem, measurement-device-independent quantum key distribution (MDI-QKD) schemes have been proposed to close the detection loopholes once and for all~\cite{lo2012Measurement}. Various experimental systems have been successfully demonstrated~\cite{rubenok2013real,liu2013experimental,silva2013proof,woodward2021gigahertz}, with extension to a communication network~\cite{Tang2016MDInet}. 


A generic MDI-QKD setup is shown in Fig.~\ref{fig:MDIQKDcomp}(a).
Each of the two communicating parties, Alice and Bob, holds a quantum light source, encodes random bits into quantum pulses, and sends these pulses to a measurement site through lossy channels.
Measurement devices are possessed by an untrusted party, Charlie, who is supposed to correlate Alice's and Bob's signals via interference detection.
Based on the detection results announced by Charlie, Alice and Bob sift the local random bits encoded in the pulses to generate secure key bits.
Note that the security of MDI-QKD schemes does not rely upon the physical implementation of the detection devices.
Alice and Bob need to trust only their own locally encoded quantum sources.
Since neither Alice nor Bob receives quantum signals from the channel during key distribution, any hacker's attempt to manipulate the users' devices becomes extremely difficult compared to regular QKD schemes~\cite{makarov2008effects,qi2007time}.


\begin{figure}[htbp!]
\centering \includegraphics[width=8cm]{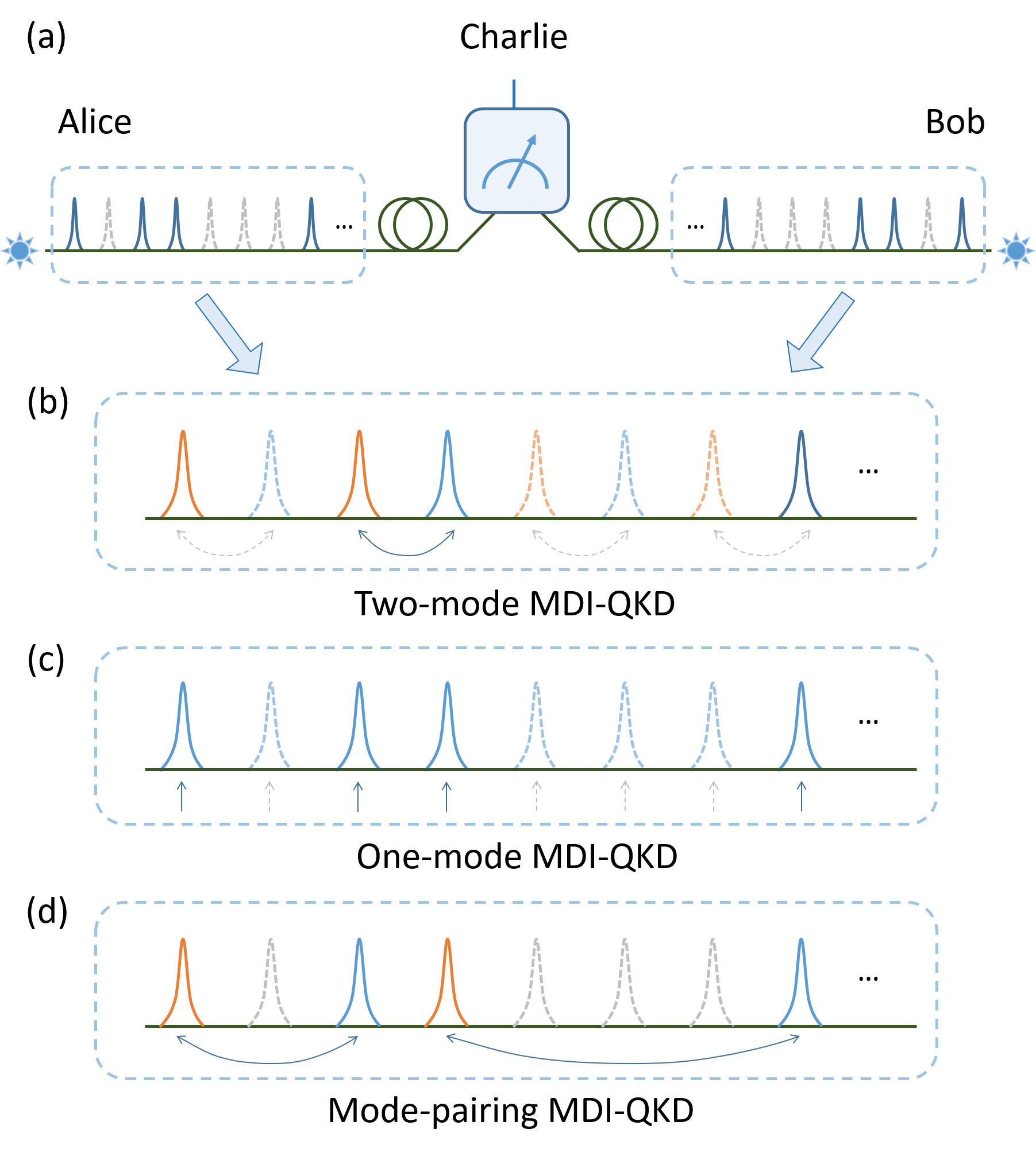}
\caption{ Comparison of two-mode, one-mode and mode-pairing MDI-QKD schemes. (a) Schematic diagram of a generic MDI-QKD scheme. The solid and dashed pulses yield successful and unsuccessful detection, respectively, at the measurement site. For (b),(c), and (d), each wave packet in the diagram represents two independent pulses emitted simultaneously by Alice and Bob. (b) In two-mode MDI-QKD schemes, the pairing of the blue pulses (as phase references) and orange pulses (as signals) is predetermined, necessitating coincidence detection. (c) In one-mode MDI-QKD schemes, there is no phase reference pulse, necessitating global phase locking. (d) In the mode-pairing MDI-QKD scheme, in accordance with the detection results, Alice and Bob pair the clicked pulses and assign them to be either reference or signal pulses, such that neither coincidence detection nor global phase locking is required.} \label{fig:MDIQKDcomp}
\end{figure}

Strictly speaking, MDI-QKD is not a point-to-point scheme, as there is an interference site between Alice and Bob.
Consequently, it is not necessarily limited by the repeaterless rate-transmittance bound.
Nevertheless, the original MDI-QKD scheme~\cite{lo2012Measurement}, in which Alice and Bob both encode a `dual-rail' qubit into a single-photon subspace on two polarization modes, unfortunately cannot overcome this bound. Later, alternative schemes were proposed~\cite{tamaki2012phase,Ma2012alternative} in which the qubit is encoded into two optical time bins. We refer to schemes of this type as two-mode MDI-QKD, in the sense that the single-side key information is encoded in the relative phase of the coherent states in the two orthogonal optical modes, i.e., second-quantized electromagnetic fields. To correlate Alice's and Bob's encoded information in a two-mode scheme, a successful two-photon interference measurement is required. If either Alice or Bob's emitted photon is lost in transmission, there will be no conclusive detection result. For example, in the time-bin encoding scheme~\cite{Ma2012alternative} shown in Fig.~\ref{fig:MDIQKDcomp}(b), Alice and Bob each emit a qubit encoded in two time-bin modes, with Alice emitting $A_1$ and $A_2$ and Bob emitting $B_1$ and $B_2$. Only when both the interference between modes $A_1$ and $B_1$ and that between $A_2$ and $B_2$ yield successful detection can Alice restore Bob's raw key information. Thus, successful interference requires a coincidence detection. Due to this coincidence-detection requirement, rounds with only a single detection are discarded, resulting in a relatively low key generation rate --- one that is a linear function of the transmittance, $O(\eta)$. From the perspective of practical implementation, however, the coincidence detection also has certain merits. This approach can ensure stable optical interference, while Alice and Bob need only to stabilise the relative phases between the two modes.

Coincidence detection is the essential factor that prevents MDI-QKD from overcoming the linear key-rate bound. To eliminate this requirement, a new type of MDI-QKD scheme called twin-field quantum key distribution (TF-QKD) based on encoding information into a single optical mode have been proposed~\cite{lucamarini2018overcoming}, illustrated in Fig.~\ref{fig:MDIQKDcomp}(c).
Later on, variants of TF-QKD have been proposed, among which the key information in encoded in either the phase~\cite{Ma2018phase,lin2018simple} (known as phase-matching QKD) or the intensity~\cite{wang2018twin} (known as sending-or-not-sending TF-QKD) of coherent states.
In this work, we refer to these twin-field-type schemes as one-mode MDI-QKD schemes for a conceptual comparison to the traditional two-mode MDI-QKD schemes, since the single-side information in these schemes is encoded into a single optical mode in each round.
Similar to the Duan-Lukin-Cirac-Zoller-type repeater design~\cite{duan2001long}, such one-mode schemes use single-photon interference instead of coincidence detection, hence yielding a quadratic improvement in key rate compared to two-mode schemes~\cite{lucamarini2018overcoming,Ma2018phase,lin2018simple}.
As a result, they can overcome the point-to-point linear key-rate bound~\cite{takeoka2014fundamental,pirandola2017fundamental}.
Unfortunately, one-mode schemes are more challenging to implement due to the unstable optical interference resulting from the lack of global phase references.
For example, in the phase-matching QKD (PM-QKD) scheme~\cite{Ma2018phase}, the key information is encoded into the \textit{global phase} of Alice's and Bob's coherent states.
The phases of the coherent states generated by two remote and independent lasers need to be matched at the measurement site.
A small phase drift or fluctuation caused by the lasers and/or channels is hazardous for key generation.

At first glance, it seems that we cannot simultaneously enjoy the advantages of one-mode schemes (i.e., quadratic improvement in successful detection) and two-mode schemes (i.e., stable optical interference), due to an intrinsic trade-off between the information-encoding efficiency and robustness. On the one hand, the relative information among different optical modes is more difficult to retrieve when the channel loss is large. On the other hand, the global phase of a coherent state is not as stable as the relative phase between two coherent states travelling through the same quantum channel.
In a typical 200-km fibre with a telecommunication frequency of 1550 nm, the phase of a coherent state is susceptible to small fluctuations in the optical transmission time ($\sim 10^{-15}$~s), optical length ($\sim 200$~nm) and light frequency ( $\sim 100$~kHz).
Recently, experimentalists have made great efforts to demonstrate high-performance in one-mode schemes, utilising high-end technologies to perform a precise control operation to stabilise the global phase by locking the frequency and phase of the coherent states~\cite{minder2019experimental,wang2019beating,fang2020implementation,zhong2020proofofprinciple, chen2020sending,pittaluga2021600,clivati2020coherent,wang2022twin}. However, this significantly increases the experimental difficulty and undermines the applicability of one-mode schemes in real life.

In this work, we propose a mode-pairing MDI-QKD scheme that can offer both --- simple implementation and high performance. Hereafter, we refer to this scheme as the mode-pairing scheme for simplicity.
By observing that the majority detection events are single-clicks and are discard in the two-mode MDI-QKD schemes, we try to recycle the discarded single-click in the mode-pairing scheme.
To do that, the coherent states in the transmitted modes are initially prepared independently with randomly encoded information.
Based on the fact that \emph{the two detection events used to read out the encoded information do not need to occur at two predetermined locations}, the key is extracted from two paired detection events rather than coincidence detection, as shown in Fig.~\ref{fig:MDIQKDcomp}(d). This offers a quadratic improvement akin to that of one-mode schemes when the local phases can be stabilized using currently available phase stabilization techniques.
Moreover, key information about the mode-pairing scheme is encoded in the relative phases or intensities, whose stability relies only upon the conditions of the local phase references and optical paths.
Therefore, the technical complexity is similar to that of two-mode schemes, which have been widely implemented both in the laboratory~\cite{rubenok2013real,liu2013experimental,silva2013proof,Tang2014Experimental} and in the field~\cite{Tang2014field,Tang2016MDInet}.
Notably, to adapt to different hardware conditions, the mode-pairing scheme can be freely tuned between the one-mode and two-mode schemes by adjusting a pulse-interval parameter (as discussed later in Section~\ref{ssec:modepairing}) during data postprocessing to optimise the system performance.



\section{Results}


\subsection{Mode-pairing scheme} \label{Sec:MPscheme}
In the mode-pairing scheme, Alice and Bob first prepare coherent states with independently and randomly chosen intensities and phases in each emitted optical mode.
These coherent states are sent to the untrusted measurement site, Charlie. Based on Charlie's announced measurement results, Alice and Bob pair the optical modes with successful detection and determine the key bits and bases for each mode pair locally. They then sift the bases and generate secure key bits via postprocessing. The scheme is introduced in Box~\ref{box:MPprotocol} and illustrated in Fig.~\ref{fig:MPprotocol}. For simplicity of the introduction of the main protocol design, we omit the details of the decoy-state method~\cite{Lo2005Decoy} and discrete phase randomisation here. A complete description of the mode-pairing scheme is given in the Methods section, subsection~\ref{method:decoyscheme}.

\begin{figure*}[htbp]
\centering \includegraphics[width=16cm]{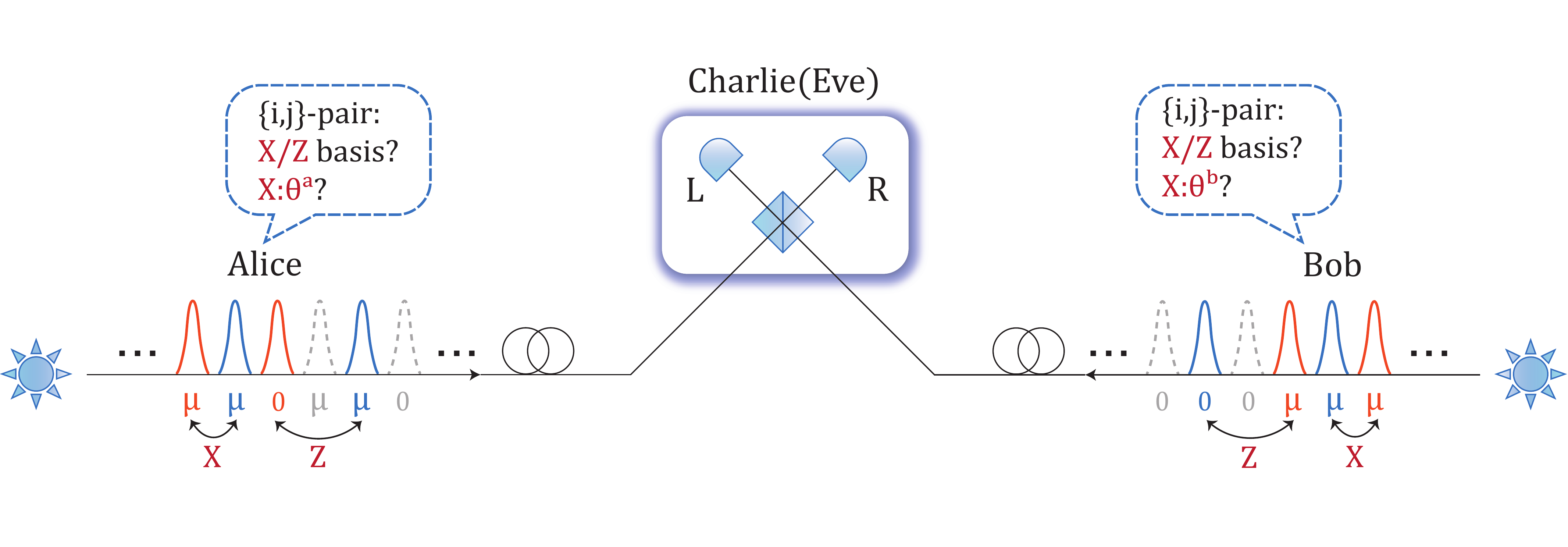}
\caption{Schematic diagram of the mode-pairing MDI-QKD scheme. The two communicating parties, Alice and Bob, first prepare coherent pulses with random intensities chosen from $\{0,\mu\}$ and random phases $\phi^{a(b)}_i\in[0,2\pi)$ and send them to an untrusted measurement site, Charlie. After interference measurement, Charlie announces the detection results. Based on the detection results, Alice and Bob pair the pulses and determine their encoding bases. For $X$-pairs, they announce the alignment angles $\theta^a$ and $\theta^b$ and maintain data for which $\theta^a = \theta^b$. They use $Z$-pairs to generate keys and other data for parameter estimation.} \label{fig:MPprotocol}
\end{figure*}

\begin{Boxes}{Mode-pairing scheme}{MPprotocol}
1.\textbf{State preparation}: In the $i$-th round ($i=1,2,...,N$), Alice prepares a coherent state $\ket{\sqrt{\mu^a_i}\e^{\i\phi^a_i}}$ in optical mode $A_i$ with an intensity $\mu^a_i$ randomly chosen from $\{0,\mu\}$ and a phase $\phi^a_i$ uniformly chosen from $[0, 2\pi)$. Similarly, Bob randomly chooses $\mu^b_i$ and $\phi^b_i$ and prepares $\ket{\sqrt{\mu^b_i}\e^{\i\phi^b_i}}$ in mode $B_i$.

2.\textbf{Measurement}: Alice and Bob send modes $A_i$ and $B_i$ to Charlie, who performs single-photon interference measurements. Charlie announces the click patterns for both detectors $L$ and $R$.


Alice and Bob repeat the above two steps for $N$ rounds. Then, they postprocess the data as follows.
\tcblower

3.\textbf{Mode pairing}: For all rounds with successful detection, in which one and only one of the two detectors clicks, Alice and Bob apply a strategy of grouping two clicked rounds as a pair. The encoded phases and intensities in these two rounds form a data pair. A simple pairing strategy is introduced in Section~\ref{ssec:modepairing}.

4.\textbf{Basis sifting}: Based on the intensities of the two grouped rounds indexed by $i$ and $j$, Alice labels the `basis' of the data pair as $Z$ if the intensities are $(0,\mu)$ or $(\mu,0)$, as $X$ if the intensities are $(\mu,\mu)$, or as `0' if the intensities are $(0,0)$. Bob sets the basis using the same method. Alice and Bob announce the basis of each data pair; if they both announce the basis $X$ or $Z$, they maintain the data pairs, whereas otherwise, the data pairs are discarded.

5.\textbf{Key mapping}:
For each $Z$-basis pair ($Z$-pair for simplicity) at locations $i$ and $j$, Alice sets her key as $\kappa^a = 0$ if $(\mu^a_i,\mu^a_j)= (0,\mu)$ and $\kappa^a = 1$ if $(\mu^a_i,\mu^a_j) = (\mu,0)$. For each $X$-basis pair ($X$-pair for simplicity) at locations $i$ and $j$, the key is extracted from the relative phase $(\phi^a_j - \phi^a_i) = \theta^a + \pi\kappa^a$, where the raw key bit is $\kappa^a = \left\lfloor((\phi^a_j - \phi^a_i)/\pi \text{ mod } 2)\right\rfloor$ and the alignment angle is $\theta^a := (\phi^a_j - \phi^a_i) \text{ mod } \pi$. In a similar way, Bob assigns his raw key bit $\kappa^b$ and determines $\theta^b$. The difference in the key mapping for $Z$-pairs is that, Bob sets the raw key bit $\kappa^b$ as $0$ if $(\mu^b_i,\mu^b_j)=(\mu,0)$ and $\kappa^b = 1$ if $(\mu^b_i,\mu^b_j) = (0,\mu)$.
As an extra step on the $X$-pairs, if Charlie's detection announcement is $(L,L)$ or $(R,R)$, Bob keeps the bit $\kappa^b$; otherwise, if Charlie's announcement is $(L,R)$ or $(R,L)$, Bob flips $\kappa^b$.
For the $X$-pairs, Alice and Bob announce the alignment angles $\theta^a$ and $\theta^b$. If $\theta^a = \theta^b$, then the data pairs are kept; otherwise, the data pairs are discarded.

6.\textbf{Parameter estimation}: Alice and Bob estimate the fraction of clicked signals $q_{(1,1)}$ and the corresponding phase error rate $e^X_{(1,1)}$ of $Z$-pairs where Alice and Bob both emit a single photon at locations $i$ and $j$, using the data of the $Z$-pairs and $X$-pairs. They also estimate the quantum bit error rate $E^{(\mu,\mu),Z}$ of the $Z$-pairs.

7.\textbf{Key distillation}: Alice and Bob use the $Z$-pairs to generate a key. They perform error correction and privacy amplification on the basis of $q_{(1,1)}$, $E^{(\mu,\mu),Z}$ and $e^{X}_{(1,1)}$.
\end{Boxes}

In the mode-pairing scheme, we mainly consider the keys generated from the $Z$-pair data, since they have a much lower quantum bit error rate $E^Z_{\mu\mu}$ than the $X$-pair data.
The encoding of the mode-pairing scheme in Box~\ref{box:MPprotocol} originates from the time-bin encoding MDI-QKD scheme~\cite{Ma2012alternative}. If Alice's two paired optical modes $\{A_i, A_j\}$ are assigned to the $Z$-basis, then the state of the two optical modes is either $\ket{0}_{A_i}\ket{\sqrt{\mu}\e^{\i\phi^a_j}}_{A_j}$ or $\ket{\sqrt{\mu}\e^{\i\phi^a_i}}_{A_i}\ket{0}_{A_j}$, where $\phi^a_i$ and $\phi^a_j$ are two independent random phases. We can write the encoded states in a unified form:
\begin{equation} \label{eq:Zbasisstate}
\ket{\psi^a_Z}_{A_i,A_j} = \ket{\sqrt{\kappa^a\mu}\e^{\i\phi^a_i}}_{A_i}\ket{\sqrt{\bar{\kappa}^a\mu}\e^{\i\phi^a_j}}_{A_j},
\end{equation}
where $\kappa^a$ is the encoded key information and $\bar{\kappa} := \kappa\oplus 1$ is the inverse of $\kappa$. In the other case, in which the two optical modes $\{A_i, A_j\}$ are assigned to the $X$-basis, we can rewrite their two independent random phases $\phi^a_i$ and $\phi^a_j$ as
\begin{equation}
\begin{aligned}
\phi^a &:= \phi^a_i \in [0,2\pi), \\
\phi^a_{\delta} &:= \phi^a_j - \phi^a_i \in [0,2\pi).
\end{aligned}
\end{equation}
In this way, the phase $\phi^a$ becomes a global random phase on the pulse pair, while $\phi^a_\delta$ is the relative phase for quantum information `encoding'. Due to the independence of $\phi^a_i$ and $\phi^a_j$, the phases $\phi^a$ and $\phi^a_\delta$ are also independent of each other and uniformly range from $[0,2\pi)$. By definition, we have $\phi^a_\delta = \theta^a + \pi\kappa^a$. Then, the $X$-pair state can be written as,
\begin{equation} \label{eq:Xbasisstate}
\ket{\psi^a_X}_{A_i,A_j} = \ket{\sqrt{\mu^a}\e^{\i\phi^a}}_{A_i}\ket{\sqrt{\mu^a} \e^{\i(\phi^a + \theta^a + \kappa^a\pi)}}_{A_j},
\end{equation}
where $\mu^a\in\{0,\mu\}$. When $\theta=0$ or $\pi/2$, Alice emits $X$-basis or $Y$-basis states, respectively, as used in the time-bin encoding MDI-QKD scheme~\cite{Ma2012alternative}.

We remark that in either the $Z$-pair state in Eq.~\eqref{eq:Zbasisstate} or the $X$-pair state in Eq.~\eqref{eq:Xbasisstate}, there is a global random phase $\phi^a$, which will not be revealed publicly. With this (global coherent state) phase randomisation, the emitted $Z$- and $X$-pair states can be regarded as a mixture of photon number states~\cite{Lo2005Decoy}. Then, Alice and Bob can estimate the detections caused by the pairs where they both emit single photons and use them to generate secure keys, in a manner similar to traditional two-mode schemes. Therefore, the security of the mode-pairing scheme is similar to that of two-mode schemes.
Nevertheless, the mode-pairing scheme in Box~\ref{box:MPprotocol} has the following unique features.
\begin{enumerate}
\item
The emitted states in different optical modes $\{A_i\}$ are independent and identically distributed (i.i.d.). Therefore, the information encoded in different optical modes is completely decoupled.
\item
Based on the postselection of clicked signals, different optical modes are paired \emph{afterwards}. The relative information between the two modes is then converted into raw key data.
\end{enumerate}

In the mode-pairing scheme, the key information is determined not in the state preparation step, but by the detection location, sharing some similarities with the differential-phase-shifting QKD scheme~\cite{inoue2003differential,sasaki2014practical}. It is the untrusted measurement site that determines the location of successful detection and thereby affects the pairing setting. The `dual-rail' qubits encoded on the single photons are `postselected' on the basis of this detection. By virtual of the independence of the optical modes, the information encoded in the `postselected' qubits cannot be revealed from other optical pulses.

For another comparison, the sending-or-not-sending (SNS) TF-QKD scheme~\cite{wang2018twin} also uses a $Z$-basis time-bin encoding, whereby either Alice or Bob emits an optical mode to generate key bits. The state preparation of the mode-pairing scheme shares similarities with the SNS TF-QKD scheme. However, the information of the mode-pairing scheme is encoded into the relative information between the two optical modes. As a result, the basis-sifting and key mapping of the mode-pairing scheme follow different logic originated from the time-bin encoding MDI-QKD scheme~\cite{Ma2012alternative}. Note that in the SNS scheme, bits 0 and 1 are highly biased in the $Z$ basis, whereas in the mode-pairing scheme, they are evenly distributed.

A critical issue in the security analysis of the mode-pairing scheme is to \emph{maintain the flexibility to determine in which two optical modes to perform the overall photon number measurement until Charlie announces the detection results}. Note that, in the original two-mode QKD schemes, the encoders can always be assumed to perform an overall photon number measurement and post-select the single-photon components as good `dual-rail' qubits \emph{before they emit their signals to Charlie}. In the mode-pairing scheme, however, this is not viable because the optical pulse pair, for which the single-photon component is defined, is postselected based on Charlie's detection announcement. To solve this problem, we introduce source replacement for the random phases in the coherent states to purify them as ancillary qudits and define an indirect overall photon number measurement on them. The source-replacement procedure can be found in the Methods section, subsection~\ref{method:sourceReplace}. Conditioned on the indirect overall photon number measurement result to be single-photon states, the $X$-basis error rate fairly estimates the $Z$-basis phase error rate for the signals for which Alice and Bob both emit single photons.

In Appendix~\ref{Sec:Security}, we provide a detailed security proof based on entanglement distillation.
The main idea is to introduce a `fixed-pairing' scheme, in which the pairing setting, i.e., which locations are paired together, is predetermined and hence independent of Charlie's announcement. We first prove that, with any given pairing setting, the fixed-pairing scheme is secure, as it can be reduced to a two-mode MDI-QKD scheme. Afterwards, we examine the private state generated by the mode-pairing scheme and prove that it is the same as that of a fixed-pairing scheme under all possible measurements that Charlie could perform and announcement methods. In this way, we prove the equivalence of the mode-pairing scheme to a group of fixed-pairing schemes with different pairing settings.


\subsection{Pairing strategy} \label{ssec:modepairing}
The pairing strategy mentioned in Step~3 lies at the core of the mode-pairing scheme in Box~\ref{box:MPprotocol}, which correlates two independent signals and determines their bases and key bits. Note that the relative phase between two paired quantum signals determines the key information in the $X$ basis. When the time interval between these two pulses becomes too large, the key information suffers from phase fluctuation, which is charactesized by the laser coherence time. Therefore, Alice and Bob should establish a maximal pairing interval $l$, such that the number of pulses between the two paired signals should not exceed $l$. In practice, $l$ can be estimated by multiplying the laser coherence time by the system repetition rate.

Here, we consider a simple pairing strategy in which Alice pairs adjacent detection pulses together if the time interval between them is not too large ($\le l$). The details are shown in Algorithm \ref{alg:pairing} and illustrated in Fig.~\ref{fig:directpair}. Charlie's announcement in the $i$-th round is denoted by a Boolean variable $C_i$ that indicates whether the detection is successful. That is, $C_i=1$ implies that either the detector $L$ or $R$ clicks. Otherwise, there is no click or double clicks.

\begin{figure}[htbp]
\centering
\includegraphics[width=8cm]{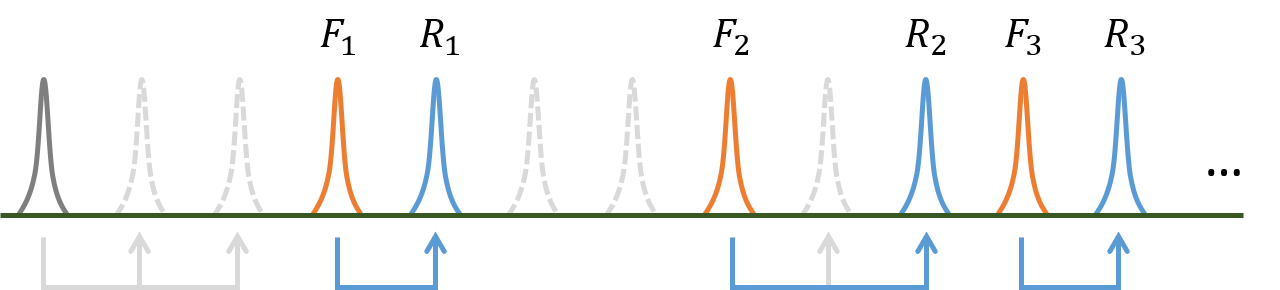}
\caption{Schematic diagram of the simple pairing strategy. Here, we set $l=2$. The solid and dashed pulses are those with and without successful detection, respectively. Orange and blue pulses are, respectively, the front pulses and rear pulses that succeed in pairing within $l=2$ pulses, while grey pulses are the ones fail in pairing. The labels $F_k$ and $R_k$ represent the front and rear pulses, respectively, in the $k$-th successful pair.} \label{fig:directpair}
\end{figure}

\begin{algorithm}[H]
\caption{Simple pairing strategy}
\label{alg:pairing}
\begin{algorithmic}[1]
\Require
Charlie's announced detection results $C_i$ for $i=1$ to $N$; maximal pairing interval $l$.
\Ensure
$K$ pairs; front- and rear-pulse locations $(F_k,R_k)$ for the $k$-th pair, where $k=1$ to $K$.
\State Initialise the pairing index $k:= 1$; initialise the flag $f:=0$.
\For{$i= 1~\text{\textbf{to}}~N$} \Comment{Enumerating all locations}
	\If{$f=0$}  \Comment{Searching for the front-pulse location}
		\If{$C_i = 1$} \Comment{Successful detection}
			\State Set the temporary front-pulse location to $F_k:=i$; set the flag to $f:=1$.
		\EndIf
    \Else \Comment{Searching the rear-pulse location}
        \If{$C_i = 1$} \Comment{Successful detection}
        \State Set the rear-pulse location to $R_k:=i$; update the pairing index to $k:=k+1$; reset the flag to $f:=0$.
        \ElsIf{$F_k - i\geq l$} \Comment{Pairing interval exceeding $l$}
        \State Reset the flag to $f:=0$.
		\EndIf
	\EndIf
\EndFor
\State Set the total number of pairs to $K:=k-1$.
\end{algorithmic}
\end{algorithm}

To check the efficiency of this pairing strategy, let us calculate the pairing rate $r_p$ (i.e. the average number of pairs generated per pulse). We assume that Alice and Bob choose intensities $0$ and $\mu$ with equal probability, maximising the number of successful pairs in the $Z$ basis. With a typical QKD channel model, the pairing rate $r_p$ is calculated as shown in the Methods Section \ref{method:pairrate},
\begin{equation} \label{eq:rp}
r_p(p, l) = \left[ \frac{1}{p [1 - (1-p)^l]} + \frac{1}{p} \right]^{-1},
\end{equation}
where $p$ is the probability that the emitted pulses result in a click event, given approximately by $\eta_s\mu$. Here, $\eta_s$ and $\eta$ denote the channel transmittance from Alice to Charlie and the total transmittance from Alice to Bob, respectively. When the channel is symmetric for Alice and Bob, we have $\eta = \eta_s^2$.
An explicit simulation formula for $p$ in a pure-loss channel is given in Appendix~\ref{Sec:Simu}.
Note that both the pairing ratio $r_p$ and the detection probability $p$ can be directly obtained by experimentation.

The raw key rate mainly depends on the pairing rate $r_p$. Now, let us check the scaling of $r_p$ with the channel transmittance in the symmetric-channel case. If the local phase reference is sufficiently stable, then the maximal interval can be set to $l\to +\infty$. In this case,
\begin{equation}
r_p=\frac{p}{2}\approx \frac{\eta_s\mu}{2} = O(\sqrt{\eta}),
\end{equation}
where the optimal intensity is $\mu=O(1)$, as evaluated in Appendix~\ref{Sec:optimalmu}.
On the other hand, if the local phase reference is not at all stable, one must set $l=1$; then,
\begin{equation}
r_p = \frac{p^2}{1+p} \approx \frac{\eta_s^2 \mu^2}{1+ \eta_s\mu} = O(\eta).
\end{equation}
In this case, the experimental requirements for the mode-pairing scheme are close to those of the existing time-bin MDI-QKD scheme~\cite{Ma2012alternative}.

In practice, $l$ can be adjusted in accordance with the laser quality and quantum-channel fluctuations. Note that $l$ can also be adjusted during data postprocessing, offering flexibility for various environmental changes in real time. Generally, the whole pairing strategy can be adjusted through different realisations.

\subsection{Practical issues and simulation} \label{Ssec:simulation}
The key rate of the mode-pairing scheme, as rigorously analysed in Appendix~\ref{Sec:Security}
, has a decoy-state MDI-QKD form:
\begin{equation}
R = r_p  r_s \left\{q_{(1,1)}\left[ 1 - H(e^X_{(1,1)}) \right] - f H(E^{(\mu,\mu),Z}) \right\},
\end{equation}
where $r_p$ is the pairing rate contributed by each block, $r_s$ is the proportion of $Z$-pairs among all the generated location pairs (approximately $1/8$), $q_{(1,1)}$ is the fraction of $Z$-pairs caused by single-photon-pair states $\rho^{(1,1)}$ in which both Alice and Bob send single-photon states in the two paired modes, $e^X_{(1,1)}$ is the phase error rate of the detection caused by $\rho^{(1,1)}$, $f$ is the error-correction efficiency, and $E^{(\mu,\mu),Z}$ is the bit error rate of the sifted raw data. The fraction $q_{(1,1)}$ and the phase error $e^X_{(1,1)}$ can be estimated using the decoy-state method~\cite{hwang2003decoy,Lo2005Decoy,wang2005decoy}.
A detailed estimation procedure for $q_{(1,1)}$ and $e^X_{(1,1)}$ with the vacuum + weak decoy-state method is introduced in Appendix~\ref{Sec:decoyestimate}.

During the key mapping step in Box~\ref{box:MPprotocol}, the $X$-pair sifting condition $\theta^a = \theta^b$ is impossible to fulfil exactly. This results in insufficient data for $X$-basis error rate estimation. To solve this problem, one can apply discrete phase randomisation~\cite{Cao2015discrete} such that $\theta^a$ and $\theta^b$ are chosen from a discrete set. We expect the discretisation effect to be negligible when the number of discrete phases is reasonably large, such as $D=16$, similar to the situation in previous works on one-mode MDI-QKD~\cite{Zeng2019Symmetryprotected}.

Based on the above analysis, we simulate the asymptotic performance of the mode-pairing scheme under a typical symmetric quantum-channel model, using practical experimental parameter settings. We assign the maximal pairing interval $l$ of the mode-pairing scheme as a value between $1$ and $1\times 10^6$, aiming to illustrate the dependence of the key rate on $l$. We also compare the key rate of the mode-pairing scheme with those of a typical two-mode scheme, time-bin encoding MDI-QKD~\cite{Ma2012alternative}, and two one-mode schemes --- PM-QKD~\cite{Zeng2019Symmetryprotected} and SNS TF-QKD~\cite{jiang2019unconditional}. The simulation results are shown in Fig.~\ref{fig:keysimu}. Here, we compare the asymptotic key rate performance of all the schemes under the scenario of one-way local-operation and classical communication.
The simulation formulas for these schemes are listed in Appendix~\ref{Sec:Simu}.

\begin{figure}[htbp]
\centering \includegraphics[width=18cm]{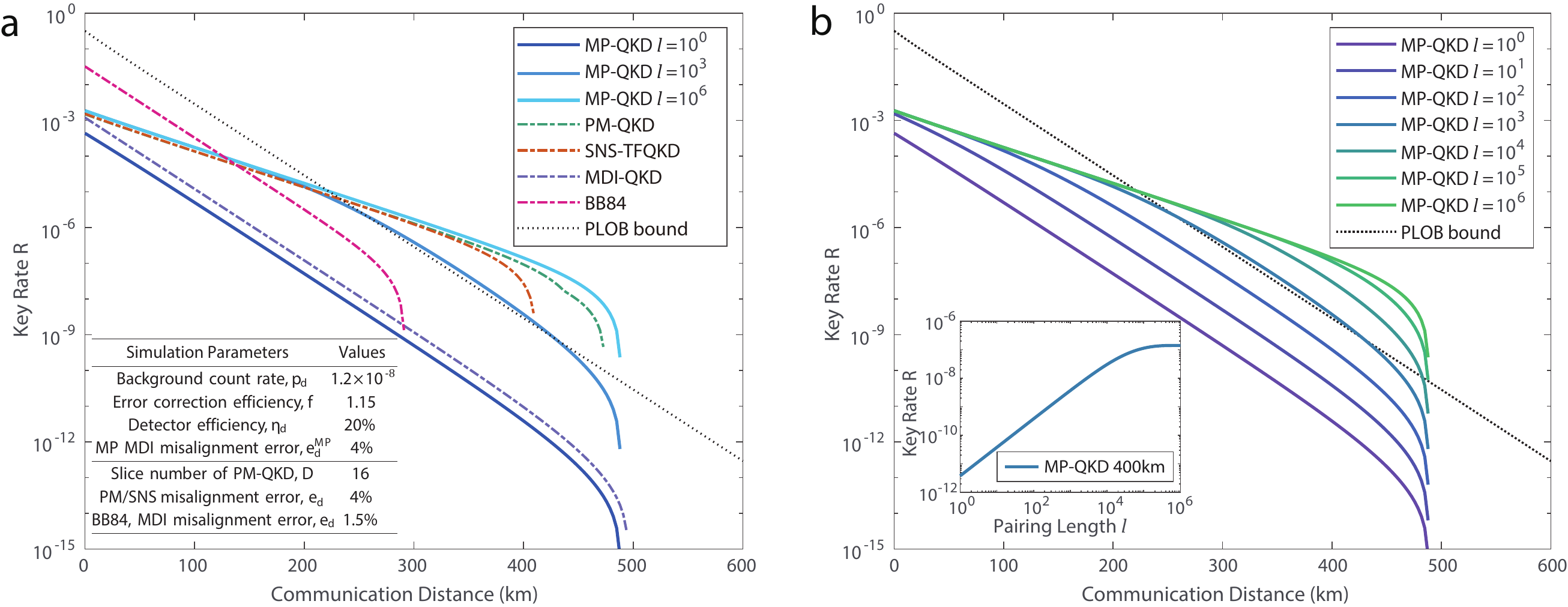}
\caption{Asymptotic key-rate performance of the mode-pairing scheme, with the horizontal axis representing the total communication distance with a fibre loss of $0.2$ dB$/$km and the vertical axis representing the key generation rate. \textbf{a}. Main Panel: Performance comparison of the mode-pairing scheme (denoted by MP-QKD in the plot) with the decoy-state Bennett-Brassard 1984 (BB84)~\cite{bennett1984quantum,gottesman2004security,Lo2005Decoy}, MDI-QKD~\cite{lo2012Measurement}, PM-QKD~\cite{Ma2018phase,Zeng2019Symmetryprotected}, SNS-TFQKD~\cite{wang2018twin,jiang2019unconditional} schemes and the repeaterless rate-transmittance bound (PLOB bound)~\cite{pirandola2017fundamental}. Inset: The simulation parameters used in the key-rate plot, which are mainly from Ref.~\cite{fang2020implementation}. \textbf{b}. Main Panel: The rate-distance dependence of the mode-pairing scheme with different maximal-pairing intervals $l$. Inset: The key rate with respect to the pairing interval $l$ for a communication distance of $400$ km.} \label{fig:keysimu}
\end{figure}

As shown in Fig.~\ref{fig:keysimu}a, the mode-pairing scheme with only neighbour pairing, $l=1$, show a performance comparable to that of the original two-mode scheme. These two schemes have the same scaling property, i.e., $R= O(\eta)$. The deviation is caused by an extra sifting factor in the mode-pairing scheme as a result of independent encoding. When the maximal pairing interval $l$ is increased to $1\times 10^3$, the key rate is significantly enhanced by $3$ orders of magnitude compared to the $l=1$ case, making it able to surpass the linear key-rate bound. If we further increase $l$ above $1\times 10^5$, then the mode-pairing scheme has a similar key rate to PM-QKD and SNS-TFQKD and a scaling property given by $R= O(\sqrt{\eta})$.
In Fig.~\ref{fig:keysimu}b, we further compare the key-rate performance of the mode-pairing scheme under different settings for $l$. When $l$ falls within the range of $1$ to $1\times 10^6$, the key rate of the mode-pairing scheme lies between the two extreme cases of $O(\eta)$ and $O(\sqrt{\eta})$. The key-rate behaviour is dominated by the pairing rate given in Eq.~\eqref{eq:rp}.

In typical optical experiments, the typical line-width of a common commercial laser is $3$ kHz (see for example, Ref.~\cite{fang2020implementation}). Hence, the coherence time of the laser is around $333$ $\mu$s.
In practice, the frequency fluctuation of the lasers will affect the stabilization of the phase. To test the feasibility of the mode-pairing scheme, we perform an interference experiment using a commercial optical communication system with a repetition rate of $625$ MHz. 
The experiment detail is shown in Appendix~\ref{sc:ExpDemo}.
Based on the experimental data, we find that the phase coherence can be maintained well in a time interval of $5$ $\mu$s, correspond to $l=3000\sim 4000$. If we apply the state-of-the-art optical communication system with the repetition rate of $4$ GHz \cite{wang2022twin}, we can realize a pairing interval over $l=20000$.
As an extra remark, our current discussion on the implementation of the mode-pairing scheme is based on the multiplexing of optical time-bin modes. Nonetheless, the proposed mode-pairing design is generic for the multiplexing of other optical degrees of freedom. For example, we can use optical modes with different frequencies, encoding information and interfering them independently, and pair them during the post-processing. This can be used to increase the maximal pairing interval to an even larger value without the global phase locking.
From Fig.~\ref{fig:keysimu}b we can see that the key rate of the mode-pairing scheme with $l=1\times 10^4$ remains $R\sim O(\sqrt{\eta})$ when $\eta$ is smaller than 30~dB, corresponding to a communication distance of 300 km.
The asymptotic key rate of the mode-pairing scheme is $3$ to $5$ orders of magnitude higher than that of the two-mode scheme. We remark that the decoherence effect caused by the optical-fibre channel is negligible compared to the laser coherence time.
When the fibre length is around $500$ km, the velocity of phase drift in the fibre is less than $10$ rad/ms~\cite{fang2020implementation}, which can be calibrated using strong laser pulses without the need for real-time feedback control. As a result, the value of $l$ depends only upon the local phase reference and not the communication distance.

One advantage of the mode-pairing scheme is that it can be adapted to specific hardware conditions. In practice, optical systems may be unstable, causing the local phase reference to fluctuate rapidly. In this case, we can reduce the maximal pairing interval $l$ and search for the optimal pairing strategy during the postprocessing procedure. As shown in the inset plot of Fig.~\ref{fig:keysimu}b, the key rate of the mode-pairing scheme first increases linearly with increasing $l$ before saturating when $l$ is larger than $p^{-1} = (\mu \sqrt{\eta})^{-1}$. In this case, Alice and Bob find successful detection within $l$ locations with a high probability.
Even when the optical system is unstable, the key rate can be nearly $l$ times higher than that of the original time-bin MDI-QKD scheme when the value of $l$ does not exceed $p^{-1} = (\mu \sqrt{\eta})^{-1}$. We remark that, with the original experimental apparatus used in time-bin MDI-QKD, one can directly enhance the key rate by a factor of approximately $100$ using the mode-pairing scheme. On the other hand, we note that for a given communication distance, $l$ does not need to be very large to reach the maximal key-rate performance. For example, when the distance reaches $200$ km, a maximal pairing interval of $l=1000$ is sufficient to achieve the optimal key-rate performance. We leave a detailed evaluation for future research.


\section{Discussion}
Based on a re-examination of the conventional two-mode MDI-QKD schemes and the recently proposed one-mode MDI-QKD schemes, we have developed a mode-pairing MDI-QKD scheme that retains the advantages of both, namely, achieving a high key rate with easy implementation. Since MDI-QKD schemes have the highest practical security level among the currently feasible QKD schemes, we expect the mode-pairing scheme paves the way for an optimal design for QKD, simultaneously enjoying high practicality, implementation security, and performance.

There remain several interesting directions for future work. Natural follow-up questions lie in the statistical analysis of the mode-pairing scheme in the finite-data-size regime and efficient parameter estimation. Due to the photon-number-based property of the mode-pairing scheme, previous studies of the statistical analysis of two-mode MDI-QKD schemes~\cite{ma2012statistical,curty2014finite,xi2014protocol} can be readily extended to analyse the mode-pairing scheme. To improve the efficiency of data usage, Alice and Bob may perform parameter estimation before basis sifting in order to use all signals that were originally discarded. On the other hand, one could design a mode-pairing scheme using the $X$-basis for key generation and the $Z$-basis for parameter estimation.

In this work, we employ a simple mode-pairing strategy based on pairing adjacent detection pulses. A more sophisticated pairing method might make bit and basis sifting more efficient. To improve the pairing strategy, Alice and Bob could reveal parts of the encoded intensity and phase information. For example, in the simple pairing strategy introduced as Algorithm~\ref{alg:pairing}, Alice and Bob reveal the bases of the generated data pairs immediately after locations $i$ and $j$ are paired. If their basis choices differ, Alice and Bob `unpair' locations $i$ and $j$, and seek the next good pairing location for location $i$ until the basis choices match.

To further enhance the performance, we could extend the mode-pairing design to other optical degrees of freedom, such as angular momentum and spectrum mode. Meanwhile, we could multiplex the usage of different degrees of freedom to enhance the repetition rate and extend the pairing interval $l$. Such multiplexing techniques would have additional benefits for the mode-pairing scheme. Suppose that we multiplex $m$ quantum channels for a QKD task. In a normal setting, the key generation speed would be improved by a factor of $m$. For the mode-pairing scheme, in addition to this $m$-fold improvement, multiplexing would also introduce a larger pairing interval $ml$, since Alice and Bob would be able to pair quantum signals from different channels. A larger pairing interval $ml$ would result in more paired signals and, hence, more key bits. Especially in the high-channel-loss regime where the distance between two clicked signals is large, the number of successful pairs becomes proportional to the maximum pairing interval $ml$. Thus, the key generation rate is proportional to $m^2$ in the high-channel-loss regime.

Meanwhile, entanglement-based MDI-QKD schemes are essentially based on entanglement-swapping, which is the core design feature of quantum repeaters. The mode-pairing technique may help design a robust quantum repeater against a lossy channel. Note that our work shares similarities with the memory-assisted MDI-QKD protocol~\cite{panayi2014memory} with quantum memories in the middle and with the all-photonic intercity MDI-QKD protocol~\cite{azuma2015intercity} with adaptive Bell-state measurement on the postselected photons. It is interesting to discuss the possibility of combining the mode-pairing design with an adaptive Bell-state measurement to tolerate more losses.

Moreover, the mode-pairing scheme has a unique feature in that the key bits are determined not in the encoding or measurement steps but upon postprocessing, which is an approach can be further explored in other quantum communication tasks, including continuous-variable schemes.

\section{Methods}

\subsection{Source replacement of the encoding state} \label{method:sourceReplace}

The main idea of the security proof for the mode-pairing scheme is to introduce an entanglement-based scheme and reduce the security of the scheme to that of a traditional two-mode MDI-QKD scheme. To realise this, we perform a systematic source-replacement procedure~\cite{Scarani2009security,Ferenczi2012symmetries}. Without loss of generality, in this subsection, we always assume the paired locations $(i,j)$ to be $(1,2)$ to simplify the notations.

For convenience in the security proof, we slightly modify the scheme described in Box~\ref{box:MPprotocol}. First, we assume that the random phase of each mode is discretely chosen from a set of $D$ phases, evenly distributed in $[0,2\pi)$. We expect the corresponding correction term in the security analysis due to the discretisation effect to be negligible~\cite{Cao2015discrete,Zeng2019Symmetryprotected}. Second, in the security proof, we modify the phase encoding and postprocessing procedures, as shown in Table~\ref{tab:phaseSifting}. In the original scheme, Alice modulates $A_1$ and $A_2$ based on two random phases $\phi^a_1$ and $\phi^a_2$, respectively. During the $X$-basis processing, she calculates the relative phase difference $\phi^a_\delta:= \phi^a_2 - \phi^a_1$ and splits it into an alignment angle $\theta^a$ in the range of $[0,\pi)$ and a raw key bit $\kappa^a$. We modify these procedures as follows: in addition to the two random phases $\phi^a_1$ and $\phi^a_2$, Alice also generates two bits $z_1''$ and $z_2''$ and applies extra phase modulations of $z_1''\pi$ and $z_2''\pi$ to $A_1$ and $A_2$, respectively. During the $X$-basis processing, she calculates the relative phase difference $\phi^a_\delta:= \phi^a_2 - \phi^a_1$ and directly announces it for alignment-angle sifting.
In the Appendix~\ref{ssc:PrepareandMeasure}, we prove the equivalence of these two encoding methods.

\begin{table}[htbp]
\centering
\begin{tabular}{c|c|c|c}
\hline
 & Modulated phase & $X$-basis postprocessing & Sifting condition \\
 \hline
Original scheme & \tabincell{c}{$A_1: \phi^a_1$, \\ $A_2: \phi^a_2$}  & \tabincell{c}{$\theta^a = (\phi^a_2 - \phi^a_1) \text{ mod } \pi,$ \\ $\kappa^a = \lfloor (\phi^a_2 - \phi^a_1)/\pi \text{ mod } 2 \rfloor$} & $ \theta^a = \theta^b $ \\
 \hline
Modified scheme & \tabincell{c}{$A_1: \phi^a_1 + z_1'' \pi$, \\ $A_2: \phi^a_2 + z_2'' \pi$} & \tabincell{c}{$\theta^a = \phi^a_2 - \phi^a_1,$ \\ $\kappa^a =  z_1''\oplus z_2''$} & $ \theta^a - \theta^b = 0 \text{ or } \pi $ \\
\hline
\end{tabular}
\caption{Comparison of the phase encoding and postprocessing procedures of the mode-pairing scheme presented in the main text and the modified scheme considered in the security proof. In the modified scheme, Alice introduces an extra $\pi$-phase modulation for the storage of a bit $z_1''$. This helps to decouple the phase randomisation and phase encoding analysis.} \label{tab:phaseSifting}
\end{table}

With the modification above, Alice further generates a random bit $z_1'$ and a random dit ($d=D$) $j_1$ in the first round. Based on the values of $z_1'$, $z_1''$ and $j^a_1$, she prepares the state
\begin{equation}
\ket{\psi^{Com}} = \ket{\sqrt{z_1'\mu}e^{\i (\pi z_1'' + \phi_1^a)}},
\end{equation}
with $\phi_1 = j_1 \frac{2\pi}{D}$. As shown in Fig.~\ref{fig:sourceRep}, we substitute the encoding of random encoded information into the introduction of extra ancillary qubit and qudit systems labelled as $\tilde{A}_1$, $A_1''$ and $A_1'$. The purified encoding state is
\begin{equation} \label{eq:PsiComPurified}
\begin{aligned}
\ket{\tilde{\Psi}^{Com}}_{\tilde{A}_1, A_1', A_1'',A_1} = \frac{1}{2\sqrt{D}} \sum_{j_1=0}^{D-1} \ket{j_1}_{\tilde{A}_1} \left( \ket{00}\ket{0} + \ket{01}\ket{0} + \ket{10}\ket{\sqrt{\mu}e^{i\phi^a_1}} + \ket{11}\ket{\sqrt{\mu}e^{i(\phi^a_1+\pi)}} \right)_{A_1',A_1'';A_1}.
\end{aligned}
\end{equation}
In Fig.~\ref{fig:sourceRep}, we provide a specific state preparation procedure. The initial state is
\begin{equation}
\begin{aligned}
\ket{+_D} &:= \frac{1}{\sqrt{D}} \sum_{j=0}^{D-1} \ket{j}, \\
\ket{+} &:= \ket{+_2}.
\end{aligned}
\end{equation}
Here, Alice applies a controlled-phase gate $C_D\text{-}\hat{U}(\phi_\Delta)$ with $\phi_\Delta:=\frac{2\pi}{D}$ from the qudit $\tilde{A}_1$ to optical mode $A_1$. The controlled-phase gate is defined as
\begin{equation}
C_D\text{-}\hat{U}(\phi)_{\tilde{A}A} := \sum_{j=0}^{D-1} \ket{j}_{\tilde{A}}\bra{j}\otimes e^{\i \phi j a^\dag a},
\end{equation}
where $a^\dag$ and $a$ are the creation and annihilation operators, respectively, of mode $A_1$. Alice also applies a controlled-phase gate $C\text{-}\hat{U}(\pi)$ from $A_1''$ to $A_1$.

\begin{figure}[htbp!]
\centering \includegraphics[width=12cm]{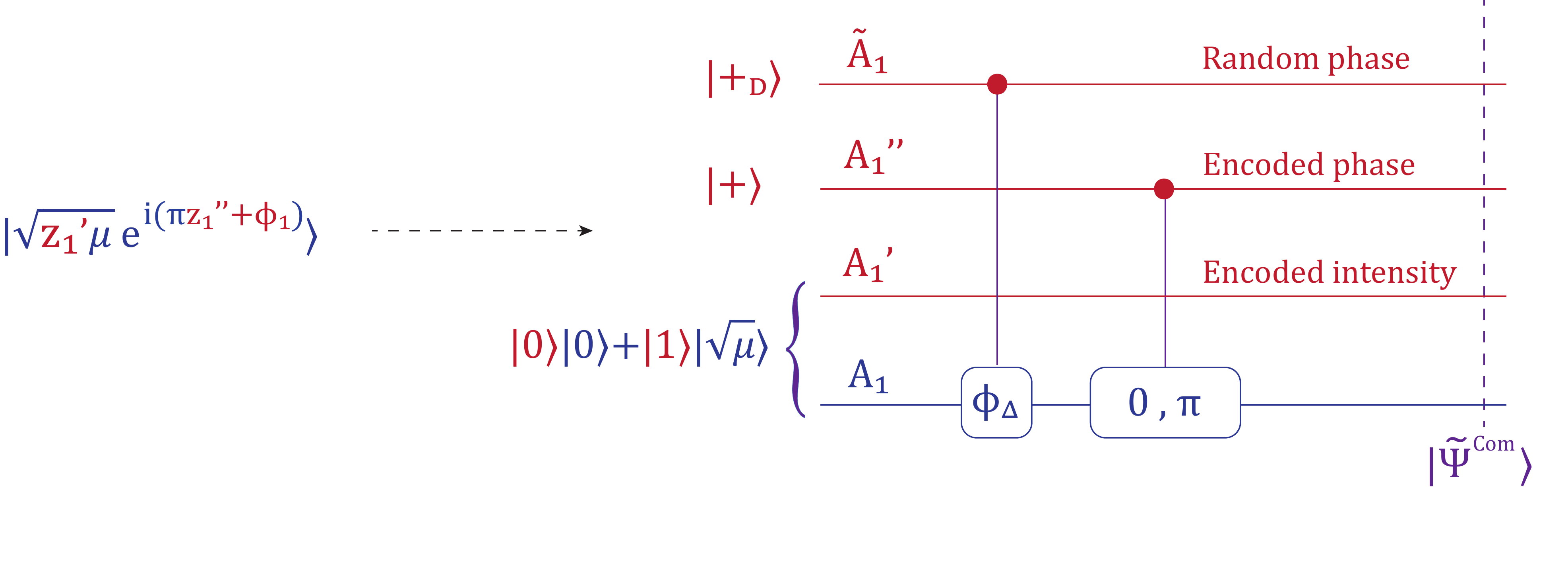}
\caption{Source-replacement procedure for the mode-pairing scheme. We substitute the encoding of all random encoded information into the introduction of purified ancillary systems.} \label{fig:sourceRep}
\end{figure}

In the entanglement-based mode-pairing scheme, Alice and Bob generate the composite encoding state $\ket{\tilde{\Psi}^{Com}}$ defined in Eq.~\eqref{eq:PsiComPurified} in each round. They emit the optical modes to Charlie for interference. Based on Charlie's announcement, they pair the locations and perform global operations on the corresponding ancillaries to generate raw key bits and useful parameters. In Fig.~\ref{fig:systemUsage}, we list the global operations performed on Alice's paired locations. Among them, the relative encoded intensity $\tau^a:=z_1'\oplus z_2'$ is used to determine the basis choice. The encoded intensity $\lambda^a:=z_1'$ and the relative encoded phase $\sigma^a=z_1''\oplus z_2''$ are the raw key bits in the $Z$-basis and $X$-basis postprocessing, respectively.

\begin{figure}[htbp!]
\centering \includegraphics[width=16cm]{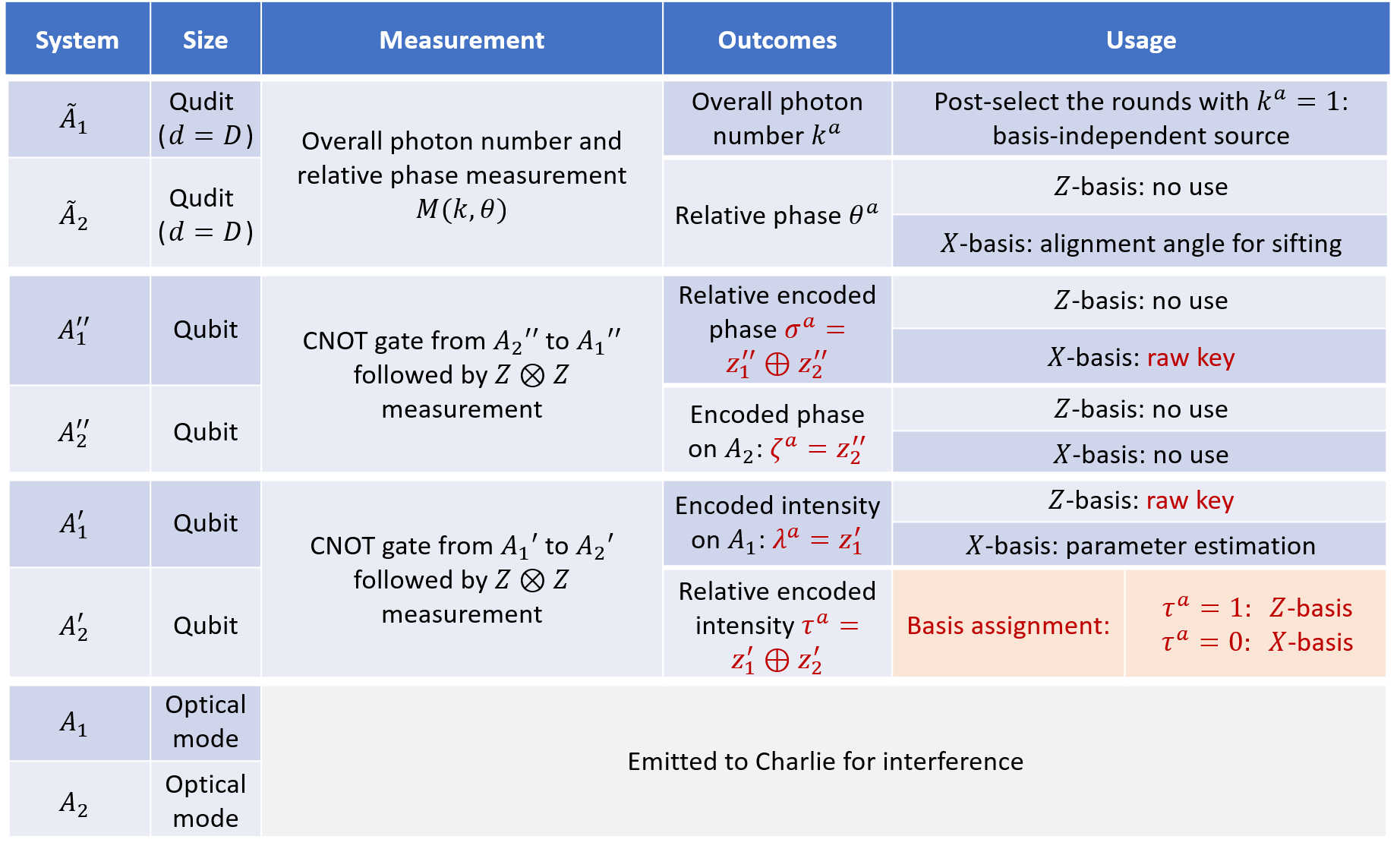}
\caption{The quantum operations and usage of Alice's encoding states on two paired locations $(1,2)$. There are 8 systems based on Alice's two paired locations. Among them, the two qudits $\tilde{A}_1$ and $\tilde{A}_2$ are measured to obtain the overall photon number $k^a$ and the relative phase $\theta^a$ of two optical modes $A_1$ and $A_2$. The two qubits $A_1''$ and $A_2''$ are measured to obtain the relative phase, which is the raw key bit in the $X$-basis. Another two qubits $A_1'$ and $A_2'$ are measured to obtain the encoded intensity in $A_1$ and the relative encoded intensity, which are used for the key mapping on the $Z$-basis and the basis assignment, respectively. } \label{fig:systemUsage}
\end{figure}

A key point in our security proof is that we replace the random phases and register them into purified systems $\tilde{A}_1$ and $\tilde{A}_2$. This enables us to define a global measurement $M(k,\theta)$ on $\tilde{A}_1$ and $\tilde{A}_2$ to \emph{simultaneously} obtain the overall photon number and the relative phase information encoded in optical modes $A_1$ and $A_2$.
The construction of $M(k,\theta)$ is described in Appendix~\ref{Sec:random}.
With the introduction of the purified systems $\tilde{A}_1$ and $\tilde{A}_2$ and the existence of the global measurement $M(k,\theta)$, Alice (same for Bob) is able to determine \emph{at which two locations to perform the global photon number measurement after Charlie's announcement}. With this measurement, Alice and Bob can further reduce the encoding state to a two-mode scheme.
The detailed security proof is provided in Appendix~\ref{Sec:Security}.

\subsection{Mode-pairing scheme with decoy states} \label{method:decoyscheme}
Here, we present the mode-pairing scheme with an extra decoy intensity $\nu$ to estimate the parameters $q_{11}$ and $e^X_{11}$.
Of course, more decoy intensities can be applied in a similar manner.


\begin{enumerate}
\item
\textbf{State preparation}: In the $i$-th round ($i=1,2,...,N$), Alice prepares a coherent state $\ket{\sqrt{\mu^a_i}\exp(\i\phi^a_i)}$ in optical mode $A_i$ with an intensity $\mu^a_i$ randomly chosen from $\{0,\nu,\mu\}$ ($0<\nu<\mu<1$) and a phase $\phi^a_i$ uniformly chosen from the set $\{\frac{2\pi}{D}k\}_{k=0}^{D-1}$. She records $\mu^a_i$ and $\phi^a_i$ for later use. Likewise, Bob chooses $\mu^b_i$ and $\phi^b_i$ randomly and prepares $\ket{\sqrt{\mu^b_i}\exp(\i\phi^b_i)}$ in mode $B_i$.

\item
\textbf{Measurement}: (Same as Step~2 in Box~\ref{box:MPprotocol}.) Alice and Bob send modes $A_i$ and $B_i$ to Charlie, who performs the single-photon interference measurement. Charlie announces the clicks of the detectors $L$ and/or $R$.

Alice and Bob repeat the above two steps $N$ times; then, they perform the following data postprocessing procedures:

\item
\textbf{Mode pairing}: (Same as Step~3 in Box~\ref{box:MPprotocol}.) For all rounds with successful detection ($L$ or $R$ clicks), Alice and Bob establish a strategy for grouping two clicked rounds as a pair. A specific pairing strategy is introduced in Section~\ref{ssec:modepairing}.

\item
\textbf{Basis sifting}: Based on the intensities of two grouped rounds, Alice labels the `basis' of the data pair as:
\begin{enumerate}
\item $Z$ if one of the intensities is $0$ and the other is nonzero;
\item $X$ if both of the intensities are the same and nonzero; or
\item `0' if the intensities are $(0,0)$, which will be reserved for decoy estimation of both the $Z$ and $X$ bases; or
\item `discard' when both intensities are nonzero and not equal.
\end{enumerate}
See also the table below for the basis assignment.
\begin{table}[H]
\centering
\begin{tabular}{|l|c c c|}
\hline
\diagbox{$\mu_i$}{$\mu_j$} & 0 & $\nu$ & $\mu$ \\ \hline
0     & `0' & $Z$ & $Z$ \\
$\nu$ & $Z$  &  $X$ & `discard' \\
$\mu$ & $Z$  &  `discard' & $X$ \\\hline
\end{tabular}
\end{table}

Alice and Bob announce the basis ($X$, $Z$, `0', or `discard') and the sum of the intensities $(\mu^a_{i,j}, \mu^b_{i,j})$ for each location pair $i,j$. If the announced bases are the same and no `discard' state occurs, they record the pair basis and maintain the data pairs; if one of the announced bases is `0' and the other one is $X$($Z$), they record the pair basis as $X$($Z$) and keep the data pairs; if both of the announced bases are `0', they record the pair basis as `0' and maintain the data pairs; and otherwise, they discard the data. See the table below for the basis-sifting strategy.
\begin{table}[H]
\centering
\begin{tabular}{|l|c c c|}
\hline
\diagbox{Bob}{Alice} & `0' & $X$ & $Z$ \\ \hline
`0'     & `0' & $X$ & $Z$ \\
$X$ & $X$  &  $X$ & `discard' \\
$Z$ & $Z$  &  `discard' &  $Z$ \\\hline
\end{tabular}
\end{table}

\item
\textbf{Key mapping}: (Same as Step~5 in Box~\ref{box:MPprotocol}.) For each $Z$-pair at locations $i$ and $j$, Alice sets her key to $\kappa^a = 0$ if the intensity of the $i$-th pulse is $\mu^a_i = 0$ and to $\kappa^a = 1$ if $\mu^a_j = 0$. For each $X$-pair at locations $i$ and $j$, the key is extracted from the relative phase $(\phi^a_j - \phi^a_i) = \theta^a + \pi\kappa^a$, where the raw key bit is $\kappa^a = \left\lfloor((\phi^a_j - \phi^a_i)/\pi \text{ mod } 2)\right\rfloor$ and the alignment angle is $\theta^a := (\phi^a_j - \phi^a_i) \text{ mod } \pi$. Similarly, Bob also assigns his raw key bit $\kappa^b$ and determines $\theta^b$.
For the $X$-pairs, Alice and Bob announce the alignment angles $\theta^a$ and $\theta^b$. If $\theta^a = \theta^b$, they keep the data pairs; otherwise, they discard them.
	
\item
\textbf{Parameter estimation}: Alice and Bob estimate the quantum bit error rate $E_{\mu\mu}^Z$ of the raw key data in $Z$-pairs with overall intensities of $(\mu^a_{i,j}, \mu^b_{i,j})=(\mu,\mu)$. They use $Z$-pairs with different intensity settings to estimate the clicked single-photon fraction $q_{11}$ using the decoy-state method, and the $X$-pairs are used to estimate the single-photon phase error rate $e^X_{11}$.
Specially, $q_{11}$ and $e^X_{11}$ are estimated via the decoy-state method introduced in Appendix~\ref{Sec:decoyestimate}.

\item
\textbf{Key distillation}: (Same as Step~7 in Box~\ref{box:MPprotocol}.) Alice and Bob use the $Z$-pairs to generate a key. They perform error correction and privacy amplification in accordance with $q_{11}$, $E_{\mu\mu}^Z$ and $e^{X}_{11}$.
\end{enumerate}

\subsection{Mode-pairing-efficiency calculation} \label{method:pairrate}
We calculate the expected pairing number $r_p(p,l)$ that corresponds to the simple mode-pairing strategy in Algorithm~\ref{alg:pairing}, which is related to the average click probability $p$ during each round, and the maximal pairing interval $l$.

For calculation convenience, we assume that in addition to the front and rear locations $(F_k, R_k)$ of the $k$-th pair, Alice and Bob also record the starting location $S_k$, which indicates the location at which the first successful detection signal occurs during the pairing procedure for the $k$-th pair. If the second successful detection signal $R_k$ is found within the next $l$ locations, then $F_k=S_k$; otherwise, $F_k$ will be larger than $S_k$.
Let $G_k:=S_{k+1} - S_k$ denote a random variable that reflects the location gap between the $k$-th and $(k+1)$-th starting pulses. Then the expected pairing number per pulse is given by
\begin{equation}
r_p = \frac{1}{\mathbb{E}(G_k)}.
\end{equation}
Hence, we need to calculate only the expectation value of $G_k$. First, we split it into two parts,
\begin{equation}
G_k = (R_k - S_k) + (S_{k+1} - R_k) = H_k + G^{(b)}_k,
\end{equation}
where $H_k: = R_k - S_k$ and $G^{(b)}_k := S_{k+1} - R_k$. Hence,
\begin{equation}
\mathbb{E}(G_k) = \mathbb{E}(H_k) + \mathbb{E}(G^{(b)}_k).
\end{equation}

It is easy to show that $G^{(b)}_k$ obeys a geometric distribution,
\begin{equation}
\Pr(G^{(b)}_k = d) = (1-p)^{d-1}p, \quad d=1,2,...
\end{equation}
Then, the expectation value is $\mathbb{E}(G^{(b)}_k) = 1/p$.

The calculation of the pulse interval $H_k$ is more complex. Suppose that we already know the expectation value $\mathbb{E}(H_k)$; now we calculate the expectation value $\mathbb{E}(H_k|d)$ conditioned on the distance between the starting point and the following click. We have
\begin{equation}
\mathbb{E}(H_k|d) =
\begin{cases}
d, & \quad d\leq l, \\
\mathbb{E}(H_k) + d, & \quad d>l.
\end{cases}
\end{equation}
Therefore,
\begin{equation}
\begin{aligned}
\mathbb{E}(H_k) &= \sum_{d=1}^{+\infty} \Pr(d)\mathbb{E}(H_k|d) \\
&= \sum_{d=1}^{l} (1-p)^{d-1}p  d + \sum_{d>l} (1-p)^{d-1}p [\mathbb{E}(H_k) + d] \\
&= \sum_{d=1}^{+\infty} (1-p)^{d-1}p  d + \mathbb{E}(H_k) \sum_{d>l} (1-p)^{d-1}p \\
&= \frac{1}{p} + \mathbb{E}(H_k) (1-p)^l \\
\end{aligned}
\end{equation}
We have
\begin{equation}
\mathbb{E}(H_k) = \frac{1}{p [1 - (1-p)^l]};
\end{equation}
therefore,
\begin{equation}
\begin{aligned}
\mathbb{E}(G_k) &= \frac{1}{p [1 - (1-p)^l]} + \frac{1}{p} ,\\
\Rightarrow r_p &= \left[ \frac{1}{p [1 - (1-p)^l]} + \frac{1}{p} \right]^{-1}.
\end{aligned}
\end{equation}

\section*{Data availability}
The data that support the plots within this paper and other findings of this study are available from the corresponding authors upon reasonable request.

\section*{Acknowledgments}
We thank Yizhi Huang, Guoding Liu, Zhenhuan Liu, Tian Ye, Junjie Chen, Minbo Gao, and Xingjian Zhang for the helpful discussion on the pairing rate calculation and general comments  on the presentation.
We especially thank Norbert L\"{u}tkenhaus for the helpful discussions on the security analysis, thorough proofreading, and beneficial suggestions on the manuscript presentation. 
We especially thank Hao-Tao Zhu and Teng-Yun Chen for providing us some preliminary results showing the phase stabilization after removing phase-locking in the mode-pairing scheme.
This work was supported by the National Natural Science Foundation of China Grants No.~11875173 and No.~12174216 and the National Key Research and Development Program of China Grants No.~2019QY0702 and No.~2017YFA0303903.

\emph{Note added.---}
After we submitted our work for reviewing, we became aware of a relevant work by \textcite{xie2021breaking}, who consider a similar MDI-QKD protocol that match the clicked data to generate key information. Under the assumption that the single-photon distributions in all the Charlie's successful detection events are independent and identically distributed, the authors simulate the performance of the protocol and show its ability to break the repeaterless rate-transmittance bound.


\section*{Competing interests}
The authors declare no competing interests.




%

\onecolumngrid
\newpage

\begin{appendix}


\maketitle

We present a full security proof of the mode-pairing measurement-device-independent quantum key distribution scheme (MP scheme for simplicity), the details of the numerical simulation and a proof-of-principle experimental demonstration to show the feasibility of the MP scheme.
In Sec.~\ref{Sec:random}, we introduce a source-replacement procedure of the random phases and photon numbers in the coherent states as a preliminary.
In Sec.~\ref{Sec:Security}, we provide the security proof of the MP scheme.
In Sec.~\ref{Sec:decoyestimate}, we introduce the decoy-state estimation method for the mode-pairing scheme introduce in Methods ``Mode-pairing scheme with decoy states''.
In Sec.~\ref{Sec:Simu}, we list the simulation formulas of the MP scheme, as well as the two-mode measurement-device-independent quantum key distribution (MDI-QKD), decoy-state BB84, and phase-matching quantum key distribution (PM-QKD) schemes.
In Sec.~\ref{Sec:optimalmu}, we present numerical results related to the optical pulse intensities of different schemes.
In Sec.~\ref{sc:ExpDemo}, we present a proof-of-principle demonstration of the mode-pairing scheme without global phase locking.


\section{Source replacement of the random phase systems} \label{Sec:random}
In the security proof of the coherent-state QKD schemes, we will frequently come across the photon-number measurement on the coherent states with randomised phases. Using the overall photon-number measurement on two optical modes, one can post-select the single-photon component as an ideal encoding qubit.
A unique feature of the MP scheme is that, Alice and Bob do not want to decide which two optical modes to perform the overall photon-number measurement before obtaining the detection result from the untrusted party, Charlie.
To analyze the security of the MP scheme in this scenario, we perform an extra source-replacement procedure by introducing extra ancillary systems to record the random phase or photon number information. In this way, we can realize the overall photon-number measurement of the coherent states \emph{indirectly} based on the measurement on the ancillary systems.

For the convenience of later discussion, here, we introduce the source-replacement procedure of the random phases in detail.
First, we consider the source replacement of the encoded random phase in a single optical mode in Sec.~\ref{ssc:single_random}. With the Fourier transform on the ancillary basis, we then show the complementarity between encoded phases and photon numbers. After that, we describe the source replacement of the encoded random phases in the two-optical-mode case in Sec.~\ref{ssc:two_random}. We will show the compatibility of the relative-phase measurement and the overall photon-number measurement.



\subsection{Single-optical-mode case} \label{ssc:single_random}
Consider a coherent state $\ket{\sqrt{\mu}e^{i\phi}}$ on system $A$ with a random phase $\phi$. For the convenience of later discussion, we assume that the phase $\phi$ is discretely randomised from the set $\{\phi_j= \frac{2\pi}{D}j\}_{j\in[D]}$. Here, $D$ is the number of discrete phases and $[D]:=\{0,1,...,D-1\}$. In the source replacement, we introduce a qudit with $d=D$ to record the random phase information. The joint purified system can be written as
\begin{equation} \label{eq:psiArandom}
\ket{\psi}_{\tilde{A}A} = \frac{1}{\sqrt{D}} \sum_{j=0}^{D-1} \ket{j}_{\tilde{A}} \ket{\sqrt{\mu} e^{i\phi_j}}_A.
\end{equation}
This purified state can be prepared by first preparing a coherent state $\ket{\sqrt{\mu}}_A$ and a maximal coherent qudit state
\begin{equation}
\ket{+_D}_{\tilde{A}} = \frac{1}{\sqrt{D}} \sum_{j=0}^{D-1} \ket{j}_{\tilde{A}},
\end{equation}
and then applying a controlled-phase gate $C_D\text{-}\hat{U}(\phi_\Delta)$ with $\phi_\Delta:=\frac{2\pi}{D}$ from the qudit $\tilde{A}$ to the optical mode $A$. The controlled-phase gate is defined as
\begin{equation}
C_D\text{-}\hat{U}(\phi)_{\tilde{A}A} := \sum_{j=0}^{D-1} \ket{j}_{\tilde{A}}\bra{j}\otimes e^{i \phi j a^\dag a},
\end{equation}
where $a^\dag$ and $a$ are the creation and annihilation operators of the mode $A$, respectively.

\begin{figure}[htbp]
\includegraphics[width=9cm]{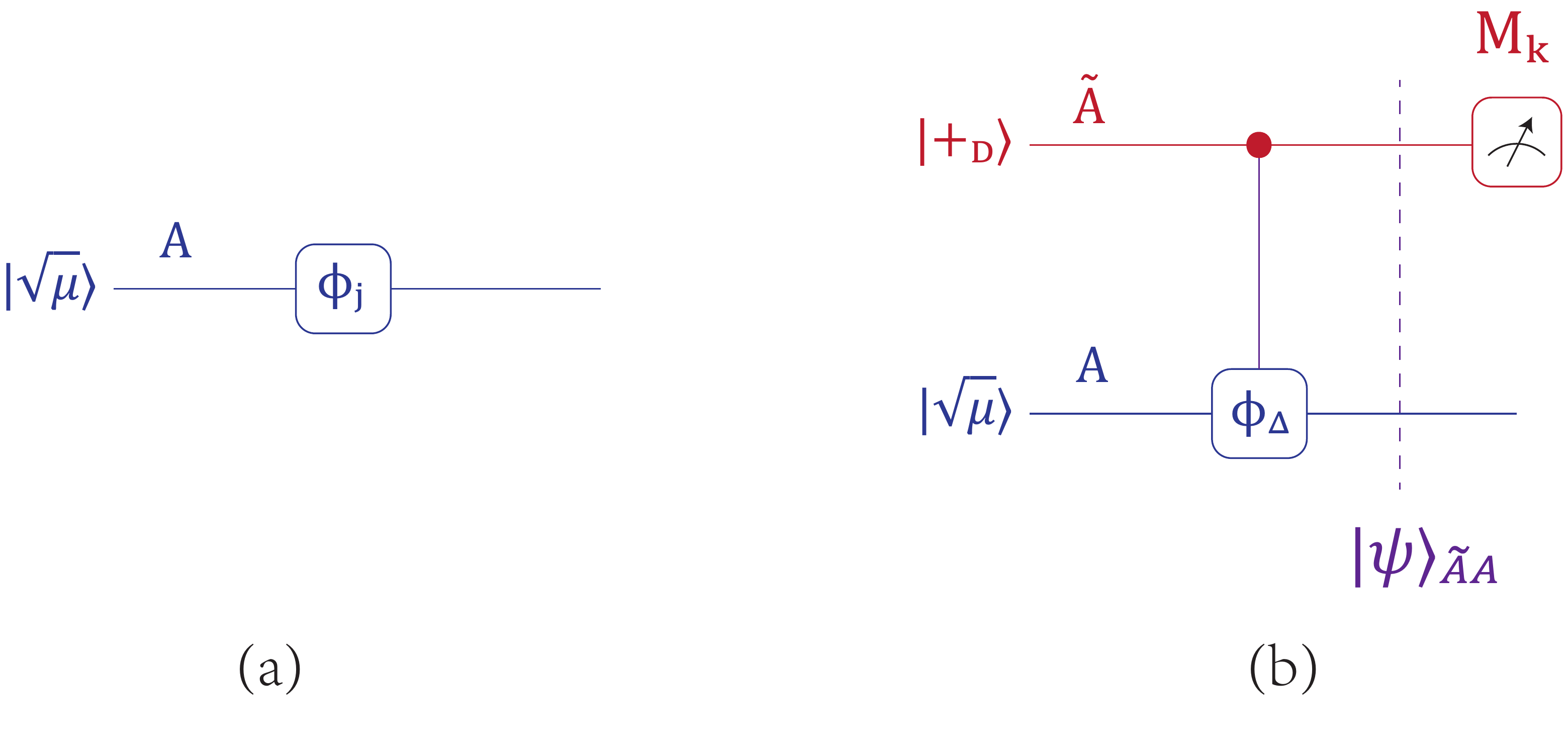}
\caption{(a) A coherent state with a discrete random phase $\phi_j$. (b) A possible way to produce the entangled state generated by the source replacement on the coherent states with random phase encoding. Alice first prepares a maximal coherent qudit state $\ket{+_D}$ on the system $\tilde{A}$, and then performs the control gate $C_D-\hat{U}(\phi_\Delta)$ from $\tilde{A}$ to $A$ to generate the state $\ket{\psi}$. Alice can then choose to measure on $\tilde{A}$ on $\{\ket{j}\}_j$ or $\{\ket{\tilde{k}}\}_k$ basis to read out either the phase or the photon number information of the system $A$. Here, $\phi_\Delta = \frac{2\pi}{D}$.} \label{fig:singlePurified}
\end{figure}

To define the measurement on the ancillary $\tilde{A}$ that informs the photon number of $A$, we introduce a complementary basis $\{\ket{\tilde{k}}_{\tilde{A}}\}_{k=0}^{D-1}$ on the qudit system $\tilde{A}$ via the Fourier transform, $\forall k, j \in [D]$,
\begin{equation} \label{eq:singleFT}
\begin{aligned}
\ket{\tilde{k}}_{\tilde{A}} &:= \frac{1}{\sqrt{D}} \sum_{j=0}^{D-1} e^{i \frac{2\pi}{D}jk } \ket{j}_{\tilde{A}}, \\
\ket{j}_{\tilde{A}} &= \frac{1}{\sqrt{D}} \sum_{k=0}^{D-1} e^{-i \frac{2\pi}{D}jk } \ket{\tilde{k}}_{\tilde{A}}.
\end{aligned}
\end{equation}

In this way, the state in Eq.~\eqref{eq:psiArandom} can be reformulated as
\begin{equation} \label{eq:psiArandom2}
\begin{aligned}
\ket{\psi}_{\tilde{A}A} &= \frac{1}{\sqrt{D}} \sum_{j=0}^{D-1} \ket{j}_{\tilde{A}} \ket{\sqrt{\mu} e^{i\phi_j}}_A \\
&= \frac{1}{D} \sum_{j=0}^{D-1} \sum_{k=0}^{D-1} e^{-i \frac{2\pi}{D}jk } \ket{\tilde{k}}_{\tilde{A}} \ket{\sqrt{\mu} e^{i\phi_j}}_A \\
&= \frac{1}{D} \sum_{k=0}^{D-1} \ket{\tilde{k}}_{\tilde{A}} \left( \sum_{j=0}^{D-1} e^{-i \frac{2\pi}{D}jk } \ket{\sqrt{\mu} e^{i\phi_j}}_A  \right) \\
&= \sum_{k=0}^{D-1} \sqrt{P_k} \ket{\tilde{k}}_{\tilde{A}} \ket{\lambda^{\mu}(k)}_A,
\end{aligned}
\end{equation}
where the normalized pseudo-Fock state is given by,
\begin{equation} \label{eq:pseudoFock}
\begin{aligned}
\ket{\lambda^\mu(k)}_A &:= \frac{1}{D\sqrt{P_k^\mu}} \sum_{j=0}^{D-1} e^{-i \frac{2\pi}{D}jk } \ket{\sqrt{\mu} e^{i\phi_j}}_A \\
&= \frac{1}{\sqrt{P_k^\mu}} e^{-\frac{\mu}{2}} \sum_{m=0}^{\infty} \frac{\sqrt{\mu}^{(mD+k)}}{\sqrt{(mD+k)!}} \ket{mD+k}_A,
\end{aligned}
\end{equation}
and the corresponding probability distribution is,
\begin{equation}  \label{eq:PseudoPoissson}
P_k^\mu := e^{-\mu}\sum_{m=0}^{\infty} \frac{\mu^{(mD+k)}}{(mD+k)!}.
\end{equation}

From Eq.~\eqref{eq:psiArandom2} we can see that, the global phase $\phi_j$ and the (pesudo) photon number $k$ are two complementary observables, which cannot be determined simultaneously. Furthermore, since the normalized pesudo-Fock states $\{\ket{\lambda^\mu(k)}_A\}_k$ are orthogonal with each other, if we measure the ancillary system $\tilde{A}$ on the basis $\{\ket{\tilde{k}}\}$, it is equivalent to perform the projective (pesudo) photon-number measurement on the system $A$.

When $D\to\infty$, Eq.~\eqref{eq:pseudoFock} becomes Fock states and Eq.~\eqref{eq:PseudoPoissson} becomes a Poisson distribution, $P^{\mu}(k)=e^{-\mu}\frac{\mu^k}{k!}$, and then the discrete phase randomisation approaches the continuous phase randomisation. In practice, it suffices to choose $D\geq 12$ to make the discretisation effect ignorable~\cite{Cao2015discrete,Zeng2019Symmetryprotected}.

\subsection{Two-optical-mode case} \label{ssc:two_random}

In the coherent-state MDI-QKD scheme, usually we encode the information into the single-photon subspace on two orthogonal optical modes. To see how this works, we consider the two-optical-mode encoding and show the compatibility between the global photon number and the relative-phase measurement.

Consider two coherent states $\ket{\sqrt{\mu}e^{i\phi_1}}_{A_1}$ and $\ket{\sqrt{\mu}e^{i\phi_2}}_{A_2}$ on two orthogonal optical modes with random phases $\phi_1$ and $\phi_2$. Similarly, we assume the random phases are discretely randomised from the sets $\{\phi_{j_1} =\frac{2\pi}{D}j_1\}_{j_1=0}^{D-1}$ and $\{\phi_{j_2} =\frac{2\pi}{D}j_2\}_{j_2=0}^{D-1}$, respectively. The joint purified system can be written as,
\begin{equation} \label{eq:psiA12random}
\begin{aligned}
\ket{\psi_{1,2}}_{\tilde{A}_1\tilde{A}_2,A_1 A_2} &= \frac{1}{D} \sum_{j_1=0}^{D-1} \sum_{j_2=0}^{D-1} \ket{j_1}_{\tilde{A}_1} \ket{j_2}_{\tilde{A}_2} \ket{\sqrt{\mu} e^{i\phi_{j_1}}}_{A_1}\ket{\sqrt{\mu} e^{i\phi_{j_2}}}_{A_2} \\
&= \frac{1}{D} \sum_{j_2=0}^{D-1} \sum_{j_\theta=0}^{D-1} \ket{j_2+j_\theta}_{\tilde{A}_1} \ket{j_2}_{\tilde{A}_2} \ket{\sqrt{\mu} e^{i\phi_{(j_2+j_\theta)}}}_{A_1}\ket{\sqrt{\mu} e^{i\phi_{j_2}}}_{A_2}, \\
\end{aligned}
\end{equation}
where we substitute the variable $j_1:= j_2 + j_\theta$. The addition $+$ for all the discrete indices ($j_1$, $j_2$ and $j_\theta$) is defined on the ring $\mbb{Z}_D$ (taking modulus of $D$).

\begin{figure}[htbp]
\includegraphics[width=14cm]{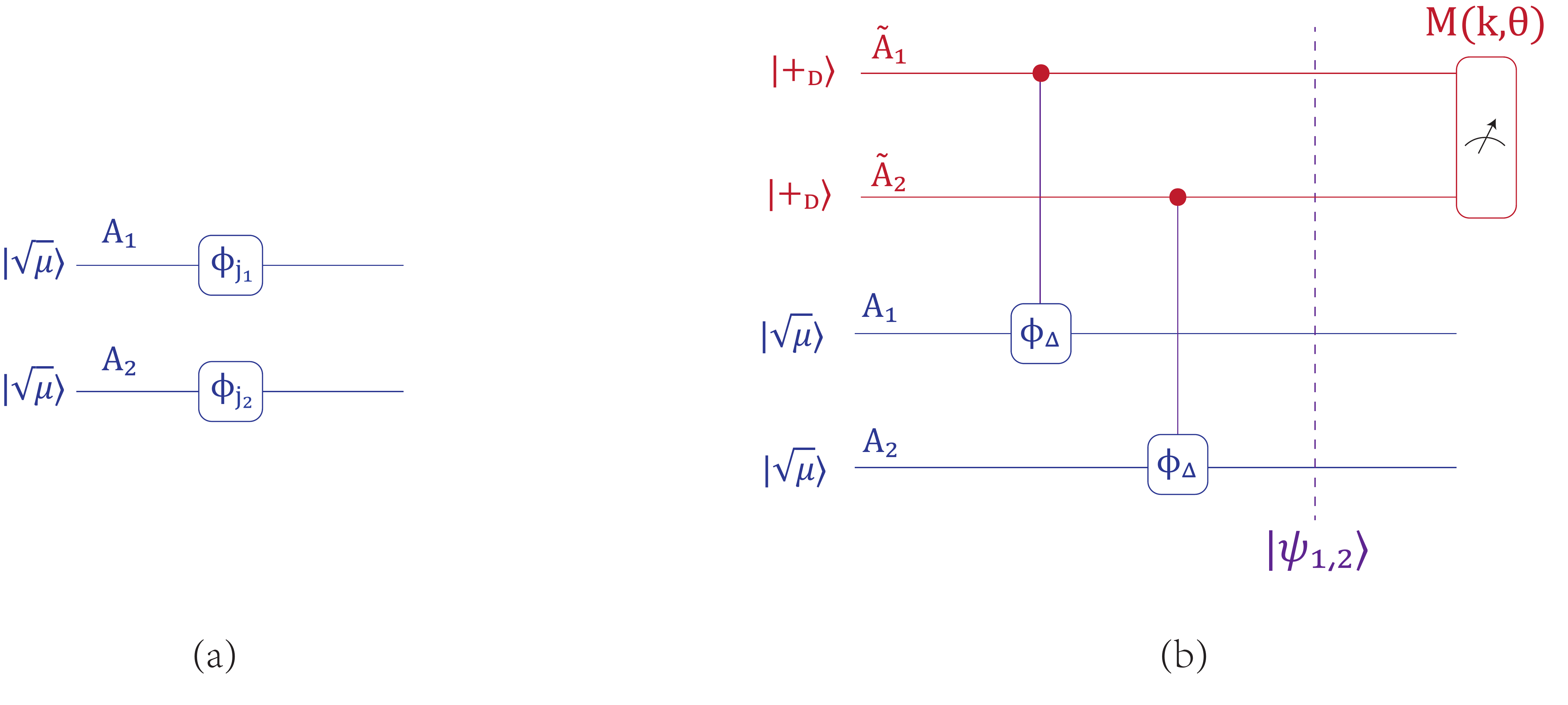}
\caption{(a) Two coherent states with independent random phases $\phi_{j_1}$ and $\phi_{j_2}$. (b) A possible way to produce the entangled state generated by the source replacement on the coherent states with random phase encoding. Alice first prepares two maximal coherent qudit state $\ket{+_D}$ on the systems $\tilde{A}_1$ and $\tilde{A}_2$, then performs the control gates $C_D-\hat{U}(\phi_\Delta)$ from $\tilde{A}_1$ and $\tilde{A}_2$ to $A_1$ and $A_2$, respectively, to generate the state $\ket{\psi_{1,2}}$. When Alice performs the joint measurement $M(k,\theta)$ defined in Eq.~\eqref{eq:Mktheta} on $\tilde{A}_1$ and $\tilde{A}_2$, she can obtain the global photon number $k$ as well as the relative phase information $\theta$ on systems $A_1$ and $A_2$.} \label{fig:twoPurified}
\end{figure}

For each given $\theta$, the basis vectors $\{ \ket{j_2+j_\theta}_{\tilde{A}_1} \ket{j_2}_{\tilde{A}_2} \}_{j_2=0}^{D-1}$ defines a subspace on the joint Hilbert space of two qudits $\tilde{A}_1$ and $\tilde{A}_2$. Similar to Eq.~\eqref{eq:singleFT}, we introduce a partial Fourier transform on each subspace, $\forall k, j_2\in [D]$,
\begin{equation} \label{eq:twoFT}
\begin{aligned}
\ket{\tilde{k},j_\theta}_{\tilde{A}_1,\tilde{A}_2} &:= \frac{1}{\sqrt{D}} \sum_{j_2=0}^{D-1} e^{i \frac{2\pi}{D}j_2 k} \ket{j_2+j_\theta}_{\tilde{A}_1} \ket{j_2}_{\tilde{A}_2}, \\
\ket{j_2+j_\theta}_{\tilde{A}_1} \ket{j_2}_{\tilde{A}_2} &= \frac{1}{\sqrt{D}} \sum_{k=0}^{D-1} e^{-i \frac{2\pi}{D}j_2 k } \ket{\tilde{k},j_\theta}_{\tilde{A}_1\tilde{A}_2}.
\end{aligned}
\end{equation}

In this way, the state in Eq.~\eqref{eq:psiA12random} can be reformulated as
\begin{equation} \label{eq:psiA12random2}
\begin{aligned}
\ket{\psi_{1,2}}_{\tilde{A}_1\tilde{A}_2,A_1 A_2} &= \frac{1}{D} \sum_{j_2=0}^{D-1} \sum_{j_\theta=0}^{D-1} \ket{j_2+j_\theta}_{\tilde{A}_1} \ket{j_2}_{\tilde{A}_2} \ket{\sqrt{\mu} e^{i\phi_{(j_2+j_\theta)}}}_{A_1}\ket{\sqrt{\mu} e^{i\phi_{j_2}}}_{A_2} \\
&= \frac{1}{D\sqrt{D}} \sum_{j_2=0}^{D-1} \sum_{j_\theta=0}^{D-1} \sum_{k=0}^{D-1} e^{-i \frac{2\pi}{D}j_2 k } \ket{\tilde{k},j_\theta}_{\tilde{A}_1,\tilde{A}_2} \ket{\sqrt{\mu} e^{i\phi_{(j_2+j_\theta)}}}_{A_1}\ket{\sqrt{\mu} e^{i\phi_{j_2}}}_{A_2} \\
&= \frac{1}{D\sqrt{D}} \sum_{j_\theta=0}^{D-1} \sum_{k=0}^{D-1} \ket{\tilde{k},j_\theta}_{\tilde{A}_1\tilde{A}_2} \left(\sum_{j_2=0}^{D-1} e^{-i \frac{2\pi}{D}j_2 k }  \ket{\sqrt{\mu} e^{i\phi_{(j_2+j_\theta)}}}_{A_1}\ket{\sqrt{\mu} e^{i\phi_{j_2}}}_{A_2} \right) \\
&= \frac{1}{D\sqrt{D}} \sum_{j_\theta=0}^{D-1} \sum_{k=0}^{D-1} \ket{\tilde{k},j_\theta}_{\tilde{A}_1\tilde{A}_2}\; \hat{U}_{A_1}(\frac{2\pi}{D}j_\theta) \circ \hat{BS} \left(\sum_{j_2=0}^{D-1} e^{-i \frac{2\pi}{D}j_2 k }  \ket{\sqrt{2\mu} e^{i\phi_{j_2}}}_{A_1} \ket{0}_{A_2} \right) \\
&= \sum_{j_\theta=0}^{D-1} \frac{1}{\sqrt{D}} \sum_{k=0}^{D-1} \sqrt{P_k^{2\mu}} \ket{\tilde{k},j_\theta}_{\tilde{A}_1\tilde{A}_2}\; \ket{\lambda_{1,2}^{2\mu}(k),j_\theta}_{A_1A_2},
\end{aligned}
\end{equation}
where in the second equality, we use the partial Fourier transform in Eq.~\eqref{eq:twoFT}. In the fourth equality, we simplify the expression by a phase-gate $\hat{U}_{A_1}(\phi) := e^{i\phi a^\dag_1 a_1}$ and the 50:50 beam-splitter $\hat{BS}$ with the following transformation relationship,
\begin{equation}
\begin{pmatrix}
a_1' \\
a_2'
\end{pmatrix}
=
\dfrac{1}{\sqrt{2}}
\begin{pmatrix}
1 & 1 \\
1 & -1
\end{pmatrix}
\begin{pmatrix}
a_1 \\
a_2
\end{pmatrix},
\end{equation}
where $a_1$ and $a_2$ denote the annihilation operators of the two input modes, while $a_1'$ and $a_2'$ denote the annihilation operators of the two output modes. In the fifth equality, we define
\begin{equation} \label{eq:lambda_twomode}
\begin{aligned}
\ket{\lambda^{2\mu}_{1,2}(k),j_\theta}_{A_1A_2} :&= \frac{1}{D\sqrt{P_k^{2\mu}}} \hat{U}_{A_1}(\frac{2\pi}{D}j_\theta) \circ \hat{BS} \left(\sum_{j_2=0}^{D-1} e^{-i \frac{2\pi}{D}j_2 k }  \ket{\sqrt{2\mu} e^{i\phi_{j_2}}}_{A_1} \ket{0}_{A_2} \right) \\
&= \hat{U}_{A_1}(\frac{2\pi}{D}j_\theta) \circ \hat{BS} \left( \ket{\lambda^{2\mu}(k)}_{A_1} \ket{0}_{A_2} \right) \\
&= \hat{U}_{A_1}(\frac{2\pi}{D}j_\theta) \circ \hat{BS} \left( \frac{e^{-\mu}}{\sqrt{P_k^{2\mu}}}\sum_{m=0}^{\infty} \frac{\sqrt{2\mu}^{(mD+k)}}{\sqrt{(mD+k)!}}  \ket{mD+k}_{A_1} \ket{0}_{A_2} \right),
\end{aligned}
\end{equation}
to be a normalized state with pseudo-Fock number $k$ and relative phase $\theta=\frac{2\pi}{D}j_\theta$.

From Eq.~\eqref{eq:psiA12random2} we can see that, the overall (pseudo) photon number $k$ and the (discrete) relative phase $\theta$ are two \emph{compatible} observables for two coherent states with random phases. Denote the projective measurement on the basis $\{\ket{\tilde{k},j_\theta}_{\tilde{A}_1,\tilde{A}_2}\}$ to be a global photon-number and encoded relative-phase measurement $M(k,\theta)$,
\begin{equation} \label{eq:Mktheta}
M(k,\theta):= \left\{ \Pi_{k,j_\theta}  := \ket{\tilde{k},j_\theta}_{\tilde{A}_1,\tilde{A}_2} \bra{\tilde{k},j_\theta} \right\},
\end{equation}
which will be frequently used in the following security proof.


It is easy to check that, the conditional states $\{\ket{\lambda^{2\mu}_{1,2}(k),j_\theta}\}_{k,j_\theta}$ with different photon numbers $k$ are orthogonal with each other. We now consider the following two encoding procedures,
\begin{enumerate}
\item Alice first prepares an entangled state $\ket{\psi_{1,2}}_{\tilde{A}_1\tilde{A}_2,A_1 A_2}$. She them performs the measurement $M(k,\theta)$ defined in Eq.~\eqref{eq:Mktheta} on $\tilde{A}_1$ and $\tilde{A}_2$ to determine the global photon number $k$ and the encoded relative phase $\theta = j_\theta \frac{2\pi}{D}$ on the emitted optical modes $A_1$ and $A_2$;
\item
Alice generates two coherent states on $A_1$ and $A_2$ with random phases $\phi_{j_1}$ and $\phi_{j_2}$, respectively. She records the relative phase $\theta = j_\theta \frac{2\pi}{D}$ and then performs direct global photon-number measurement on the two coherent states to determine the global photon number $k$.
\end{enumerate}
In Ref.~\cite{Cao2015discrete}, it has been shown that the resultant states in the second procedure are also $\{\ket{\lambda^{2\mu}_{1,2}(k),j_\theta}\}_{k,j_\theta}$. As a result, these two encoding procedures are equivalent.

In the extreme case when $D\to \infty$, the state $\ket{\lambda^{2\mu}_{1,2}(k),j_\theta}_{A_1A_2}$ becomes
\begin{equation} \label{eq:lambda12kDInf}
\ket{\lambda_{1,2}(k),j_\theta} = \hat{U}_{A_1}(\frac{2\pi}{D}j_\theta) \circ \hat{BS} \left( \ket{k}_{A_1} \ket{0}_{A_2} \right),
\end{equation}
which is independent of the intensity $\mu$. Especially, the state with global photon number $k=1$ and the relative phase $\phi_{j_\theta}$ becomes
\begin{equation}
\ket{\lambda_{1,2}(1),j_\theta} = \frac{1}{\sqrt{2}}(\ket{01} + e^{i\frac{2\pi}{D}j_\theta} \ket{10} )_{A_1A_2},
\end{equation}
which forms a qubit subspace widely used in QKD.

\section{Security of mode-pairing scheme} \label{Sec:Security}
In this section, we prove the security of the MP scheme via proving the security of its equivalent entanglement version. The major tool we use in the security proof is the QKD equivalence argument.

For a generic MDI-QKD setting, Alice and Bob can view Charlie's site as a joint measurement, $M_c$, on emitted optical pulses. Here, measurement $M_c$ contains all Charlie's operations, including measurement on emitted optical pulses, data processing, and the announcement strategy. The measurement-device-independent property makes measurement $M_c$ a black box to Alice and Bob. Based on the measurement result, $\vec{C}$, Alice and Bob perform a key-generation operation, $M_{ab}$, to their ancillary states to get the final private states, where the key-generation operation $M_{ab}$ includes state measurement, key mapping, parameter estimation, and data post-processing. The whole procedure is depicted in Fig.~\ref{fig:MDIMcMab}.

\begin{figure}[htbp]
\includegraphics[width=12cm]{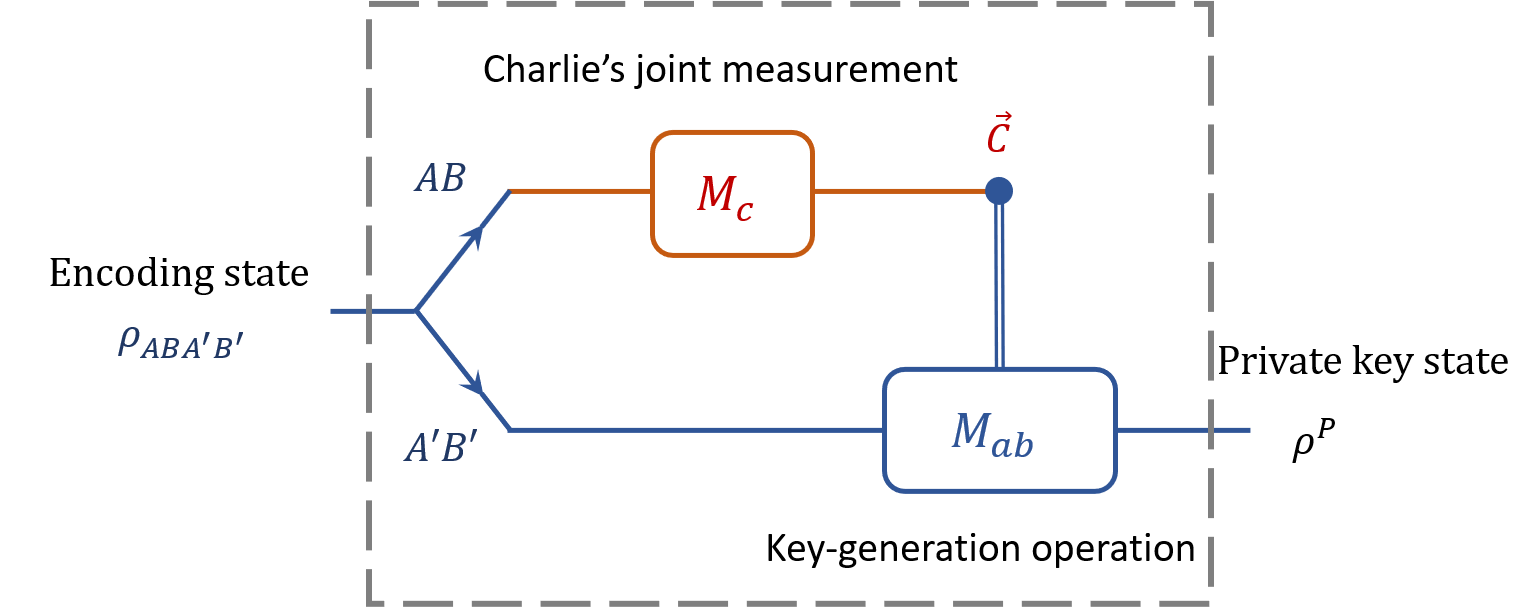}
\caption{Generic MDI-QKD. The red part owned by Charlie is uncharacterized and untrusted. The blue part is well characterized in Alice and Bob's hands. In general, Charlie's joint measurement outcome, $\vec{C}$, controls Alice and Bob's key-generation operation, $M_{ab}$. This (classical) control operation is fixed by the protocol. That is, given a measurement result, $\vec{C}$, Alice and Bob know exactly how to operate on their ancillary states to extract the final key. We can treat the whole circuit in the gray dashed box as a gigantic operation on the initial state, $\rho_{ABA'B'}$, to output a private key state, $\rho^P$.} \label{fig:MDIMcMab}
\end{figure}

\begin{lemma}[Equivalent MDI-QKD]\label{lem:equivMDI}
Two MDI-QKD schemes, in the form of Fig.~\ref{fig:MDIMcMab}, generate the same private state given the same attack and hence are equivalent in security, if the following items are the same,
\begin{enumerate}
\item
Alice and Bob's initial states, including emitted quantum states of system $AB$ and ancillary states of system $A'B'$;
\item
the (classical) control operation for key generation, $\vec{C}$-$M_{ab}$.
\end{enumerate}
\end{lemma}


Here, item~2 implies the dependence of Alice and Bob's key-generation operation $M_{ab}$ shown in Fig.~\ref{fig:MDIMcMab} are the same in the two schemes if Charlie's announcements $\vec{C}$ are the same.

\begin{proof}
Since Alice and Bob's initial states in both schemes are the same, the spaces of possible operations $M_c$ Charlie can perform on system $AB$ are the same. Then, we only need to compare the resultant private key states, $\rho^P$, when Charlie performs the same operations $M_c$ on the two schemes.

In Fig.~\ref{fig:MDIMcMab}, we can treat the whole circuit as one gigantic operation, as shown in the gray dashed box. If Charlie's operations $M_c$ are the same in the two schemes, the final private states are the same, because they both come from the same operation on the same state. In MDI-QKD, an eavesdropper's attack is reflected in $M_c$. Given any attack, the two schemes would render the same private key states. Therefore, the security of the two schemes are equivalent in security.
\end{proof}

\begin{corollary}[Equivalent MDI-QKD under post-selection] \label{coro:equivMDIpost}
Two MDI-QKD schemes, in the form of Fig.~\ref{fig:MDIMcMab}, are equivalent in security, if the following items are the same,
\begin{enumerate}
\item
Alice and Bob's initial states, including emitted quantum states of system $AB$ and ancillary states of system $A'B'$, after a post-selection procedure that is independent of Charlie's announcement $\vec{C}$;
\item
the (classical) control operation for key generation $\vec{C}$-$M_{ab}$.
\end{enumerate}
\end{corollary}

The Corollary~\ref{coro:equivMDIpost} is a direct result of Lemma~\ref{lem:equivMDI}. Based on Corollary~\ref{coro:equivMDIpost}, we can introduce extra encoding redundancy in MDI-QKD. If the encoding state of the new MDI-QKD scheme is the same as the original one with proper post-selection independent of Charlie's announcement $\vec{C}$, then the security of the new scheme is equivalent to the original one.

The main procedure to prove the security of the mode-pairing scheme is to reduce the entanglement-based MP scheme to the traditional single-photon MDI-QKD scheme through a few source-replacement steps, as sketched out in Fig.~\ref{fig:sketch}. For the security proof, we work our way backward from the final stage to the first one. For the completeness of the analysis, we will also review the well-established security proof of the single-photon two-mode MDI-QKD and the coherent-state two-mode MDI-QKD schemes with the source-replacement language.


\begin{figure}[htbp]
\includegraphics[width=16cm]{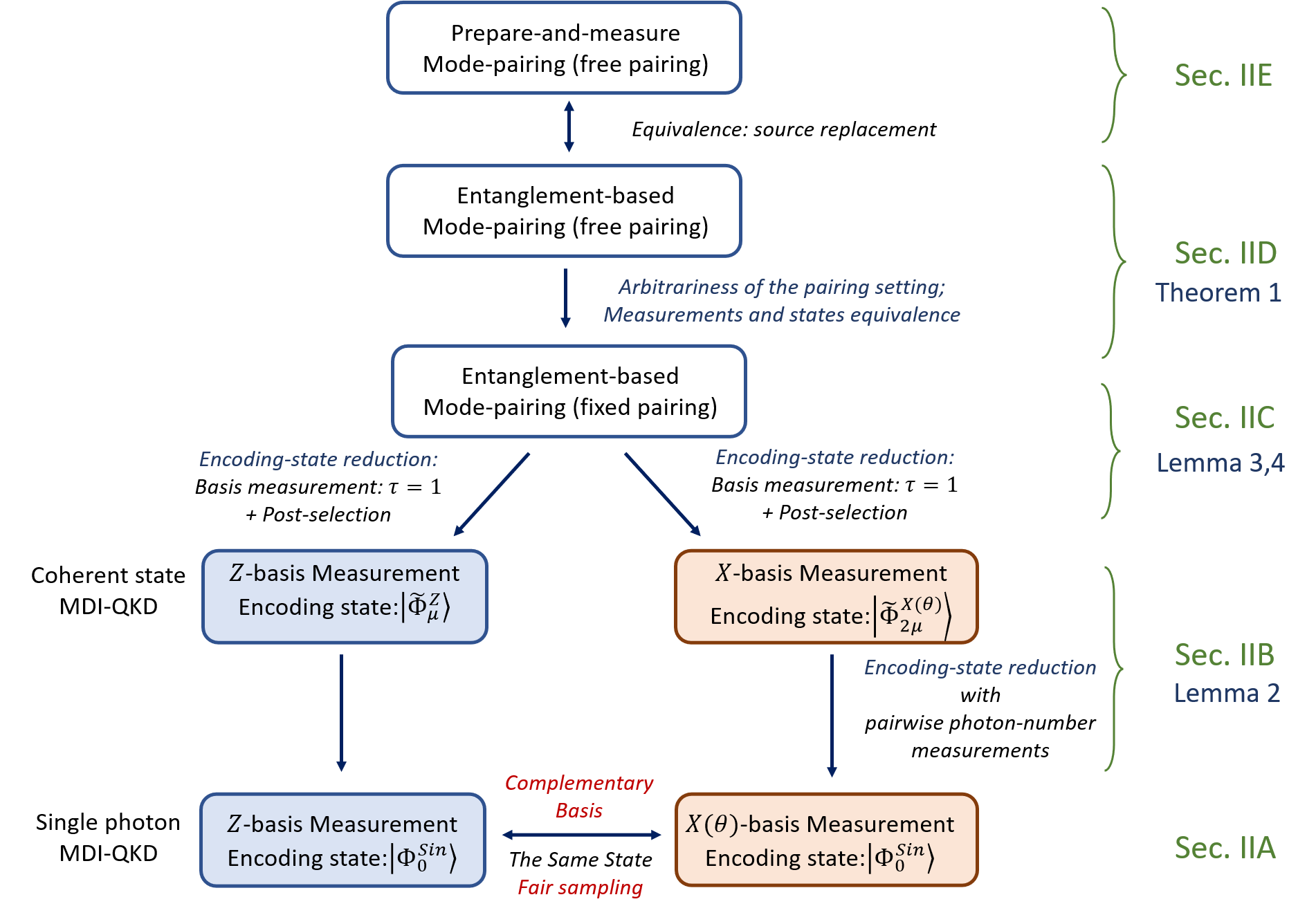}
\caption{Sketch of the security proof of the MP scheme. The main idea is to introduce source replacement and then reduce the security of the MP scheme to that of single-photon MDI-QKD, where the $X$-basis error rate can be used to fairly estimate the $Z$-basis phase-error rate.} \label{fig:sketch}
\end{figure}

In Sec.~\ref{ssc:secureTMMDI}, we review an single-photon two-mode MDI-QKD scheme whose security can be easily verified by the Lo-Chau security proof based on entanglement distillation \cite{lo1999Unconditional}.

In Sec.~\ref{ssc:secureTMMDIcoherent}, we review coherent-state two-mode MDI-QKD. Based on the methods in Sec.~\ref{Sec:random}, we replace the random phase systems with ancillary qudits. By introducing overall photon-number measurement, we can reduce the encoding state to the single-photon MDI-QKD case reviewed in Sec.~\ref{ssc:secureTMMDI}.

In Sec.~\ref{ssc:MPfixed}, we introduce a mode-pairing scheme with a fixed pairing setting, where the prepared states in different rounds are identical and independently distributed (i.i.d.). We then reduce the MP scheme to the two-mode MDI-QKD scheme. To do so, we perform global control gates on the ancillary qubits and measure them to assign the bases, denoted by $\tau$, and perform post-selection. We show that the encoding states of the MP scheme with proper post-selection will be the same as those of the two-mode MDI-QKD scheme.

In Sec.~\ref{ssc:MPfree}, we consider the MP scheme in which the pairing setting is not predetermined, but rather is determined by Charlie's announcements. We will show that when choosing a pairing setting based on Charlie's announcements, the free-pairing MP scheme is equivalent to the fixed-pairing MP scheme with the same pairing setting. The arbitrariness of the pairing setting gives us the freedom to choose pairing strategies, which can even be determined by Charlie.

Finally, in Sec.~\ref{ssc:PrepareandMeasure}, we reduce the entanglement-based scheme to the prepare-and-measure one in the main text via the Shor-Preskill argument \cite{shor2000Simple}.

\subsection{Single-photon two-mode MDI-QKD} \label{ssc:secureTMMDI}
We start with the case where Alice and Bob both hold single-photon sources. The diagram of a two-mode MDI-QKD scheme is shown in Fig.~\ref{fig:singleMDIdia} \cite{lo2012Measurement,Ma2012alternative}. Alice holds an ancillary qubit system $A'$ that interacts with two optical modes, $A_1$ and $A_2$. The single-photon subspace of the two modes forms a qubit. Bob's encoding and post-selection procedures are the same as those of Alice unless otherwise stated. The encoding process is shown in Fig.~\ref{fig:singleMDI}. Note that the encoding methods (a) and (b) in Fig.~\ref{fig:singleMDI} produce the same state $\rho_0$.

\begin{figure}[htbp]
\includegraphics[width=16cm]{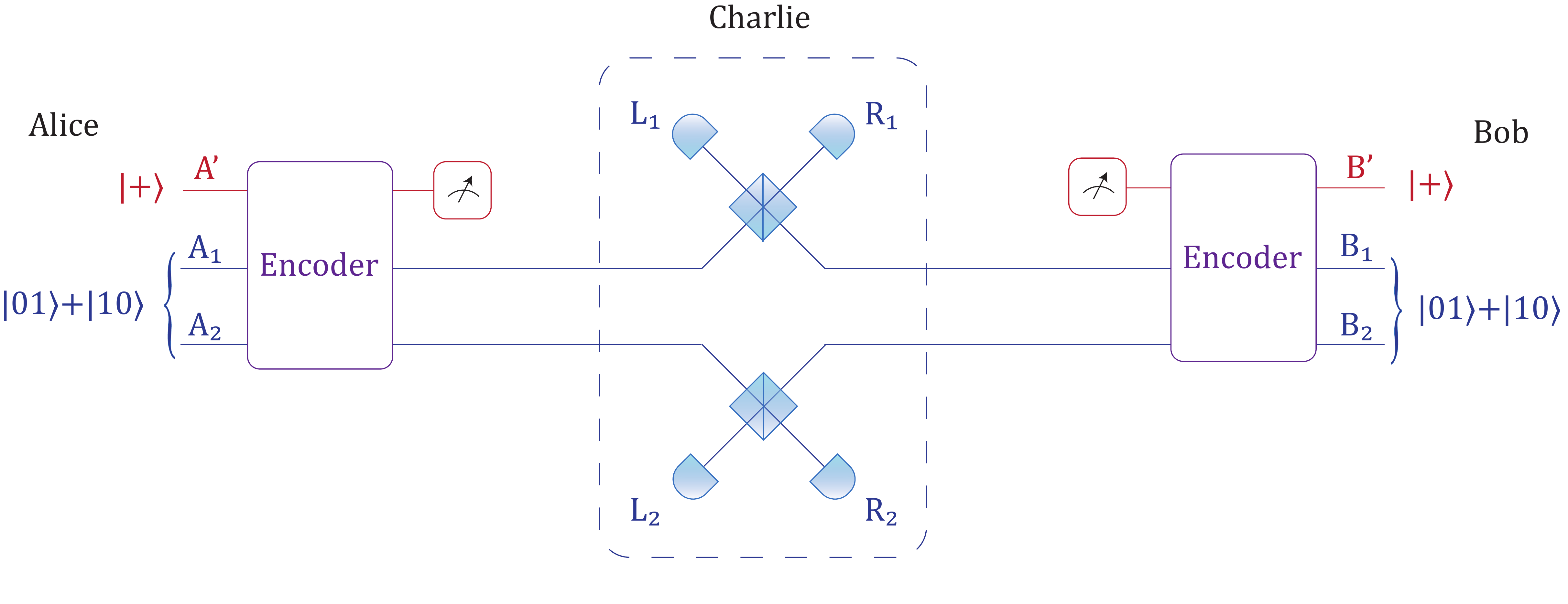}
\caption{Diagram of the entanglement version of the two-mode MDI-QKD scheme when Alice and Bob hold single photon sources. The red lines represent the ancillary qubits used by Alice and Bob to store key information and distill keys, while the blue lines represent the optical modes transmitted to Charlie.} \label{fig:singleMDIdia}
\end{figure}

\begin{figure}[htbp]
\includegraphics[width=14cm]{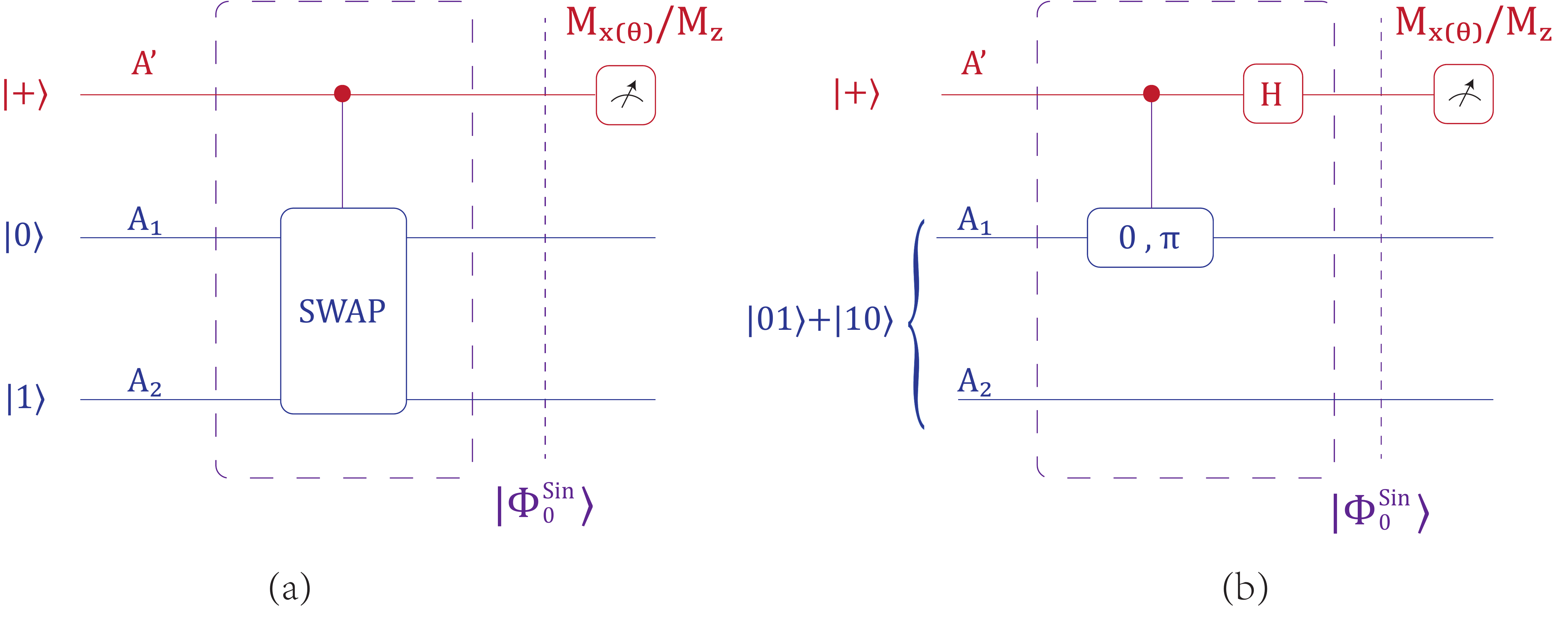}
\caption{Alice's encoding process in the entanglement version of the two-mode MDI-QKD scheme. (a) Alice uses an ancillary state $\ket{+}$ to perform a control-swap operation on two modes $A_1$ and $A_2$. (b) Alice first generates a single-photon superposition state $\frac{1}{\sqrt{2}}(\ket{01}+\ket{10})$ on two modes, $A_1$ and $A_2$; then, she uses an ancillary qubit $A'$ to control the phase of $A_1$. Note that the encoding processes described in (a) and (b) produce exactly the same state, $\ket{\Phi^{Sin}_0}$. Hereafter, in all the figures, the red line denotes the qubit ancillaries, while the blue line denotes the optical modes.} \label{fig:singleMDI}
\end{figure}

In the two-mode MDI-QKD scheme, Alice and Bob each emit encoded signals to a measurement device controlled by Charlie in the middle, who is supposed to correlate their emitted signals. Alice (Bob) uses a single photon on two orthogonal modes $A_1 (B_1)$ and $A_2 (B_2)$ as a qubit. The $Z$ basis of the qubit is naturally defined as
\begin{equation}
\begin{aligned}
\ket{\psi_z(0)} &= \ket{1}_{A_1}\ket{0}_{A_2}, \\
\ket{\psi_z(1)} &= \ket{0}_{A_1}\ket{1}_{A_2}.
\end{aligned}
\end{equation}
For simplification, we omit the tensor notation $\otimes$ between different modes unless any ambiguity occurs. The $\{X(\theta)\}_\theta$-basis is defined as
\begin{equation}
\begin{aligned}
\ket{+(\theta)} &= \frac{1}{\sqrt{2}}(\ket{\psi_z(0)} + e^{-i\theta}\ket{\psi_z(1)})=  \frac{1}{\sqrt{2}}(\ket{10} + e^{-i\theta}\ket{01})_{A_1,A_2}, \\
\ket{-(\theta)} &= \frac{1}{\sqrt{2}}(\ket{\psi_z(0)} - e^{-i\theta}\ket{\psi_z(1)})=  \frac{1}{\sqrt{2}}(\ket{10} - e^{-i\theta}\ket{01})_{A_1,A_2},
\end{aligned}
\end{equation}
where $\theta\in[0,\pi)$. When $\theta=0$ and $\pi/2$, these become the eigenstates of the $X$ and $Y$ bases, respectively. A state on the $X$-$Y$ plane can then be denoted by
\begin{equation} \label{eq:psithetaX}
\ket{\psi^\theta_x(\kappa)} = \frac{1}{\sqrt{2}}(\ket{\psi_z(0)} + e^{-i(\theta+\pi\kappa)}\ket{\psi_z(1)}),
\end{equation}
where $\theta$ denotes the basis and $\kappa$ denotes the sign of the state. The basis defined by $\{\ket{\psi^\theta_x(\kappa)}\}_{\kappa=0,1}$ is called the $X(\theta)$-basis.

In Fig.~\ref{fig:singleMDI}, the generated encoding state is
\begin{equation} \label{eq:Phisin0}
\ket{\Phi^{Sin}_0} = \frac{1}{\sqrt{2}}(\ket{0}\ket{01} + \ket{1}\ket{10})_{A';A_1,A_2},
\end{equation}
where the superscript $Sin$ indicates that the state is of a single photon. If we regard the single-photon as a qubit, then $\ket{\Phi^{Sin}_0}$ is the Bell state.

For the qubit system $A'$, we can similarly define the $X(\theta)$-basis using Eq.~\eqref{eq:psithetaX}. To realise an $X(\theta)$-basis measurement on $A'$, one first performs a $Z$-axis rotation $R_Z(\theta)$ on $A'$ before performing $X$-basis measurement,
\begin{equation} \label{eq:PhisinTheta}
\ket{\Phi^{Sin}_\theta} := R_Z(\theta)\ket{\Phi^{Sin}_0} = \frac{1}{\sqrt{2}}(\ket{0}\ket{01} + e^{i\theta}\ket{1}\ket{10})_{A';A_1,A_2}.
\end{equation}
Therefore, as shown in Fig.~\ref{fig:singleMDIreduce}(c), an $X(\theta)$-basis measurement on $A'$ can be equivalently realised by first modulating the phase of $A_1$ by $\theta$ and then measuring $A'$ on the $X$-basis. We also remark that a $Z$-basis measurement on $\ket{\Phi^{Sin}_\theta}$ is equivalent to one on $\ket{\Phi^{Sin}_0}$.

\begin{figure}[htbp!]
\includegraphics[width=12cm]{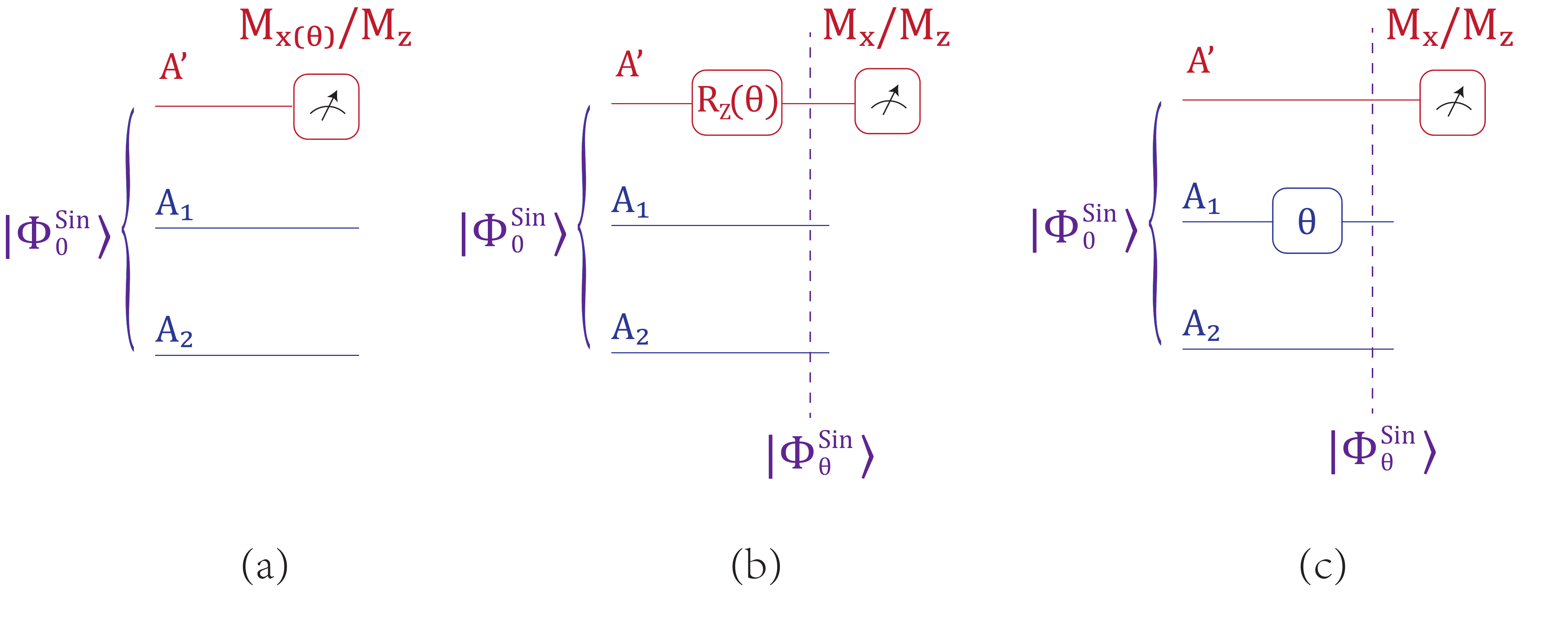}
\caption{Equivalent realisation of the encoding process and measurement in Fig.~\ref{fig:singleMDI}. The $Z$-axis rotation on qubit $A'$ can be moved to a phase modulation on mode $A_1$. The resultant state $\ket{\Phi^{Sin}_\theta}$ remains invariant. Here, the red gate $R_Z(\theta)$ denotes a $Z$-axis rotation by an angle $\theta$, while the blue gate $\theta$ denotes phase modulation on mode $A_1$.} \label{fig:singleMDIreduce}
\end{figure}

Hereafter, we establish the $Z$ basis as for key generation and the $X(\theta)$ basis  only for parameter estimation. The single-photon two-mode MDI-QKD with a fixed $\theta$ runs as shown in Box~\ref{box:singleMDIQKD}.

\begin{Boxes}{Single-photon two-mode MDI-QKD}{singleMDIQKD}

\begin{enumerate}
\item 
State preparation: Alice prepares an entangled state $\ket{\Phi^{Sin}_0}$ defined in Eq.~\eqref{eq:Phisin0} with devices (a) or (b) described in Fig.~\ref{fig:singleMDI}, containing a qubit system $A'$ and two optical modes, $A_1$ and $A_2$. Similarly, Bob prepares $\rho_0$ on $B'$, $B_1$ and $B_2$.

\item 
Measurement: Alice and Bob send their optical modes $A_1$, $A_2$, $B_1$ and $B_2$ to an untrusted party, Charlie, who is supposed to perform coincident interference measurement, as shown in Fig.~\ref{fig:singleMDIdia}.

\item 
Announcement: Charlie announces the $L_1$, $R_1$, $L_2$ and $R_2$ detection results. If one of $L_1$ and $R_1$ clicks and one of $L_2$ and $R_2$ clicks, then Alice and Bob keep their signals. If it is $(L_1,R_2)$-click or $(L_2,R_1)$-click, then Bob applies $Z$ gate on his qubit $B'$.


Alice and Bob perform the above steps over many rounds and end up with a joint $2n$-qubit state $\rho_{A' B'}\in(\mathcal{H}_{A'}\otimes\mathcal{H}_{B'})^{\otimes n}$. 

\item 
Parameter estimation: Alice decides at random whether to perform measurements in the $Z$ or $X(\theta)$ basis and announces her basis choice to Bob. They then measure $A'$, $B'$ in the same basis. They announce the $X(\theta)$-basis measurement results and estimate the (phase) error rate. The $Z$-basis measurement results on $\rho_{A' B'}$ are denoted by the raw-data string $\kappa^a$ and $\kappa^b$.

\item 
Classical post-processing: Alice and Bob reconcile the key string to $\kappa^a$ via an classical channel by consuming $l_{ec}$ key bits. They then perform privacy amplification using universal-2 hashing matrices. The sizes of the hashing matrices are determined by the estimated phase-error rate of $\kappa^a$ from the $X(\theta)$-basis error rate.
\end{enumerate}

\end{Boxes}


The security of this single-photon MDI-QKD can be reduced to that of the BBM92 scheme \cite{bbm1992quantum}. Following the security proof based on complementarity \cite{lo1999Unconditional,shor2000Simple,koashi2009simple}, we must estimate the phase-error rate, i.e., the information disturbance in the complementary basis of the $Z$ basis for key generation. Since the single-photon source is basis-independent, one can fairly estimate the phase-error rate of the $Z$ basis using the $X(\theta)$ basis.

In the classical post-processing, the information reconciliation is conducted by an encrypted classical channel with the consummation of a preshared $l_{ec}$-bit key. This is for the convenience of the description of the security analysis. In practice, one can apply different ways of one-way information reconciliation without encryption. The key rate of the single-photon MDI-QKD under the one-way classical communication will be the same.

\subsection{Coherent-state two-mode MDI-QKD} \label{ssc:secureTMMDIcoherent}
In practice, the users can replace single-photon sources with weak coherent-state sources. The security of coherent-state MDI-QKD schemes have already been well studied in previous MDI-QKD works \cite{lo2012Measurement,Ma2012alternative}, where a photon-number-channel model is assumed. That is, the untrusted Charlie first measures the global photon number $k^a$ and $k^b$ on Alice's and Bob's emitted optical modes, respectively. Based on the measurement outcomes, Charlie then decides the follow-up measurement and announcement strategy.

Here, we prove the security of coherent-state two-mode MDI-QKD from a new perspective where the random phases are purified and stored locally, as mentioned in Sec.~\ref{ssc:two_random}. In this way, Alice and Bob can decide and perform the photon-number measurement on the optical modes after emitting the signals to Charlie.

In the original coherent-state two-mode MDI-QKD scheme, Alice first generates two random phases, $\phi^a_1$ and $\phi^a_2$, which are independently and uniformly chosen from $[0,2\pi)$, and then performs the encoding process shown in Fig.~\ref{fig:coherentMDI}. The $Z$-basis encoding state is
\begin{equation} \label{eq:PhiZ0}
\ket{\Phi^Z_\mu(\phi^a_1,\phi^a_2)} = \frac{1}{\sqrt{2}} \left(\ket{0}\ket{0}\ket{\sqrt{\mu}e^{i\phi^{a}_2} } + \ket{1}\ket{\sqrt{\mu}e^{i\phi^{a}_1} }\ket{0} \right)_{A';A_1;A_2}.
\end{equation}
We use the notation $\theta^a:= \phi^a_1 - \phi^a_2$. The state can also be written as
\begin{equation}
\ket{\Phi^Z_\mu(\theta^a,\phi^a_2)} = \frac{1}{\sqrt{2}} \left(\ket{0}\ket{0}\ket{\sqrt{\mu}e^{i\phi^{a}_2} } + \ket{1}\ket{\sqrt{\mu}e^{i(\phi^{a}_2+\theta^a)} }\ket{0} \right)_{A';A_1;A_2}.
\end{equation}
If we randomise the phase $\phi^a_2$ uniformly in $[0,2\pi)$ and keep $\theta^a$ fixed, then the single-photon part of the resultant state is $\ket{\Phi^{Sin}_{\theta^a}}$ in Eq.~\eqref{eq:PhisinTheta},
\begin{equation}
\ket{\Phi^{Sin}_{\theta^a}} = R_Z(\theta^a)\ket{\Phi^{Sin}_0} = \frac{1}{\sqrt{2}}(\ket{0}\ket{01} + e^{i\theta^a}\ket{1}\ket{10})_{A';A_1,A_2}.
\end{equation}
In what follows, the coherent state with complex amplitude $\sqrt{\mu}e^{i\phi}$ will always be written in this form $\ket{\sqrt{\mu}e^{i\phi}}$ to avoid the ambiguity to the Fock states $\ket{0},\ket{1},...,\ket{k}$.

\begin{figure}[htbp]
\centering \includegraphics[width=17cm]{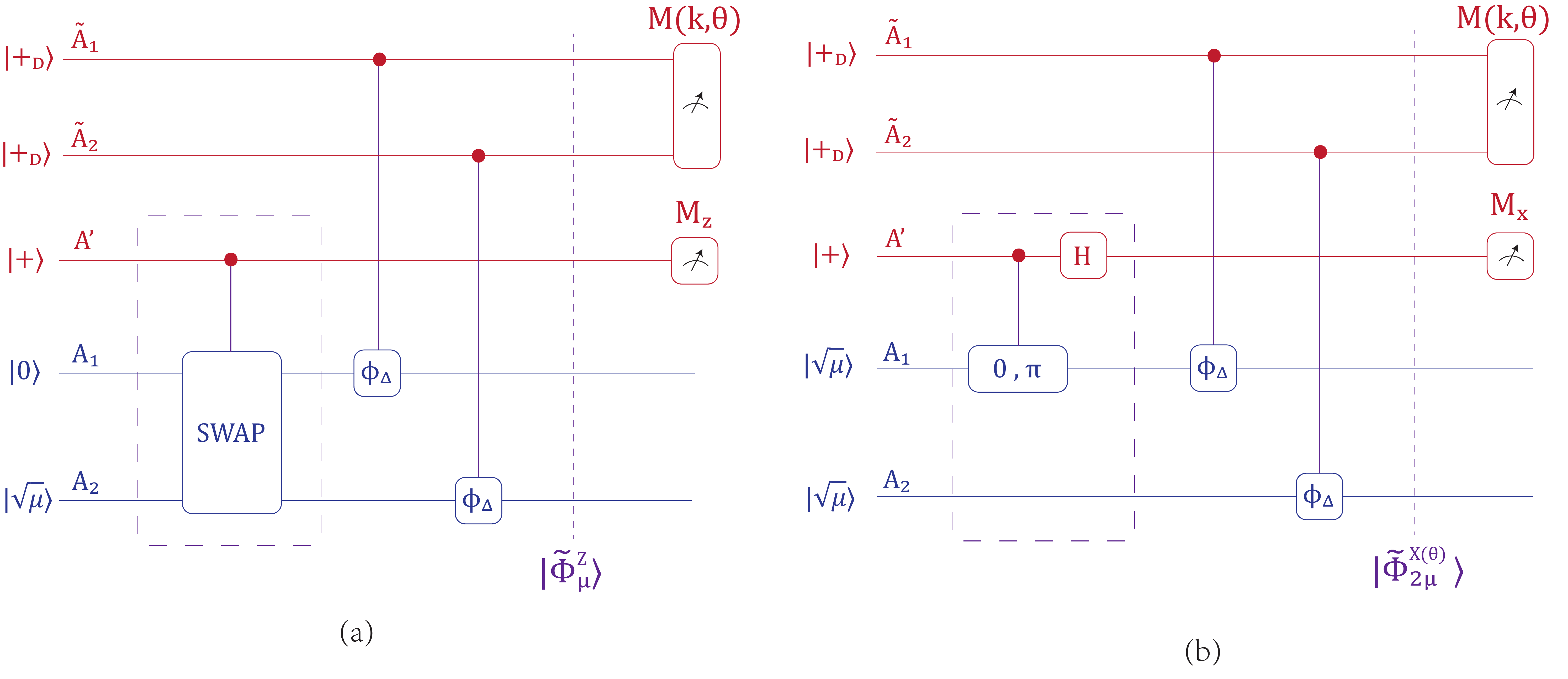}
\caption{Alice's encoding in two-mode MDI-QKD with A coherent-state source of intensity $\mu$. Alice generates the $Z$-basis state $\ket{\tilde{\Phi}^Z_\mu}$ or $X(\theta)$-basis encoding state $\ket{\tilde{\Phi}^{X(\theta^a)}_0}$ with two qudit systems $\tilde{A}_1$ and $\tilde{A}_2$ to store the random phase information $\phi^a_1$ and $\phi^a_2$, respecively. The qubit system $A'$ is used for the final key generation. Here, $\theta^a:= \phi^a_1 - \phi^a_2$. (a) and (b) present the $Z$- and $X(\theta)$-basis encoding processes, respectively. The state $\ket{+_D}:=\frac{1}{\sqrt{D}}\sum_{m=0}^{D-1}\ket{m}$. When Alice performs the global measurements $M(k,\theta)$ defined in Eq.~\eqref{eq:Mktheta} on the systems $\tilde{A}_1$ and $\tilde{A}_2$ and the photon number $k=1$ and relative phase is $\theta^a$, the post-selected state on the systems $A'$, $A_1$ and $A_2$ is approximately the single-photon-encoding state $\ket{\Phi^{Sin}_{\theta^a}}$ for both $Z$-basis and $X(\theta)$-basis case. In this case, the encoding processes in (a) and (b) will be reduced to those in (a) and (b) of Fig.~\ref{fig:singleMDI}.
} \label{fig:coherentMDI}
\end{figure}

Recall that the $Z$-basis measurement result on the state $\ket{\Phi^{Sin}_{\theta^a}}$ is independent of the relative phase $\theta^a$. For any fixed $\theta^a$, an $X(\theta)$-basis measurement on $\ket{\Phi^Z_\mu(\theta^a,\phi^a_2)}$ is equivalent to an $X(\theta+\theta^a)$-basis measurement on $\ket{\Phi^Z_\mu(0,\phi^a_2)}$. Hereafter, we do not discriminate between encoded states with different $\theta^a$ for the $Z$ basis.


Now, we further perform the source replacement for the encoded random phases $\phi^a_1$ and $\phi^a_2$, following the methods used in Sec.~\ref{Sec:random}. To do this, we first assume the random phases $\phi^a_1$ and $\phi^a_2$ are \emph{discretely and uniformly randomised} from the set $\{\frac{2\pi}{D}j\}_{j\in[D]}$. Here $D$ is the number of discrete phases. Note that picking up a $D\geq 10$ would make the discrete phase randomization very close to the continuous one \cite{Cao2015discrete}. Then, we introduce two ancillary qudit systems $\tilde{A}_1$ and $\tilde{A}_2$ with $d=D$ to store the random phase information, as shown in Fig.~\ref{fig:twoPurified}(b).

The whole $Z$-basis encoding state with the purified random phase system is
\begin{equation} \label{eq:PhiZ0purified}
\ket{\tilde{\Phi}^Z_\mu} = \frac{1}{\sqrt{2}D} \sum_{j_1=0}^{D-1} \sum_{j_2=0}^{D-1} \ket{j_1}_{\tilde{A}_1}\ket{j_2}_{\tilde{A}_2} \left(\ket{0}\ket{0}\ket{\sqrt{\mu}e^{i\phi_{j_2}} } + \ket{1}\ket{\sqrt{\mu}e^{i\phi_{j_1}} }\ket{0} \right)_{A',A_1A_2},
\end{equation}
where $\phi_{j_1}:= \frac{2\pi}{D}j_1$ and $\phi_{j_2}:= \frac{2\pi}{D}j_2$ are the two random phases with indices $j_1$ and $j_2$, respectively.

The $X(\theta)$-basis encoding state is
\begin{equation} \label{eq:PhiXtheta0}
\ket{\Phi^{X(\theta)}_{2\mu}(\phi^a_1, \phi^a_2)} = \frac{1}{\sqrt{2}} \left(\ket{0}\ket{\sqrt{\mu}e^{i(\phi^{a}_1)}}\ket{\sqrt{\mu}e^{i\phi^{a}_2} } + \ket{1}\ket{\sqrt{\mu}e^{i(\phi^{a}_1 + \pi )  } }\ket{\sqrt{\mu}e^{i\phi^{a}_2} } \right)_{A';A_1A_2}.
\end{equation}

Similar to the $Z$-basis case, we introduce purified systems $\tilde{A}_1$ and $\tilde{A}_2$ to register the random phase information,
\begin{equation} \label{eq:PhiXtheta0purified}
\begin{aligned}
\ket{\tilde{\Phi}^{X(\theta)}_{2\mu}} = \frac{1}{\sqrt{2}D} \sum_{j_1=0}^{D-1} \sum_{j_2=0}^{D-1} \ket{j_1}_{\tilde{A}_1}\ket{j_2}_{\tilde{A}_2} \left(\ket{+}\ket{\sqrt{\mu}e^{i\phi_{j_1}} } \ket{\sqrt{\mu}e^{i\phi_{j_2}} } + \ket{-} \ket{\sqrt{\mu}e^{i(\phi_{j_1} + \pi ) } } \ket{\sqrt{\mu}e^{i\phi_{j_2}} } \right)_{A';A_1A_2}.
\end{aligned}
\end{equation}

To reduce the $Z$-basis and $X(\theta)$-basis encoding state $\ket{\tilde{\Phi}^Z_\mu}$ and $\ket{\tilde{\Phi}^{X(\theta)}_{2\mu}}$ to the single-photon encoding state $\ket{\Phi^{Sin}_\theta}$ in Eq.~\eqref{eq:PhisinTheta}, we have the following lemma.

\begin{lemma}[Encoding-state reduction from coherent-state to single-photon two-mode MDI-QKD]\label{lem:Coh2Sin}
In coherent-state two-mode MDI-QKD,  Alice generates the state $\ket{\tilde{\Phi}^Z_\mu}$ defined in Eq.~\eqref{eq:PhiZ0purified} or $\ket{\tilde{\Phi}^{X(\theta)}_{2\mu}}$ defined in Eq.~\eqref{eq:PhiXtheta0purified} for $Z$-basis and $X(\theta)$-basis encoding, respectively. She then performs the global measurement $M(k,\theta)$ defined in Eq.~\eqref{eq:Mktheta} on the local ancillary qudits $\tilde{A}_1$ and $\tilde{A}_2$. When the discrete phase number $D\to \infty$, we have,
\begin{enumerate}
\item
(Poisson photon-number distribution) For the $Z$-basis state $\ket{\tilde{\Phi}^Z_\mu}$ with overall intensity $\mu$, the probability of the photon-number measurement result $k$ is $\Pr(k) = e^{-\mu} \frac{\mu^k}{k!}$; for the $X(\theta)$-basis state $\ket{\tilde{\Phi}^{X(\theta)}_{2\mu}}$ with overall intensity $2\mu$, the probability of the photon-number measurement result $k$ is $\Pr(k) = e^{-2\mu} \frac{(2\mu)^k}{k!}$;
	\item (Independence of the photon-number states to the intensity) The resultant state after the pairwise measurement $M(k,\theta)$ is independent of the intensity value $\mu$.
	\item (Basis-independence of the single-photon state) If the measurement result on $\ket{\tilde{\Phi}^Z_\mu}$ or $\ket{\tilde{\Phi}^{X(\theta)}_{2\mu}}$ is $k=1$ and $\theta=\frac{2\pi}{D}j_\theta$, then the conditional state on $A'$, $A_1$ and $A_2$ will be reduced to $\ket{\Phi^{Sin}_{\theta}}$ defined in Eq.~\eqref{eq:PhisinTheta}.
\end{enumerate}
\end{lemma}

\begin{proof}
In Section~\ref{ssc:two_random}, we introduce a global basis transformation on the ancillary systems $\tilde{A}_1$ and $\tilde{A}_2$.
Assuming all the phases are discretely chosen from the set $\{\frac{2\pi}{D}j\}_{j\in[D]}$, we can express a purified encoded state Eq.~\eqref{eq:PhiZ0purified} and transform the basis on it,
\begin{equation} \label{eq:PhiZ0purified2}
\begin{aligned}
\ket{\tilde{\Phi}^Z_\mu} &= \frac{1}{\sqrt{2}D} \sum_{j_1=0}^{D-1} \sum_{j_2=0}^{D-1} \ket{j_1}_{\tilde{A}_1}\ket{j_2}_{\tilde{A}_2} \left(\ket{0}\ket{0}\ket{\sqrt{\mu}e^{i\phi_{j_2}} } + \ket{1}\ket{\sqrt{\mu}e^{i\phi_{j_1}} }\ket{0} \right)_{A',A_1A_2} \\
&= \frac{1}{\sqrt{2}D} \sum_{j_2=0}^{D-1} \sum_{j_\theta=0}^{D-1} \ket{j_2+j_\theta}_{\tilde{A}_1}\ket{j_2}_{\tilde{A}_2} \left(\ket{0}\ket{0}\ket{\sqrt{\mu}e^{i\phi_{j_2}} } + \ket{1}\ket{\sqrt{\mu}e^{i\phi_{j_2+j_\theta}} }\ket{0} \right)_{A',A_1A_2} \\
&= \frac{1}{\sqrt{2D}D} \sum_{j_2=0}^{D-1} \sum_{j_\theta=0}^{D-1} \sum_{k=0}^{D-1} e^{-i \frac{2\pi}{D}j_2 k } \ket{\tilde{k},j_\theta}_{\tilde{A}_1,\tilde{A}_2} \left(\ket{0}\ket{0}\ket{\sqrt{\mu}e^{i\phi_{j_2}} } + \ket{1}\ket{\sqrt{\mu}e^{i\phi_{j_2+j_\theta}} }\ket{0} \right)_{A',A_1A_2} \\
&= \frac{1}{\sqrt{2D}D} \sum_{j_\theta=0}^{D-1} \sum_{k=0}^{D-1} \ket{\tilde{k},j_\theta}_{\tilde{A}_1,\tilde{A}_2} \sum_{j_2=0}^{D-1} \left[ e^{-i \frac{2\pi}{D}j_2 k } \left(\ket{0}\ket{0}\ket{\sqrt{\mu}e^{i\phi_{j_2}} } + \ket{1}\ket{\sqrt{\mu}e^{i\phi_{j_2+j_\theta}} }\ket{0} \right)_{A',A_1A_2} \right] \\
&= \frac{1}{\sqrt{D}D} \sum_{j_\theta=0}^{D-1} \sum_{k=0}^{D-1} \ket{\tilde{k},j_\theta}_{\tilde{A}_1,\tilde{A}_2} \hat{U}_{A_1}(\frac{2\pi}{D}j_\theta) \circ \hat{C\text{-}SWAP}  \left( \sum_{j_2=0}^{D-1}  e^{-i \frac{2\pi}{D}j_2 k } \ket{+}_{A'} \ket{0}_{A_1} \ket{ \sqrt{\mu}e^{i\phi_{j_2}} }_{A_2} \right) \\
&= \sum_{j_\theta=0}^{D-1} \frac{1}{\sqrt{D}} \sum_{k=0}^{D-1} \sqrt{P_k^\mu} \ket{\tilde{k},j_\theta}_{\tilde{A}_1,\tilde{A}_2} \hat{U}_{A_1}(\frac{2\pi}{D}j_\theta) \circ \hat{C\text{-}SWAP} \left( \ket{+}_{A'} \ket{0}_{A_1} \ket{\lambda^\mu(k)}_{A_2} \right) \\
&= \sum_{j_\theta=0}^{D-1} \frac{1}{\sqrt{D}} \sum_{k=0}^{D-1} \sqrt{P_k^\mu} \ket{\tilde{k},j_\theta}_{\tilde{A}_1,\tilde{A}_2} \ket{\lambda^\mu_Z(k), j_\theta }_{A',A_1A_2}.
\end{aligned}
\end{equation}
Here, in the third equality, we introduce a partial Fourier transform defined in Eq.~\eqref{eq:twoFT}. In the fifth equality, we simplify the expression by a phase-gate $\hat{U}_{A_1}(\phi) := e^{i\phi a^\dag_1 a_1}$ and a controlled-swap gate $\hat{C\text{-}SWAP}$ defined on a qubit $A'$ and two optical modes,
\begin{equation}
\begin{aligned}
\hat{C\text{-}SWAP}\ket{0}_{A'}\ket{\psi(a^\dag_1, a^\dag_2)} &= \ket{0}_{A'}\ket{\psi(a^\dag_1, a^\dag_2)}, \\
\hat{C\text{-}SWAP}\ket{1}_{A'}\ket{\psi(a^\dag_1, a^\dag_2)} &= \ket{1}_{A'}\ket{\psi(a^\dag_2, a^\dag_1)}.
\end{aligned}
\end{equation}
In the sixth equality, we use the pseudo-Fock state definition in Eq.~\eqref{eq:pseudoFock}. The probability
\begin{equation}
P_k^\mu := e^{-\mu}\sum_{m=0}^{\infty} \frac{\mu^{(mD+k)}}{(mD+k)!},
\end{equation}
is a mixture of the Poisson distribution probability $P^{\mu}(k):= e^{-\mu}\frac{\mu^k}{k!}$. In the seventh equality, the state $\ket{\lambda^\mu_Z(k),j_\theta}_{A,A_1A_2}$ is
\begin{equation} \label{eq:lambdaZk}
\begin{aligned}
\ket{\lambda^\mu_Z(k),j_\theta}_{A,A_1A_2}:&= \hat{U}_{A_1}(\frac{2\pi}{D}j_\theta) \circ \hat{C\text{-}SWAP} \left( \ket{+}_{A'}\ket{0}_{A_1}\ket{\lambda^\mu(k)}_{A_2} \right) \\
&= \hat{U}_{A_1}(\frac{2\pi}{D}j_\theta) \circ \hat{C\text{-}SWAP} \left( \ket{+}_{A'}\, \frac{e^{-\frac{\mu}{2}}}{\sqrt{P_k^{\mu}}}\sum_{m=0}^{\infty} \frac{\sqrt{\mu}^{(mD+k)}}{\sqrt{(mD+k)!}} \ket{0}_{A_1} \ket{mD+k}_{A_2} \right),
\end{aligned}
\end{equation}
is a normalized state with pseudo-Fock number $k$ and relative phase $\theta=\frac{2\pi}{D}j_\theta$.

When $D\to \infty$, the probability $P_k^\mu$ becomes the Poisson distribution. The conditional encoding state $\ket{\lambda^\mu_Z(k), j_\theta}_{A',A_1A_2}$ becomes
\begin{equation} \label{eq:lambdaZkDinf}
\ket{\lambda^\mu_Z(k),j_\theta}_{A,A_1A_2} = \hat{U}_{A_1}(\frac{2\pi}{D}j_\theta) \circ \hat{C\text{-}SWAP} \left( \ket{+}_{A'}\, \ket{0}_{A_1} \,\ket{k}_{A_2} \right),
\end{equation}
which is irrelavant to the intensity $\mu$. Especially, when $k=1$, the conditinal state becomes
\begin{equation} \label{eq:lambdaZk1}
\ket{\lambda^\mu_Z(1),j_\theta}_{A,A_1A_2} = \frac{1}{\sqrt{2}} (\ket{0}\ket{01} + e^{i \theta}\ket{1}\ket{10})_{A,A_1A_2} = \ket{\Phi^{Sin}_{\theta}},
\end{equation}
Here, $\ket{\Phi^{Sin}_{\theta}}$ is defined in Eq.~\eqref{eq:PhisinTheta}. When $D$ is finite, the equality of Eq.~\eqref{eq:lambdaZk1} becomes approximation due to the discrete phase randomisation effect. It has been shown in the literature~\cite{Cao2015discrete,Zeng2019Symmetryprotected} that when $D\geq 12$, the discrete phase randomization is very close to the continuous one.

If we perform the same basis transformation on $\tilde{A}_1$ and $\tilde{A}_2$ as the one for the $Z$-basis state in Eq.~\eqref{eq:PhiZ0purified2}, the $X(\theta)$-basis encoding state will become
\begin{equation} \label{eq:PhiXtheta0purified2}
\begin{aligned}
\ket{\tilde{\Phi}^{X(\theta)}_{2\mu}} &=
\frac{1}{\sqrt{2}D} \sum_{j_\theta=0}^{D-1} \sum_{j_2=0}^{D-1} \ket{j_2+j_\theta}_{\tilde{A}_1}\ket{j_2}_{\tilde{A}_2} \left(\ket{+}\ket{\sqrt{\mu}e^{i\phi_{j_2+j_\theta}} } \ket{\sqrt{\mu}e^{i\phi_{j_2}} } + \ket{-} \ket{\sqrt{\mu}e^{i(\phi_{j_2+j_\theta} + \pi ) } } \ket{\sqrt{\mu}e^{i\phi_{j_2}} } \right)_{A',A_1A_2} \\
&= \frac{1}{\sqrt{2D}D} \sum_{j_\theta=0}^{D-1} \sum_{j_2=0}^{D-1} \sum_{k=0}^{D-1} e^{-i \frac{2\pi}{D}j_2 k } \ket{\tilde{k},j_\theta}_{\tilde{A}_1,\tilde{A}_2} \left(\ket{+}\ket{\sqrt{\mu}e^{i\phi_{j_2+j_\theta}} } \ket{\sqrt{\mu}e^{i\phi_{j_2}} } + \ket{-} \ket{\sqrt{\mu}e^{i(\phi_{j_2+j_\theta} + \pi ) } } \ket{\sqrt{\mu}e^{i\phi_{j_2}} } \right)_{A',A_1A_2} \\
&= \frac{1}{\sqrt{2D}D} \sum_{j_\theta=0}^{D-1} \sum_{k=0}^{D-1} \ket{\tilde{k},j_\theta}_{\tilde{A}_1,\tilde{A}_2} \sum_{j_2=0}^{D-1} \left[ e^{-i \frac{2\pi}{D}j_2 k } \left(\ket{+}\ket{\sqrt{\mu}e^{i\phi_{j_2+j_\theta}} } \ket{\sqrt{\mu}e^{i\phi_{j_2}} } + \ket{-} \ket{\sqrt{\mu}e^{i(\phi_{j_2+j_\theta} + \pi ) } } \ket{\sqrt{\mu}e^{i\phi_{j_2}} } \right)_{A',A_1A_2} \right] \\
&= \frac{1}{\sqrt{D}} \sum_{j_\theta=0}^{D-1} \sum_{k=0}^{D-1} \sqrt{P_k^{2\mu}} \ket{\tilde{k},j_\theta}_{\tilde{A}_1,\tilde{A}_2} \hat{H}_{A'} \hat{C\text{-}U(\pi)}_{A',A_1} \left( \ket{+}_{A'}\ket{\lambda_{1,2}^{2\mu}(k),j_\theta}_{A_1A_2} \right)  \\
&= \frac{1}{\sqrt{D}} \sum_{j_\theta=0}^{D-1} \sum_{k=0}^{D-1} \sqrt{P_k^{2\mu}} \ket{\tilde{k},j_\theta}_{\tilde{A}_1,\tilde{A}_2} \ket{\lambda^{2\mu}_{X(\theta)}(k)}_{A',A_1A_2}.
\end{aligned}
\end{equation}
Here, in the second equality, we use the partial Fourier transform in Eq.~\eqref{eq:twoFT}. In the fourth equality, we simplify the expression with a Hadamard gate $\hat{H}_{A'}$ on the qubit $A'$ and the controlled phase gate $\hat{C\text{-}U(\phi)}_{A',A_1}$,
\begin{equation}
\hat{C\text{-}U(\phi)}_{A',A_1} = \ket{0}_A'\bra{0}\otimes I_{A_1} + \ket{1}_A'\bra{1}\otimes U(\phi)_{A_1}.
\end{equation}
The state $\ket{\lambda^{2\mu}_{1,2}(k), j_\theta}_{A',A_1A_2}$ is defined in Eq.~\eqref{eq:lambda_twomode}. In the fifth equality, we define
\begin{equation}
\ket{\lambda^{2\mu}_{X(\theta)}(k)}_{A',A_1A_2} := \hat{H}_{A'} \hat{C\text{-}U(\pi)}_{A',A_1} \left( \ket{+}_{A'}\ket{\lambda_{1,2}^{2\mu}(k),j_\theta}_{A_1A_2} \right).
\end{equation}

If Alice performs the measurement $M(k,\theta)$ defined in Eq.~\eqref{eq:Mktheta} on the ancillary systems $\tilde{A}_1$ and $\tilde{A}_2$ and obtains $k$ and $\theta= \frac{2\pi}{D}j_\theta$, then the conditional emitted state is $\ket{\lambda^{2\mu}_{X(\theta)}(k)}_{A,A_1A_2}$. When $D\to \infty$, the probability $P^{2\mu}_k$ to get the photon-number result $k$ becomes the Poisson distribution. As is shown in Eq.~\eqref{eq:lambda12kDInf}, the state $\ket{\lambda^{2\mu}_{1,2}(k), j_\theta}_{A',A_1A_2}$ (and hence the state $\ket{\lambda^{2\mu}_{X(\theta)}(k)}_{A',A_1A_2}$) becomes independent of the intensity $\mu$. Especially, when $k=1$, the conditional state $\ket{\lambda^{2\mu}_{X(\theta)}(k)}_{A',A_1A_2}$ becomes
\begin{equation} \label{eq:lambdaX1}
\ket{\lambda^{2\mu}_{X(\theta)}(1)}_{A,A_1A_2} = \frac{1}{2} \left[\ket{+}(\ket{01} + e^{i \theta}\ket{10}) + \ket{-}(\ket{01} - e^{i \theta}\ket{10}) \right]_{A,A_1A_2} = \ket{\Phi^{Sin}_{\theta}},
\end{equation}
where $\ket{\Phi^{Sin}_{\theta}}$ is defined in Eq.~\eqref{eq:PhisinTheta}. Again, when $D\geq 12$, the approximation caused by discrete (instead of continuous) phase randomization can be ignored.
\end{proof}

As is shown in Eqs.~\eqref{eq:PhiZ0purified2} and \eqref{eq:PhiXtheta0purified2}, Alice can obtain the information of the overall photon number $k$ of $A_1$ and $A_2$ indirectly based on the collective measurement on the ancillary systems $\tilde{A}_1$ and $\tilde{A}_2$. When the measurement result is $k=1$, both of the $Z$-basis and $X(\theta)$-basis encoding state will be reduced to the single-photon encoding state $\ket{\Phi^{Sin}_{\theta}}$ defined in Eq.~\eqref{eq:PhisinTheta}. In this way, the security of the coherent-state MDI-QKD will be reduced to the one of single-photon MDI-QKD.

In the coherent-state two-mode MDI-QKD scheme, Alice and Bob use $Z$-basis data to generate a secure key and $X(\theta)$-basis data to estimate information leakage, reflected in the $X(\theta)$-basis single-photon error rate $e_{11}^{X(\theta)}$. Here, the relative phase $\theta^a (\theta^b)$ is not fixed, but determined by $\phi^a_1 - \phi^a_2$ $(\phi^b_1 - \phi^b_2)$, as randomly chosen from $[0,2\pi)$. During the basis-sifting process, Alice and Bob sift the data with $\theta^a - \theta^b = 0$ or $\pi$ to estimate the privacy of raw keys generated in the $Z$-basis. Here, we mix all $X(\theta)$ data with different $\theta$ to estimate the average phase-error rate $\bar{e}^X_{11}$,
\begin{equation}
    \bar{e}^X_{11}:= \sum_{\theta} P(\theta) e^{X(\theta)}_{(1,1)},
\end{equation}
where $P(\theta)$ is the conditional probability of choosing the alignment angle $\theta$ in all of the sifted $X$-basis data with $\theta^a = \theta^b$ and $e^{X(\theta)}_{(1,1)}$ is the phase-error rate when the alignment angle is $\theta^a = \theta^b = \theta$. This averaged phase-error rate can still faithfully characterise the privacy of $Z$-basis key generation due to the concavity of the binary entropy function.

From Lemma~\ref{lem:Coh2Sin} we can see that, in the coherent-state two-mode MDI-QKD scheme, if Alice and Bob post-select the signals with the overall photon number $k^a=k^b=1$, then the remaining signals can be used to extract key information similar to the single-photon MDI-QKD scheme in Sec.~\ref{ssc:secureTMMDI}. In practice, Alice and Bob do not need to keep track of the photon number data and perform the photon-number post-selection. Instead, they may apply the ``tagging'' idea~\cite{gottesman2004security}: they first keep all the raw-key data generated by different photon numbers $k^a$ and $k^b$ and treat the photon number as a ``tag'' on the emitted signal; they then estimate the fraction $q_{(1,1)}$ of the single-photon-pair signals with $k^a=k^b=1$ among all the raw-key data and the $X(\theta)$-basis error rate of the single-photon-pair signals $e^{X(\theta)}_{(1,1)}$. Based on the estimation of $q_{(1,1)}$ and $e^{X(\theta)}_{(1,1)}$, they can also extract the same amount of secure keys by designing proper privacy amplification procedures. The key-rate formula is given by~\cite{gottesman2004security}
\begin{equation} \label{eq:keyrateDecoyMDI}
r = q_{(1,1)} \left[ 1 - H(\bar{e}^X_{11}) \right] - f Q_{\mu\mu} H(E^Z_{(\mu,\mu)}).
\end{equation}
where $q_{(1,1)}$ is the estimated detected fraction of single-photon components, $H(x)$ is the binary entropy function, and $f$ is a factor that reflects the efficiency of information reconciliation. $Q_{\mu\mu}$ and $E_{\mu\mu}$ are the overall detected fraction and bit-error rate when Alice and Bob emit coherent states with intensity $\mu$. Furthermore, Alice and Bob do not have to measure the photon-number $k^a$ and $k^b$ in each round. Thanks to the photon-number Poisson statistics and the independence of the photon-number state to the intensity setting in Lemma~\ref{lem:Coh2Sin}, Alice and Bob can vary the intensity $\mu$ to estimate the single-photon component in the final detected rounds using decoy-state methods~\cite{Lo2005Decoy}.


To summarize, we list the coherent-state two-mode MDI-QKD in Box~\ref{box:coherentMDIQKD}. Here, we asuume that Alice and Bob perform the photon-number measurement, but perform the data post-processing based on the tagging key rate formula in Eq.~\eqref{eq:keyrateDecoyMDI}.

\begin{Boxes}{Coherent-state two-mode MDI-QKD}{coherentMDIQKD}
\begin{enumerate}
\item 
State preparation: Alice decides at random whether to prepare the $Z$- or $X(\theta)$-basis states. For the $Z$-basis, Alice prepares an entangled state $\ket{\tilde{\Phi}^Z_\mu}$ defined in Eq.~\eqref{eq:PhiZ0purified} using the device in Fig.~\ref{fig:coherentMDI}(a), which contains a qubit system $A'$, two ancillary qudit systems $\tilde{A}_1$ and $\tilde{A}_2$, and two optical modes $A_1$ and $A_2$. For the $X(\theta)$-basis, she prepares $\ket{\tilde{\Phi}^{X(\theta)}_{2\mu}}$ defined in Eq.~\eqref{eq:PhiXtheta0purified} by means of Fig.~\ref{fig:coherentMDI}(b). Similarly, Bob prepares $\ket{\tilde{\Phi}^Z_\mu}$ or $\ket{\tilde{\Phi}^{X(\theta)}_{2\mu}}$ randomly on systems $B'$, $\tilde{B}_1$, $\tilde{B_2}$, $B_1$ and $B_2$.

\item 
Measurement: Alice and Bob send their optical modes $A_1$, $A_2$, $B_1$ and $B_2$ to an untrusted party, Charlie, who is supposed to perform coincident interference measurement, as shown in Fig.~\ref{fig:singleMDIdia}.

\item 
Announcement: Charlie announces the $L_1$, $R_1$, $L_2$ and $R_2$ detection results. If one of $L_1$ and $R_1$ clicks and one of $L_2$ and $R_2$ clicks, Alice and Bob keep the signals. Alice performs $M(k,\theta)$ measurement defined in Eq.~\eqref{eq:Mktheta} on qudits $\tilde{A}_1$ and $\tilde{A}_2$ to obtain the photon number $k^a$ and the relative phase $\theta^a$. Bob operates similarly to obatin $k^b$ and $\theta^b$. If there is an $(L_1,R_2)$-click or $(L_2,R_1)$-click, then Bob applies the $Z$ gate to his qubit $B'$. They announce the measured photon number $k^a$ and $k^b$.

\item 
Basis sifting: Alice and Bob announce that the chosen basis is either $Z$ or $X(\theta)$. If the chosen basis is $X(\theta)$, they further announce the relative phases $\theta^a$ and $\theta^b$. If they both choose the $Z$-basis or the $X(\theta)$-basis with $\theta^a - \theta^b=0/\pi$, they keep their data.

Alice and Bob perform the above steps over many rounds and end up with a joint $2n$-qubit state $\rho_{A' B'}\in(\mathcal{H}_{A'}\otimes\mathcal{H}_{B'})^{\otimes n}$. 

\item 
Key mapping: Alice and Bob measure the qubit systems $A'$ and $B'$ in the $Z$ or $X$ bases when the predetermined bases are $Z$ or $X(\theta)$, respectively. They record the $Z$-basis measurement results as the raw-data strings $\kappa^a$ and $\kappa^b$.

\item 
Parameter estimation: Alice and Bob announce the $X(\theta)$-basis measurement results. Based on the announced information, they then estimate the $X(\theta)$-basis error rate $e^{X(\theta)}_{(1,1)}$ for all the signals with $k^a=k^b=1$ and the fraction of signals with $k^a=k^b=1$ in the remained signals $q_{(1,1)}$.

\item 
Classical post-processing: Alice and Bob reconcile the key string to $\kappa^a$ and perform privacy amplification using universal-2 hashing matrices. The sizes of hashing matrices are determined by the estimated single-photon-pair fraction $q_{(1,1)}$ and $X(\theta)$-basis error rate $e^{X(\theta)}_{(1,1)}$ according to Eq.~\eqref{eq:keyrateDecoyMDI}.
\end{enumerate}
\end{Boxes}

\subsection{Fixed-pairing MP scheme} \label{ssc:MPfixed}

In this subsection, starting from the coherent-state MDI-QKD scheme, we introduce extra encoding redundancies to construct the mode-pairing (MP) scheme. The major differences of the MP scheme and the coherent-state two-mode MDI-QKD are
\begin{enumerate}
\item
All the optical modes are decoupled during the encoding process; the correlation among two different rounds $i$ and $j$ are built after the generation of encoding states;
\item
Instead of being determined at the beginning, the basis choice for each pair of location is determined by a pairwise measurement on the ancillary systems afterwards.
\end{enumerate}

Unlike the coherent-state MDI-QKD scheme where Alice emits two correlated optical modes in each round, in the MP scheme, Alice first emit optical modes in an i.i.d. manner, then group each two of the sending locations together as a ``pair''. She then perform collective quantum operations on each pair to extract the correlation information. The concept of ``pairing'' is formalized below. For the convenience of later discussion, we also define Charlie's announcement.

\begin{definition}[Pairing Setting $\vec{\chi}$]\label{def:pairing}
For $N$ location indexes $i=1,2,...,N$ where $N$ is an even number, a pairing setting $\vec{\chi}$ is a set composed of $N/2$ location pairs $(i,j)$ such that all the location indexes appear once and only once in $\vec{\chi}$. The possible pairing settings $\{\vec{\chi}\}$ form a finite set $\mathcal{X}$.
\end{definition}

\begin{definition}[Charlie's Announcement $\vec{C}$]\label{def:Cannounce}
For location indexes $i=1,2,...,N$, denote Charlie's announcement for the $i$-th round to be $C_i = (L_i, R_i)$, where $L_i$ and $R_i$ are two binary variables indicating the detection results of detectors $L$ and $R$, respectively. The detection announcement on all the locations can be denoted by a vector $\vec{C}$. The possible Charlie's annoucements $\{\vec{C}\}$ form a finite set $\mathcal{C}$.
\end{definition}

Once the $N$ optical modes are paired, we can denote the location of the $k$-th pair as $(F_k, R_k)$, where $F_k$ and $R_k$ are the front and rear pulse locations of the $k$-th pair, respectively. The pairing setting $\vec{\chi}$ can also be viewed as a vector formed by $(F_k, R_k)$. Without loss of generality, we focus on two specific paired locations 1 and 2 in the following discussions.

In the security proof, we first consider a simple MP scheme, where the pairing setting $\vec{\chi}$ is not related to the untrusted Charlie's announcement. We call this MP scheme the fixed-pairing MP scheme.

\begin{definition}[$\vec{\chi}$-Fixed-Pairing MP] \label{def:fixedMP}
In a MP MDI-QKD scheme, if the determination of pairing setting $\vec{\chi}$ is independent of the measurement announcement by the untrusted interferometer, Charlie, then we call such MP scheme the fixed-pairing MP scheme. The fixed-pairing MP scheme with the pairing setting $\vec{\chi}$ is called the $\vec{\chi}$-fixed-pairing MP scheme.
\end{definition}

In the MP scheme, Alice sends out a coherent state with independent intensity and phase modulation in each round,
\begin{equation}
\ket{\psi^{Com}}_{A_1} = \ket{\sqrt{z_1' \mu} e^{i (z_1''\pi + \tilde{z}_1 \frac{2\pi}{D})} },
\end{equation}
where $z_1'$ and $z_1''$ are two random bit values indicating the encoded intensity and $0/\pi$-phase, respectively; $\tilde{z}_1$ is a random dit indicating the discrete random phase. Here, we take the encoding state of the first round, with subscript 1, as the example.

To introduce the entanglement version of the scheme, we perform the source replacement to introduce two ancillary qubits $A_1'$ and $A_2''$ and an ancillary qudit $\tilde{A}_1$. In this case, Alice will generate the same composite state $\ket{\tilde{\Psi}^{Com}}$ for every round,
\begin{equation} \label{eq:PsiComPurified2}
\begin{aligned}
\ket{\tilde{\Psi}^{Com}}_{\tilde{A}_1, A_1', A_1'',A_1} &= \frac{1}{\sqrt{2}} \left( \ket{0}_{A_1'}\ket{\tilde{\Psi}^{X(\theta)}_{0}}_{\tilde{A}_1,A_1'',A_1} + \ket{1}_{A_1'}\ket{\tilde{\Psi}^{X(\theta)}_{\mu}}_{\tilde{A}_1,A_1'',A_1} \right) \\
&= \frac{1}{2\sqrt{D}} \sum_{j_1=0}^{D-1} \ket{j_1}_{\tilde{A}_1} \left( \ket{00}\ket{0} + \ket{01}\ket{0} + \ket{10}\ket{\sqrt{\mu}e^{i\phi^a_{j_1}}} + \ket{11}\ket{\sqrt{\mu}e^{i(\phi^a_{j_1}+\pi)}} \right)_{A_1',A_1'';A_1},
\end{aligned}
\end{equation}
which is composed of an ancillary qudit $\tilde{A}_1$, two ancillary qubits $A_1'$ and $A_1''$, and an optical mode $A_1$. In the second equality, we use the following definition of the state $\ket{\tilde{\Psi}^{X(\theta)}_{\mu}}_{\tilde{A}_1,A_1'',A_1}$
\begin{equation} \label{eq:PsiXmupurified}
\ket{\tilde{\Psi}^{X(\theta)}_{\mu}}_{\tilde{A}_1,A_1'',A_1} = \frac{1}{\sqrt{2D}} \sum_{j_1=0}^{D-1} \ket{j_1}_{\tilde{A}_1} \left( \ket{0}\ket{\sqrt{\mu} e^{i\phi^a_{j_1}}} + \ket{1}\ket{\sqrt{\mu} e^{i(\phi^a_{j_1}+\pi)} } \right)_{A_1''A_1}.
\end{equation}

Interestingly, Alice can post-select the encoding state of coherent-state MDI-QKD by performing some collective operations on two location pairs, as is stated in the following lemma.

\begin{lemma}[Encoding-state reduction from the fixed-pairing MP scheme to coherent-state MDI-QKD]\label{lem:fixedMPreduce}
Alice generates two composite states, $\ket{\tilde{\Psi}^{Com}}_{\tilde{A}_1,A_1',A_1'',A_1}$ and $\ket{\tilde{\Psi}^{Com}}_{\tilde{A}_2,A_2',A_2'',A_2}$, as defined in Eq.~\eqref{eq:PsiComPurified2}, on two locations $1$ and $2$, respectively. Locally, Alice applies a CNOT gate from $A_1'$ to $A_2'$, measures $A_2'$ on the $Z$ basis, and obtains a result $\tau^a$.

If $\tau^a=1$, Alice assigns the $Z$ basis. She then measures $A_1''$ and $A_2''$ on the $Z$-basis to obtain results $z_1''$ and $z_2''$, respectively. Then, the conditional state in the systems $\tilde{A}_1$, $\tilde{A}_2$, $A_1'$, $A_1$ and $A_2$ can be reduced to the $Z$-basis encoding state $\ket{\tilde{\Phi}^Z_\mu}$ of Eq.~\eqref{eq:PhiZ0purified} used in the coherent-state MDI-QKD scheme in Sec.~\ref{ssc:secureTMMDI} with possible $\pi$-phase modulations on the optical modes $A_1$ and $A_2$ when $z_1''=1$ and $z_2''=1$, respectively.

If $\tau^a=0$, Alice assigns the $X(\theta)$ basis. She first measures $A_1'$ on the $Z$ basis to obtain the result $\lambda^a$. If $\lambda^a=0$, the state will become $\ket{\tilde{\Phi}^{X(\theta)}_{0}}$, which is used for the decoy state estimation. If $\lambda^a=1$, Alice applies a CNOT gate from $A_2''$ to $A_1''$. After that, Alice measures $A_2''$ on the $Z$ basis to obtain the result $\zeta^a$. Then the conditional state in the systems $\tilde{A}_1$, $\tilde{A}_2$, $A_1''$, $A_1$ and $A_2$ can be reduced to the $X(\theta)$-basis encoding state $\ket{\tilde{\Phi}^{X(\theta)}_{2\mu}}$ of Eq.~\eqref{eq:PhiXtheta0purified} used in the coherent-state MDI-QKD scheme in Sec.~\ref{ssc:secureTMMDI} with a possible simultaneous $\pi$-phase modulation on the optical modes $A_1$ and $A_2$ when $\zeta^a=1$.
\end{lemma}

\begin{proof}

We first summarize the introduced reduction procedure in Fig.~\ref{fig:miMDIComb}, which is independent of Charlie's operation.
\begin{figure}[htbp]
\centering \includegraphics[width=13cm]{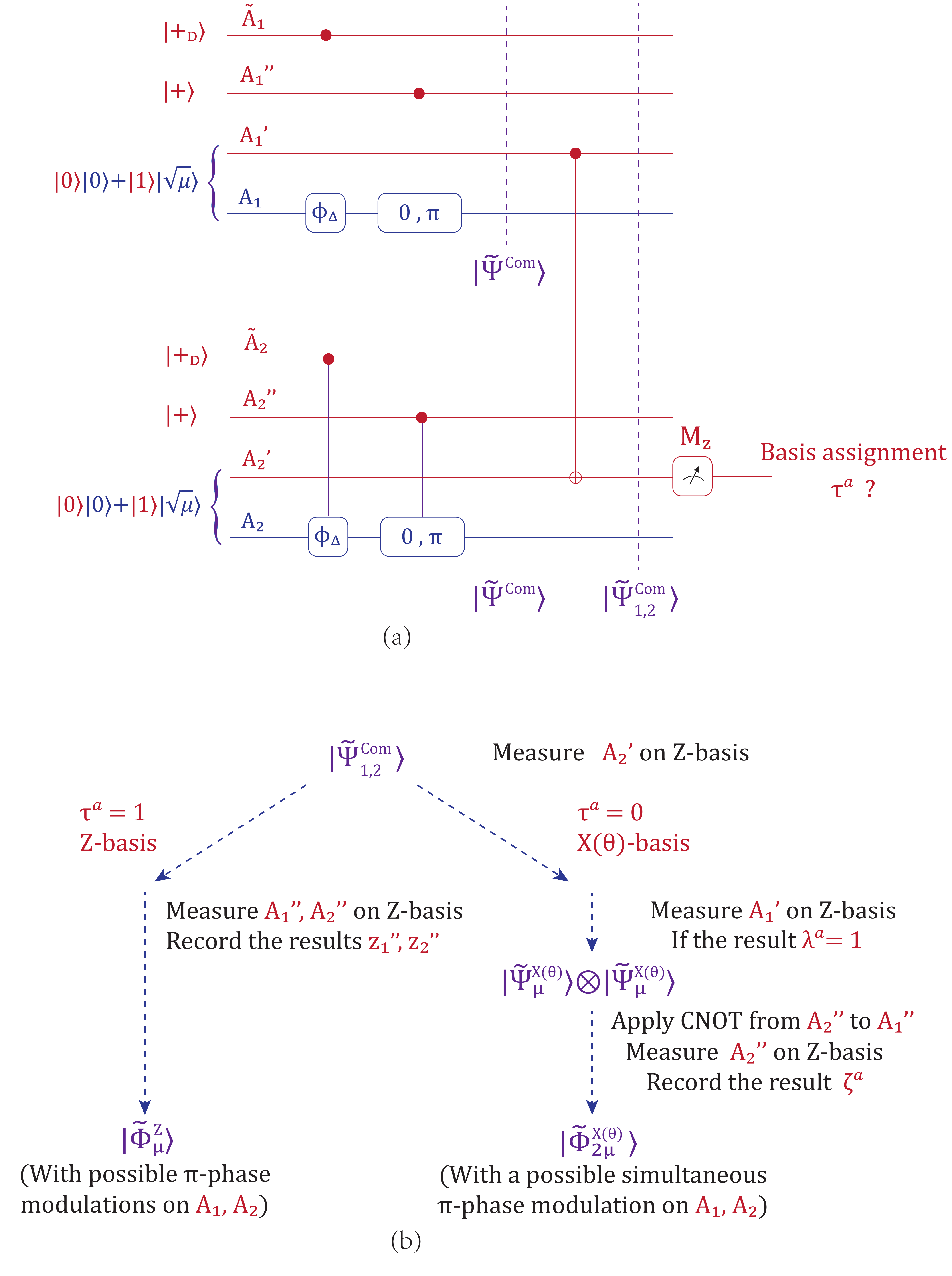}
\caption{(a) Alice's composite encoding process in the MP scheme. In each round, Alice generates a state $\ket{\tilde{\Psi}^{Com}}$ on the one optical mode $A_i$, two qubit ancillaries $A_i'$ and $A_i''$ and a qudit $\tilde{A}_i$ defined in Eq.~\eqref{eq:PsiComPurified2}. (here $i=1$ or $2$.) She transmits the system $A_i$ to Charlie. Based on the detection, she pairs the ancillary systems based on a pairing setting $\vec{\chi}$. For a pair on locations $1$ and $2$, she first performs collective detection on systems $A_1'$ and $A_2'$ to determine the basis. (b) The reduction of Alice's state to the $Z$-basis encoding state $\ket{\tilde{\Phi}^Z_\mu}$ defined in Eq.~\eqref{eq:PhiZ0purified} or the $X(\theta)$-basis encoding state $\ket{\tilde{\Phi}^{X(\theta)}_{2\mu}}$ defined in Eq.~\eqref{eq:PhiXtheta0purified} of coherent-state MDI-QKD. If the measurement result on $A_2'$ is $\tau^a=1$, Alice measures the modes $A_1''$ and $A_2''$ on $Z$-basis and records the results $z_1''$ and $z_2''$. The final state on $A_1'$, $A_1$, $A_2$, $\tilde{A}_1$ and $\tilde{A}_2$ will be the $Z$-basis encoding state $\ket{\tilde{\Phi}^Z_\mu}$ with two possible extra $\pi$-phase modulations on $A_1$ and $A_2$, based on the values of $z_1''$ and $z_2''$, respectively. On the other hand, if $\tau^a=0$, Alice measures $A_1'$ on the $Z$-basis. If the result is $1$, Alice obtains $\ket{\tilde{\Psi}^{X(\theta)}_\mu} \otimes \ket{\tilde{\Psi}^{X(\theta)}_\mu}$ defined in Eq.~\eqref{eq:PsiXmupurified}. She applies a CNOT gate from $A_2''$ to $A_1''$, measure $A_2''$ on $Z$-basis and records the result $\zeta^a$. In this way, she generates $\ket{\tilde{\Phi}^{X(\theta)}_{2\mu}}$ with a possible simultaneous $\pi$-phase modulation on $A_1$ and $A_2$, based on the value of $\zeta^a$.} \label{fig:miMDIComb}
\end{figure}

To perform the basis-assignment measurement, Alice first performs a CNOT gate on qubits $A_1'$ and $A_2'$ and then measure $A_2'$. After the CNOT gate, the joint state on $A_1''A_1'A_1$ and $A_2''A_2'A_2$ becomes
\begin{equation}
\begin{aligned}
\ket{\tilde{\Psi}^{Com}_{1,2}} =& \frac{1}{2} \big( \ket{00}\ket{\tilde{\Psi}^{X(\theta)}_0}\ket{\tilde{\Psi}^{X(\theta)}_0} + \ket{01}\ket{\tilde{\Psi}^{X(\theta)}_0}\ket{\tilde{\Psi}^{X(\theta)}_\mu} \\
+ &\ket{11}\ket{\tilde{\Psi}^{X(\theta)}_\mu}\ket{\tilde{\Psi}^{X(\theta)}_0} + \ket{10}\ket{\tilde{\Psi}^{X(\theta)}_\mu}\ket{\tilde{\Psi}^{X(\theta)}_\mu} \big)_{A_1',A_2';\tilde{A}_1,A_1'',A_1;\tilde{A}_2,A_2'',A_2},
\end{aligned}
\end{equation}
where $\ket{\tilde{\Psi}^{X(\theta)}_\mu}_{\tilde{A}_1 A_1''A_1}$ is defined in Eq.~\eqref{eq:PsiXmupurified}. Now, Alice performs $Z$-basis measurement on the system $A_2'$ to determine the encoding basis.

When the outcome is $\tau^a =1$, the post-selected state is
\begin{equation}
\begin{aligned}
& \frac{1}{\sqrt{2}} \big( \ket{0}\ket{\tilde{\Psi}^{X(\theta)}_0}\ket{\tilde{\Psi}^{X(\theta)}_\mu} + \ket{1}\ket{\tilde{\Psi}^{X(\theta)}_\mu}\ket{\tilde{\Psi}^{X(\theta)}_0} \big)_{A_1';\tilde{A}_1,A_1'',A_1;\tilde{A}_2,A_2'',A_2} \\
= & \frac{1}{2\sqrt{2}D} \sum_{j_1,j_2=0}^{D-1} \ket{j_1}_{\tilde{A}_1}\ket{j_2}_{\tilde{A}_2} \Big[ \ket{0}\left(\ket{0}\ket{0}+\ket{1}\ket{0}\right)\left(\ket{0}\ket{\sqrt{\mu}e^{i\phi^a_{j_2}}} + \ket{1}\ket{\sqrt{\mu}e^{i(\phi^a_{j_2}+\pi)}}\right) \\
\quad & + \ket{1}\left(\ket{0}\ket{\sqrt{\mu}e^{i\phi^a_{j_1}}} + \ket{1}\ket{\sqrt{\mu}e^{i(\phi^a_{j_1}+\pi)}}\right)\left(\ket{0}\ket{0}+\ket{1}\ket{0}\right) \Big]_{A_1';A_1'',A_1;A_2'',A_2} \\
= & \frac{1}{2} \Big[ \ket{00} \ket{\tilde{\Phi}^Z_\mu} + \ket{01} \hat{U}(\pi)_{A_2} \ket{\tilde{\Phi}^Z_\mu}  + \ket{10} \hat{U}(\pi)_{A_1} \ket{\tilde{\Phi}^Z_\mu} + \ket{11} (\hat{U}(\pi)\otimes \hat{U}(\pi))_{A_1,A_2} \ket{\tilde{\Phi}^Z_\mu} \Big]_{A_1';A_1'',A_1;A_2'',A_2},
\end{aligned}
\end{equation}
where $\ket{\tilde{\Phi}^Z_\mu}_{\tilde{A}_1\tilde{A}_2,A_1',A_1A_2}$ is defined in Eq.~\eqref{eq:PhiZ0purified}.
Now, suppose Alice measures the systems $A_1''$ and $A_2''$ on $Z$-basis and obtains the results $z_1''$ and $z_2''$. The post-selected state on the systems $A_1',A_1,A_2$ will become the $Z$-basis encoding state $\ket{\tilde{\Phi}^Z_\mu}$ of the coherent-state MDI-QKD scheme with a possible $\pi$-phase modulation on the optical modes $A_1$ and $A_2$, depending on the values $z_1''$ and $z_2''$. This reduction process is illustrated in Fig.~\ref{fig:miMDIreduceZ}.

\begin{figure}[htbp]
\centering \includegraphics[width=18cm]{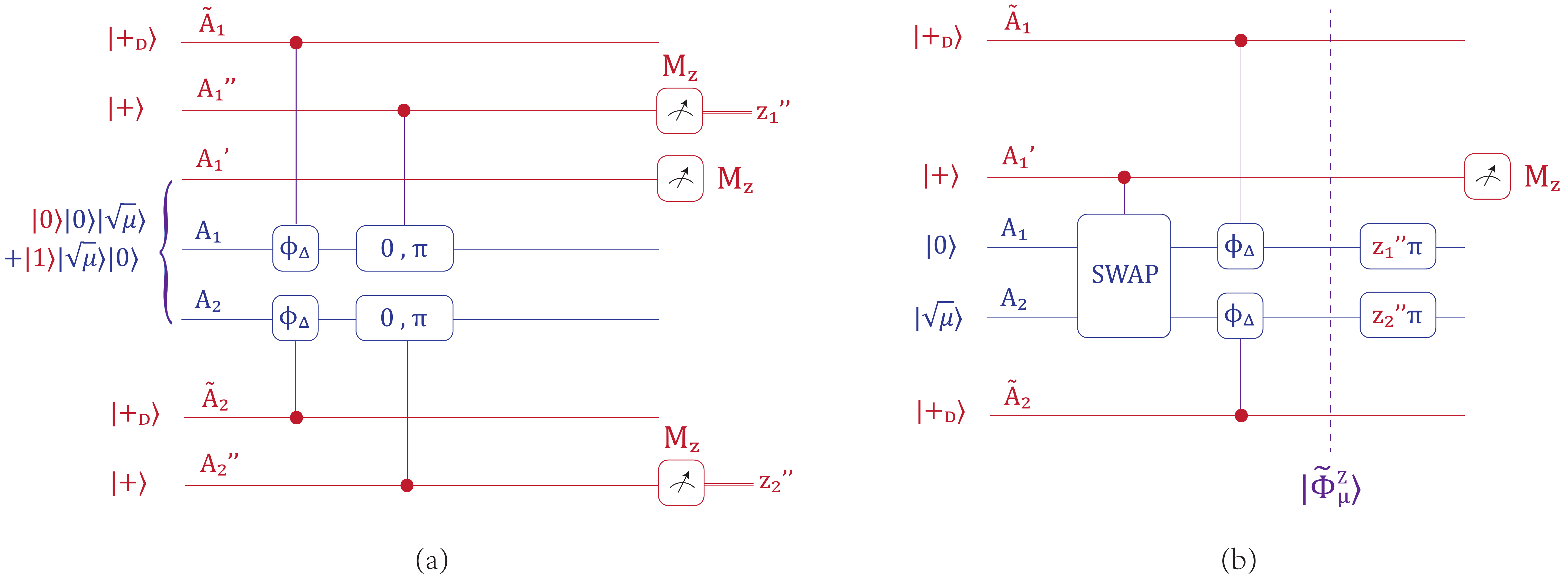}
\caption{Reduction of the encoding state of the MP scheme in Fig.~\ref{fig:miMDIComb}(b) to the $Z$-basis encoding state $\ket{\tilde{\Phi}^Z_\mu}$ in Eq.~\eqref{eq:PhiZ0purified} and Fig.~\ref{fig:coherentMDI}(a). After we measure the ancillary qubits $A_1''$ and $A_2''$ on $Z$-basis and obtain the results $z_1''$ and $z_2''$, we will project the remaining encoding state to $\ket{\tilde{\Phi}^Z_\mu}$ with extra phase modulations $z_1'' \pi$ and $z_2'' \pi$ on the optical modes $A_1$ and $A_2$, respectively. The extra $\pi$-phase modulation will not affect the security and the performance of the $Z$-basis coherent-state MDI-QKD scheme.} \label{fig:miMDIreduceZ}
\end{figure}

We remark that, for the single-photon component, a $\pi$-phase modulation on the optical modes $\tilde{A}_1$ and $\tilde{A}_2$ can be regarded as a $Z$-axis rotation on the ancillary qubit $A'$ shown in Fig.~\ref{fig:singleMDIreduce}. This rotation will not affect the $Z$-basis measurement result. Therefore, the $\pi$-phase modulation on $A_1$ and $A_2$ will not affect the security of the $Z$-basis key generation.


On the other hand, when $\tau^a = 0$, the post-selected state is
\begin{equation}
\begin{aligned}
& \frac{1}{\sqrt{2}} \big( \ket{0}\ket{\tilde{\Psi}^{X(\theta)}_0}\ket{\tilde{\Psi}^X_0} + \ket{1}\ket{\tilde{\Psi}^{X(\theta)}_\mu}\ket{\tilde{\Psi}^{X(\theta)}_\mu} \big)_{A_1';A_1'',A_1;A_2'',A_2}.
\end{aligned}
\end{equation}
In this case, Alice further measures $A_1'$ on the $Z$ basis and keeps her signals if the result is $1$. The post-selected state is then $\ket{\tilde{\Psi}^{X(\theta)}_\mu}\otimes \ket{\tilde{\Psi}^{X(\theta)}_\mu}_{\tilde{A}_1A_1''A_1;\tilde{A}_2A_2''A_2}$, where $\ket{\tilde{\Psi}^{X(\theta)}_\mu}$ is defined in Eq.~\eqref{eq:PsiXmupurified}. We remark that, if the measurement result on $A_1'$ is $0$, then the state will not be used for key generation; instead, it can be used for the parameter estimation.

To reduce the state $\ket{\tilde{\Psi}^{X(\theta)}_\mu}\otimes \ket{\tilde{\Psi}^{X(\theta)}_\mu}_{\tilde{A}_1A_1''A_1;\tilde{A}_2A_2''A_2}$ to the $X(\theta)$-basis encoding state $\ket{\tilde{\Phi}^{X(\theta)}_{2\mu}}$ of coherent-state MDI-QKD scheme defined in Eq.~\eqref{eq:PhiXtheta0purified}, Alice applies a CNOT operation from the system $A_2''$ to $A_1''$, measures $A_2''$ on the $Z$-basis, and records the result $\zeta^a$. As is shown in Fig.~\ref{fig:miMDIreduceXtheta}, we can equivalently move the $Z$-basis measurement on $A_2''$ forward to the beginning. In this way, we can reduce the encoding state to $\ket{\tilde{\Phi}^{X(\theta)}_{2\mu}}$ with a simultaneous $\pi$-phase modulation on the optical modes $A_1$ and $A_2$. For the $X(\theta)$-basis coherent-state MDI-QKD scheme, such $\pi$-phase modulation will not cause any effect.

\begin{figure}[htbp]
\centering \includegraphics[width=18cm]{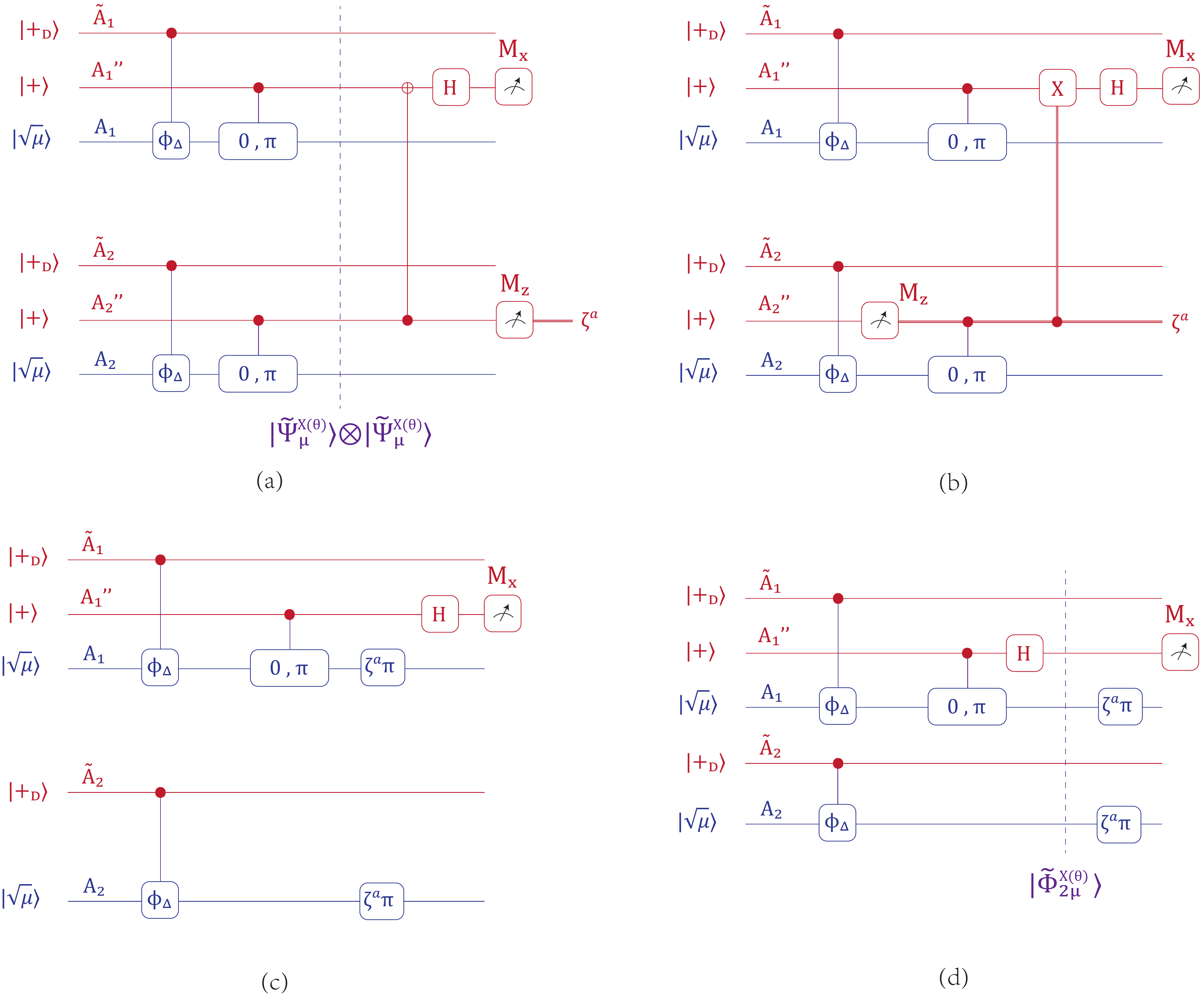}
\caption{Reduction of encoding in Fig.~\ref{fig:miMDIComb}(b) to the $X(\theta)$-basis encoding state $\ket{\tilde{\Phi}^{X(\theta)}_{2\mu}}$ in Eq.~\eqref{eq:PhiXtheta0purified} and Fig.~\ref{fig:coherentMDI}(b). From (a) to (b), we advance the $Z$-basis measurement on $A_2''$. This entails introducing an extra random number $\zeta^a$ in the phase modulations of $A_1$ and $A_2$, as shown in (c). We remark that, a simultaneous phase modulation will not affect the security and performance of the $X(\theta)$-basis coherent-state MDI-QKD.} \label{fig:miMDIreduceXtheta}
\end{figure}

\end{proof}

Based on Lemma~\ref{lem:fixedMPreduce}, we introduced the fixed-pairing MP scheme in Box~\ref{box:fixedMP}, which is also illustrated in Fig.~\ref{fig:miMDIdia}.

\begin{Boxes}{Fixed-pairing MP scheme}{fixedMP}
\begin{enumerate}
\item \label{StepPrep}
State preparation: In each round, Alice prepares the composite state $\ket{\tilde{\Psi}^{Com}}$ defined in Eq.~\eqref{eq:PsiComPurified2} on the optical mode $A_i$, two ancillary qubits $A_i'$ and $A_i''$ and an ancillary qudit $\tilde{A}_i$. In a similar manner, Bob prepares $\ket{\tilde{\Psi}^{Com}}$.

\item \label{StepMeasure}
Measurement: Alice and Bob send modes $A_i$ and $B_i$ to Charlie, who is supposed to perform the single-photon interference measurement. Charlie announces the clicks of the detectors $L$ and/or $R$.

Alice and Bob perform the above steps over many rounds. Afterwards, they perform the following data post-processing.

\item \label{StepPairing}
Mode Pairing: Based on a pre-set pairing setting $\vec{\chi}\in\mathcal{X}$, Alice and Bob pair their locations (i.e., group each two locations together).

\item \label{StepSifting}
Basis Sifting: For each pair on locations $i$ and $j$, Alice performs the basis-assignment measurement on $A_i'$ and $A_j'$ to determine the basis, as shown in Fig.~\ref{fig:miMDIComb}. If the result $\tau^a = 0$, Alice further measures $A_i'$ on the $Z$-basis and records the result $\lambda^a$. Alice assigns the basis of the pair as $Z$ if $\tau^a=1$, $X$ if $\tau^a=0$ and $\lambda^a=1$ and `0' if $\tau^a=\lambda^a=0$. Bob assigns the basis in the same way. Alice and Bob announce their bases. If the announced bases are the same, $X,X$ or $Z,Z$, they keep their signals.

\item \label{StepKeyMapping}
Key mapping: For each $Z$-pair on locations $i$ and $j$, Alice discards qubits $A_i''$ and $A_j''$ and keeps $A_i'$ for later usage. For each $X$-pair, Alice applies CNOT gate from qubit $A_j''$ to $A_i''$, discards qubit $A_j''$, keeps qubit $A_i''$. Alice then performs global measurement $M(k,\theta)$ defined in Eq.~\eqref{eq:Mktheta} on qudits $\tilde{A}_i$ and $\tilde{A}_j$ to obtain the photon number $k^a$ and the relative phase $\theta^a$. Bob operates similarly. For all $X$ pairs, Alice and Bob further announce $\theta^a$ and $\theta^b$. If $\theta^a=\theta^b$, they keep the $X$-pair data; if $\theta^a -\theta^b = \pi$, they keep the $X$-pair data with Bob applying the $Z$ gate on his left qubit. Moreover, if Charlie's announcement on locations $i$ and $j$ is $(L,R)$ or $(R,L)$, Bob applies $Z$ gate on the qubit system $B_i''$ for the $X$-pairs. Alice and Bob then measure the ancillary qubits $A''$ and $B''$ of all the $X$-pairs and the ancillary qubits $A'$ and $B'$ of all the $Z$-pairs on the $X$- and $Z$-bases, respectively, to obtain the raw key bits $\kappa^a$ and $\kappa^b$.

\item \label{StepParaEst}
Parameter estimation: Alice and Bob estimate the fraction of clicked single-photon signals $q_{(1,1)}$ and the corresponding phase-error rate $e^{X(\theta)}_{(1,1)}$ in the $Z$-pairs, with Alice and Bob both emitting a single photon at locations $i$ and $j$. 

\item \label{StepClassical}
Key distillation: Alice and Bob reconcile the key string to $\kappa^a$ and perform privacy amplification using universal-2 hashing matrices. The sizes of hashing matrices are determined by the estimated single-photon yield $q_{(1,1)}$ and the $X(\theta)$-basis error rate $e^{X(\theta)}_{(1,1)}$ according to Eq.~\eqref{eq:keyrateDecoyMDI}.
\end{enumerate}
\end{Boxes}

\begin{figure}[htbp]
\includegraphics[width=16cm]{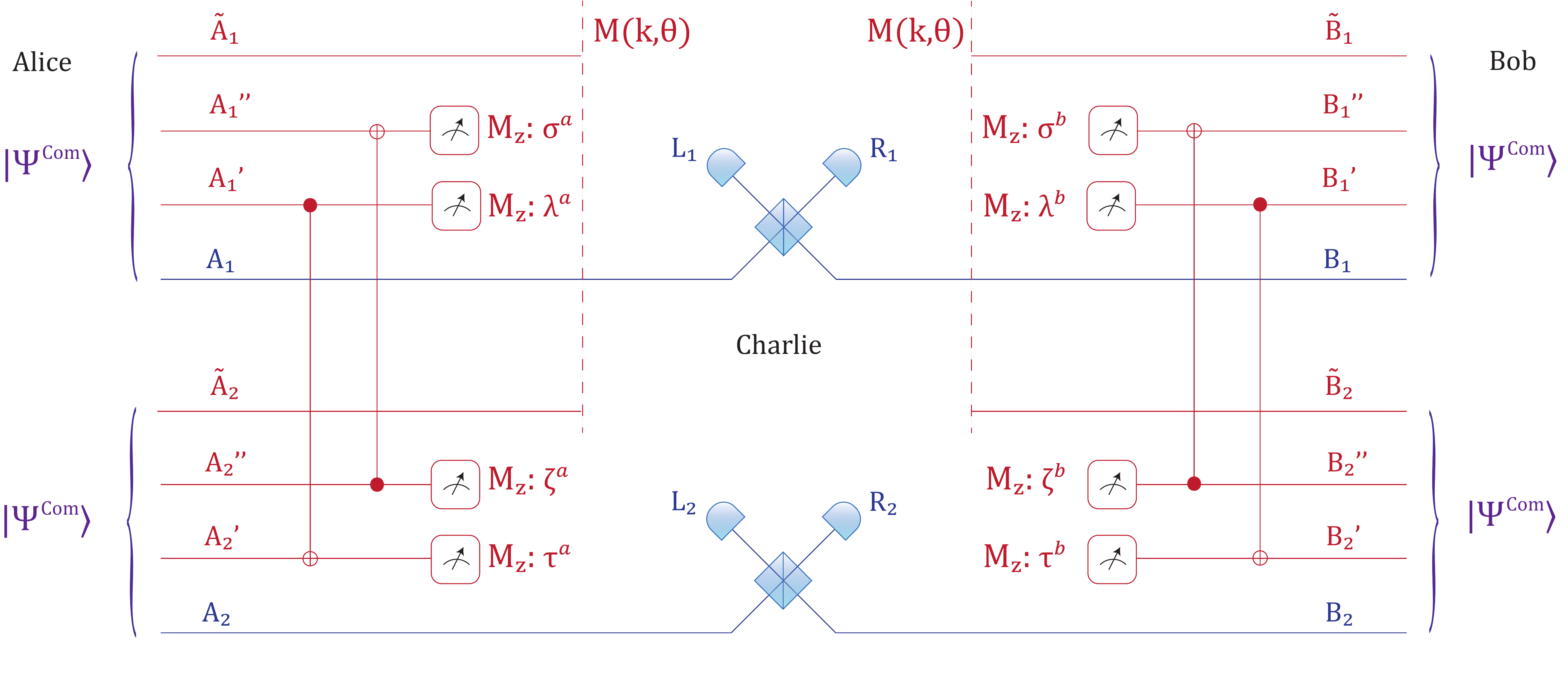}
\caption{Diagram of the entanglement-based MP scheme with a fixed pairing location. Here we take the pair on location $1$ and $2$ for example. Alice and Bob both generate a composite state $\ket{\tilde{\Psi}^{Com}}$ on an optical mode, two ancillary qubits and a ancillary qudit in each round and send the optical modes to Charlie. After Charlie's detection announcement, they perform control operations on the ancillary qubits at the predetermined paired locations and perform measurements on them. They then perform $Z$-basis measurements on all the ancillary qubits $A'_1 (B'_1)$, $A''_1(B''_1)$, $A'_2(B'_2)$ and $A''_2(B''_2)$ to obtain the measurement results $\sigma^a(\sigma^b)$, $\lambda^a(\lambda^b)$, $\zeta^a(\zeta^b)$, $\tau^a(\tau^b)$. The usage of these results is listed in Fig.~\ref{fig:miMDICombReduce}(d).} \label{fig:miMDIdia}
\end{figure}

We introduce encoding redundancies in the coherent-state MDI-QKD scheme, so that the proposed MP scheme owns i.i.d.~state preparation process in each round. Due to this i.i.d.~encoding property, we can see that the MP scheme can be reduced to the coherent-state scheme based on Lemma~\ref{lem:fixedMPreduce} \emph{for any given pairing setting $\vec{\chi}$}. We now clarify this lemma.

\begin{lemma}[Security of the MP scheme under any fixed pairing setting]\label{lem:arbitraryPairing}
In the MP scheme, for any given pairing setting $\vec{\chi}$, suppose that Alice and Bob perform the mode-pairing operations on the ancillary qubits $\{A_i',A_i'',B_i',B_i''\}_i$ in Lemma~\ref{lem:fixedMPreduce} and the pairwise measurement $M(k,\theta)$ defined in Eq.~\eqref{eq:Mktheta} on the ancillary qudits $\{\tilde{A}_i,\tilde{B}_i\}_i$. When the discrete phase number $D\to \infty$, we have,
\begin{enumerate}
\item
(Photon-number Poisson statistics) For the $(i,j)$-location pairs with overall intensities $\mu^a:= \mu^a_i + \mu^a_j, \mu^b:=\mu^b_i + \mu^b_j$, the photon-number-measurement results $(k^a,k^b)$ of Alice and Bob's pairwise measurement $M(k,\theta)$ follow Poisson distributions,
\begin{equation}
\Pr(k^a,k^b) = e^{-(\mu^a+\mu^b)} \frac{(\mu^a)^{(k^a)} (\mu^b)^{(k^b)}}{k^a! k^b!}.
\end{equation}
\item
(Independence of the photon-number states to the intensity) The resultant state after the joint measurement $M(k,\theta)$ is independent of the intensity values $\mu^a$ and $\mu^b$.
\item
(Basis-independence of the single-photon pairs) For the location pairs with the overall photon-number measurement result $k^a=k^b=1$, the $X(\theta)$-basis error rate can be used as a fair estimation of the $Z$-basis phase-error rate.
\end{enumerate}
As a result, with any fixed pairing $\vec{\chi}$, the $\vec{\chi}$-fixed-pairing scheme is able to generate secure key strings with the tagging key-rate formula, Eq.~\eqref{eq:keyrateDecoyMDI}.
\end{lemma}

\begin{proof}
From Lemma~\ref{lem:fixedMPreduce} we see that, under any pairing setting $\vec{\chi}$, with the given basis-assignment and post-selection condition, Alice and Bob can generate the encoding states $\ket{\tilde{\Phi}^Z_\mu}$ and $\ket{\tilde{\Phi}^{X(\theta)}_{2\mu}}$ in the coherent-state MDI-QKD scheme. Furthermore, if Alice and Bob perform the pairwise measurement $M(k,\theta)$ defined in Eq.~\eqref{eq:Mktheta} on the remaining state, from Lemma~\ref{lem:Coh2Sin} we see that, the results $(k^a,k^b)$ of $M(k,\theta)$ defined in Eq.~\eqref{eq:Mktheta} obey the Poisson distribution; after the pairwise measurement $M(k,\theta)$, the post-selected photon-number state is independent of the intensity setting $\mu^a$ and $\mu^b$; the post-selected encoding states with $k^a=k^b=1$ will become to the ones in the single-photon MDI-QKD scheme. If we post-select the corresponding encoding states, the security of this scheme will be equivalent to the single-photon MDI-QKD scheme, where the encoding state is basis-independent and the security proof can be easily done based on the complementarity.

If we set the photon numbers $k^a$ and $k^b$ to be the tag on the MP scheme, the key rate of the MP scheme also follows the traditional tagging key-rate formula in Eq.~\eqref{eq:keyrateDecoyMDI}.
\end{proof}

\subsection{Free-pairing MP scheme} \label{ssc:MPfree}
The freedom of choosing the pairing setting, $\vec{\chi}$, in the fixed-pairing MP scheme inspires us to consider a more flexible variant, where $\vec{\chi}$ is determined based on Charlie's announcements $\vec{C}$. We call this variant the free-pairing MP scheme, in Definition \ref{def:freeMP}. Note that the pairing setting, $\vec{\chi}\in\mc{X}$ in Definition~\ref{def:pairing}, is chosen for all the $N$ emitted signals, even for the locations without successful detection. This is due to the usage of Lemma~\ref{lem:arbitraryPairing} in the latter discussion.

\begin{definition}[Pairing Strategy $T$]\label{def:pairingstra}
For $N$ location indexes $i=1,2,...,N$ where $N$ is an even number, a pairing strategy, $T:\mc{C}\to\mc{X}$, is defined to be a map from Charlie's announcement $\vec{C}\in\mc{C}$ in Definition~\ref{def:Cannounce} to the pairing setting $\vec{\chi}\in\mc{X}$ in Definition~\ref{def:pairing}. The possible pairing strategies form a finite set $\mc{T}$.
\end{definition}

\begin{definition}[$T$-Free-Pairing MP] \label{def:freeMP}
In a MP MDI-QKD scheme, if pairing setting $\vec{\chi}$ depends on Charlie's announcement, then we call it the free-pairing MP scheme. A free-pairing MP scheme with pairing strategy $T\in\mc{T}$ is called the $T$-free-pairing MP scheme.
\end{definition}

We remark that, the pairing setting $\vec{\chi}$ in Definition~\ref{def:pairing} and the pairing strategy $T$ in Definition~\ref{def:pairingstra} are defined for \emph{all the locations}, even for the ones with unsuccessful detections. In a pairing strategy $T$, we need to pair all the locations based on Charlie’s announcement $\vec{C}$ and do not perform any post-selection. The pairing definition in this way is helpful to build up a connection between the fixed- and free-pairing schemes, so that the loss sifting in the free-pairing schemes will be handled following the same way as the fixed-pairing schemes. 
For example, suppose there are 4 rounds in the mode-pairing scheme and Charlie announces successful detections on locations 1 and 4, then based on a simple pairing strategy, Alice and Bob pair the locations as follows: $(1,4)$, $(2,3)$. The “lost” pair $(2,3)$ will be used for the parameter estimation. 

We also note that, the $\vec{\chi}$-fixed-pairing scheme can be regarded as a special case of the $T$-free-pairing scheme, where the pairing strategy is a constant function, $T(\vec{C})\equiv\vec{\chi}$, $\forall \vec{C}\in\mc{C}$. Following the fixed-pairing MP scheme in Box~\ref{box:fixedMP}, we present the procedures of the free-pairing MP scheme with a pairing strategy $T$ in Box~\ref{box:freeMP}.

\begin{Boxes}{Free-pairing MP scheme}{freeMP}

\begin{enumerate}
\item 
State preparation: Same as the fixed-pairing MP scheme in Box~\ref{box:fixedMP}.

\item 
Measurement: Same as the fixed-pairing MP scheme in Box~\ref{box:fixedMP}.

Alice and Bob perform the above steps over many rounds. After that, they perform the following data post-processing steps.

\item
Mode Pairing: Charlie announces $\vec{C}\in\mathcal{C}$ of Definition~\ref{def:Cannounce}, based on which Alice and Bob determine a pairing setting $\vec{\chi}=T(\vec{C})\in\mathcal{X}$ of Definition~\ref{def:pairing} using a predetermined pairing strategy $T\in\mc{T}$ of Definition~\ref{def:pairingstra}.

\item 
Basis Sifting: Same as the fixed-pairing MP scheme in Box~\ref{box:fixedMP}.

\item
Key mapping: Same as the fixed-pairing MP scheme in Box~\ref{box:fixedMP}.

\item 
Parameter estimation: Same as the fixed-pairing MP scheme in Box~\ref{box:fixedMP}.

\item 
Key distillation: Same as the fixed-pairing MP scheme in Box~\ref{box:fixedMP}.
\end{enumerate}
\end{Boxes}

In the free-pairing scheme, Charlie announces $\vec{C}$ and then Alice and Bob determine $\vec{\chi}$ based on $\vec{C}$. In a way, Charlie can manipulate the choice of paring settings. Here, we prove the security of the free-pairing scheme by taking the advantage of the arbitrariness of the pairing setting $\vec{\chi}$ in the fixed-pairing scheme.

\begin{figure}[htbp!]
\centering\includegraphics[width=16cm]{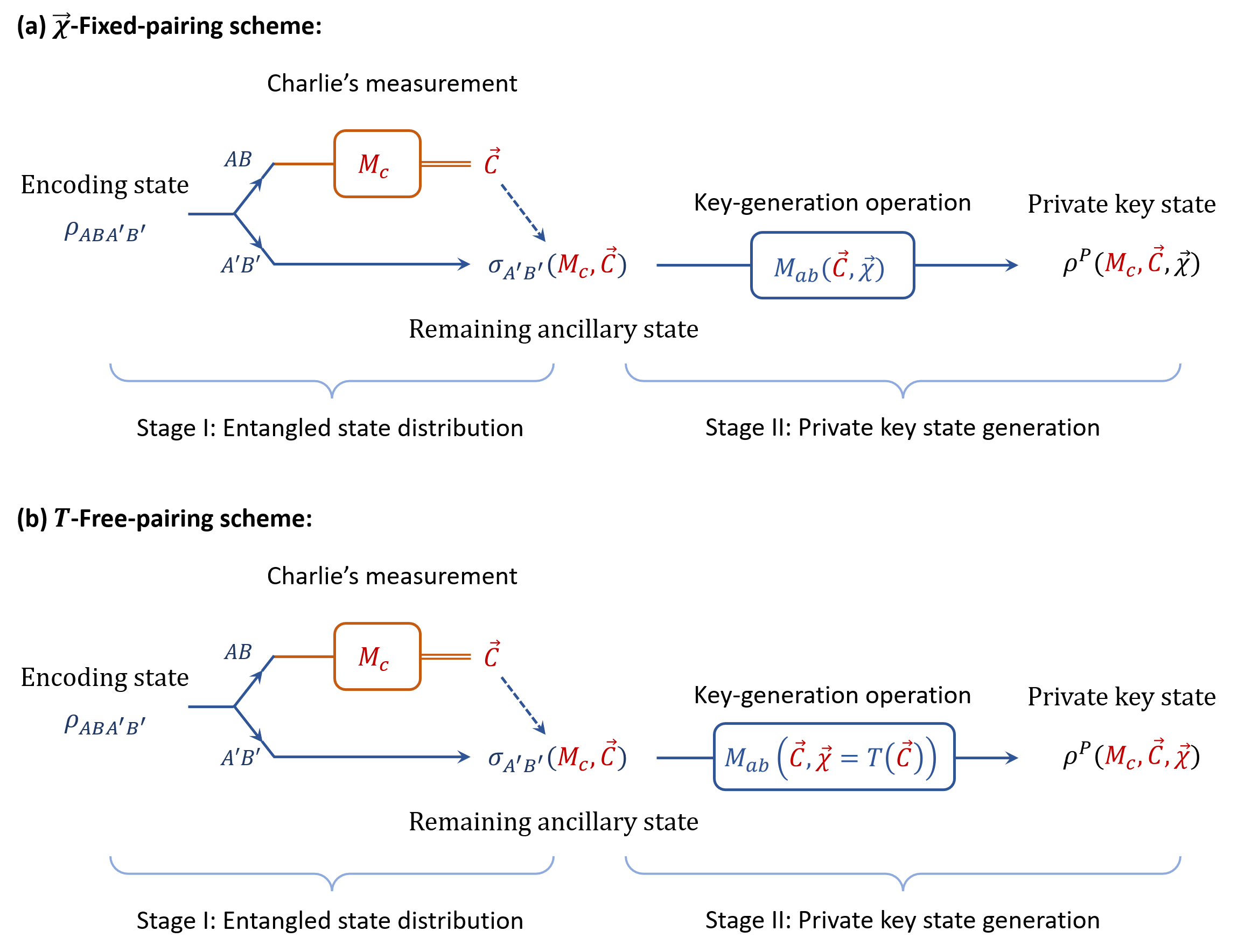}
\caption{Flowcharts of (a) $\vec{\chi}$-fixed-pairing scheme and (b) $T$-free-pairing scheme. Here, the key-generation operation $M_{ab}$ correponds to Steps~\ref{StepSifting}-\ref{StepClassical} in Box~\ref{box:fixedMP} and \ref{box:freeMP}. The major difference of the fixed-pairing and free-pairing scheme lies in how the key-generation operation $M_{ab}$ depends on the pairing setting $\vec{\chi}$ and how the pairing setting is determined by Charlie's announcement $\vec{C}$.} \label{fig:FixedVSFree}
\end{figure}

Before we examine the free-pairing scheme, we first briefly review the process of fixed-pairing scheme and revisit what we proved in Lemma~\ref{lem:arbitraryPairing}. A simple diagram of $\vec{\chi}$-fixed-pairing scheme is shown in Fig.~\ref{fig:FixedVSFree}(a), as a special case of Fig.~\ref{fig:MDIMcMab}. Alice and Bob first generate $N=2n$ rounds of the encoding state $\rho_{ABA'B'}$ as the pure state $\ket{\tilde{\Psi}^{Com}}_{Alice} \otimes \ket{\tilde{\Psi}^{Com}}_{Bob}$ defined in Eq.~\eqref{eq:PsiComPurified2}. They send all the optical signals of systems $A$ and $B$ to Charlie, who performs unknown measurement $M_c$ on them and announces the result $\vec{C}$ defined in Definition~\ref{def:Cannounce}. Note that $M_c$ refers to all Charlie's operations, including measurement, result processing, and announcement strategy.

\begin{figure}[htbp!]
\centering\includegraphics[width=8cm]{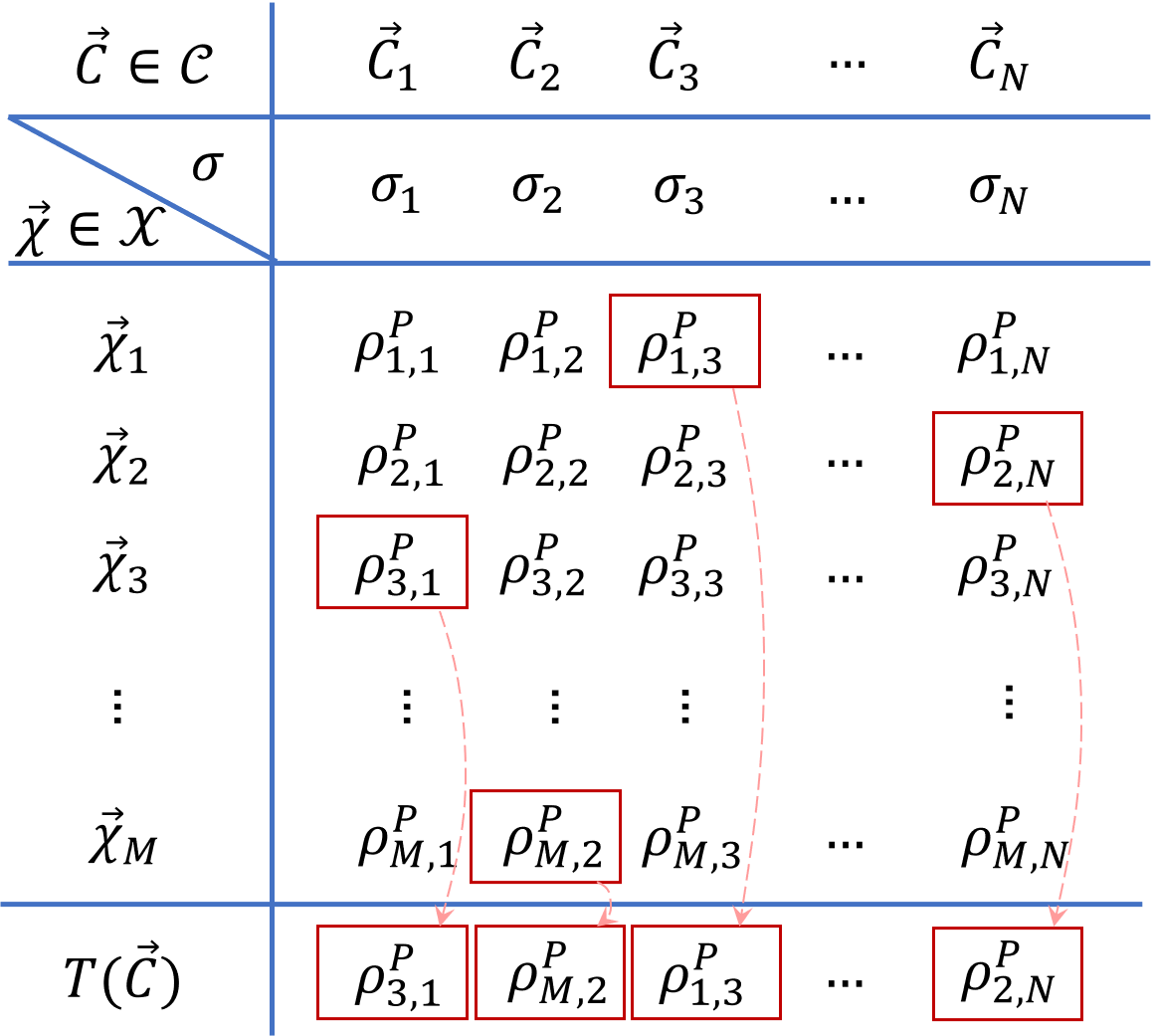}
\caption{Security of the free-pairing MP scheme. Given Charlie's measurement $M_c$, we list all possible ancillary states $\sigma_{\vec{C}}:=\sigma(\vec{C})$ and the private key states $\rho^P_{\vec{\chi},\vec{C}}:=\rho^P(\vec{C},\vec{\chi})$ with different announcements $\vec{C}$ and pairing settings $\vec{\chi}$. Adopting a free pairing strategy $T(\vec{C})$ in the key-generation procedure is essentially to pick out a private key state $\rho^P$ for each announcement $\vec{C}$ according to $\vec{\chi}=T(\vec{C})$. For the example shown in the figure, connected by dashed red arrow lines, $T(\vec{C}_1)=\vec{\chi}_3$, $T(\vec{C}_2)=\vec{\chi}_M$, $T(\vec{C}_3)=\vec{\chi}_1$, and $T(\vec{C}_N)=\vec{\chi}_2$.} \label{fig:FreePairSecure}
\end{figure}

Based on $M_c$ and $\vec{C}$, Alice and Bob's remaining ancillary state $\sigma(M_c,\vec{C})$ at hand is supposed to be entangled if Charlie honestly follows the designed measurement. Afterwards, in the $\vec{\chi}$-fixed-pairing scheme, Alice and Bob perform the key-generation operation $M_{ab}(\vec{C},\vec{\chi})$, including basis-sifting, key mapping, parameter estimation, and key distillation (i.e., Steps~\ref{StepSifting}-\ref{StepClassical} in Box~\ref{box:fixedMP}) based on a pre-determined pairing setting $\vec{\chi}$ in Definition~\ref{def:pairing} and Charlie's announcement $\vec{C}$. Following Lemma~\ref{lem:arbitraryPairing}, we have proven that, for \textit{all possible Charlie's measurement operation $M_c$ and announcement result $\vec{C}$}, the final state $\rho^P(M_c,\vec{C},\vec{\chi})$ will be private for every $\vec{\chi}\in\mc{X}$.

Now we are going to show the security of a free-pairing scheme.

\begin{theorem}[Security of the free-pairing MP scheme]
For any pairing strategy $T\in\mc{T}$ in Definition~\ref{def:pairingstra}, under all Charlie's measurement operations $M_c$ and announced results $\vec{C}$, in the $T$-free-pairing scheme, Alice and Bob is able to generate private key state $\rho^P(M_c,\vec{C},\vec{\chi})$ with the tagging key-rate formula Eq.~\eqref{eq:keyrateDecoyMDI}.
\end{theorem}

\begin{proof}
We are going to prove the following three statements,
\begin{enumerate}
\item
In any fixed-pairing scheme where $T(\vec{C})\equiv\vec{\chi}\in\mc{X}$, for any possible Charlie's measurement operation $M_c$ and announcement result $\vec{C}$, after the key-generation operation $M_{ab}$ (i.e., Steps~\ref{StepSifting}-\ref{StepClassical} in Box~\ref{box:fixedMP}) , the final state $\rho^P$ is private with the tagging key-rate formula Eq.~\eqref{eq:keyrateDecoyMDI}.
	
\item
All possible Charlie's measurements $M_c$ in a $T$-free-pairing scheme can be realized in every fixed-pairing scheme~\footnote{That is, the set of possible $M_c$ in the free-pairing scheme is contained in that of fixed-pairing scheme. Strictly speaking, the sets of possible $M_c$ in the two cases are the same, since it only depends on the state Alice and Bob prepare, which is independent of $T$-free-pairing or $\chi$-fixed-pairing.}.

\item
Given Charlie's measurement $M_c$, the final state generated by a $T$-free-pairing scheme is equal to the private state $\rho^P(M_c,\vec{C},\vec{\chi})$ generated by one of the $\vec{\chi}$-fixed-pairing schemes.
\end{enumerate}

The first statement is a direct result of Lemma~\ref{lem:arbitraryPairing}.

To prove the second statement, we notice that the encoding state prepared in both free-pairing and fixed-pairing schemes are the same. Then, Charlie receives exactly the same quantum states in the two schemes. The only difference is that Charlie receives different classical information of the pairing strategies $T\in \mc{T}$. Recall that the fixed pairing $\vec{\chi}$ is a special pairing strategy in $\mc{T}$. The second statement becomes the following question: for any Charlie's measurements $M_c$ in a $T$-free-pairing scheme, whether Charlie can perform the same $M_c$ in a fixed-pairing scheme. Since the pairing strategy $T$ is independent of the quantum states sent to Charlie, the answer is yes and hence the second statement holds.


Now, we prove the third statement. For a given Charlie's measurement $M_c$, we can list the resultant ancillary states $\sigma(\vec{C})$ at Alice and Bob's hand corresponding to different measurement results $\vec{C}$, shown in Fig.~\ref{fig:FreePairSecure}. Furthermore, we list all the possible resultant states $\rho^P$ after the key-generation operation $M_{ab}$ correspond to different fixed pairing settings $\vec{\chi}\in\mc{X}$ in the fixed-pairing MP schemes, shown in Fig.~\ref{fig:FreePairSecure}.

In the $T$-free-pairing scheme in Definition~\ref{def:freeMP}, Alice and Bob choose the pairing setting $\vec{\chi}= T(\vec{C})$ during the key generation process $M_{ab}$. For a pairing strategy $T\in\mc{T}$ in Definition~\ref{def:pairingstra}, we have $T(\vec{C})\in \mc{X}$. That is, in a $T$-free-pairing scheme, the resultant state after the key-generation operation will be the same as a private key state from one of the fixed-pairing schemes, as shown in Fig.~\ref{fig:FreePairSecure}.

Combining the three statements, we conclude that: under any possible Charlie's measurement operation $M_c$ and announcement $\vec{C}$,
a $T$-free-pairing scheme results in a private key state with the key length determined by the tagging key-rate formula Eq.~\eqref{eq:keyrateDecoyMDI}.
\end{proof}

Having said that Alice and Bob can choose the pairing strategy $T$ freely even based on Charlie's announcement $\vec{C}$,
the pairing setting $\vec{\chi}$ cannot depend on the measurement results of Alice's and Bob's ancillary states. This is due to the fact that in our security proof, the pairing setting $\vec{\chi}$ is determined before any operation on the ancillary states.



\subsection{Prepare-and-measure MP scheme} \label{ssc:PrepareandMeasure}
As the final step, we reduce the entanglement-based free-pairing MP scheme in Box~\ref{box:freeMP} to the prepare-and-measure one presented in the main text. In the following discussions, we focus on the reduction on Alice's side. The reduction of Bob's side follows the same manner. Also, for the simplicity of discussion, we take a paired location $1$ and $2$ as an example.

First, we remove all the pairwise joint measurements $M(k,\theta)$ defined in Eq.~\eqref{eq:Mktheta} used to read out the global photon number and the relative phase on each pair of optical modes. Instead, we assume Alice and Bob measure the (discrete) phase of each pulse directly. In this way, the relative phase can be obtained by classical post-processing. Then, the global photon number is lost due to the uncertain relationship between phase and photon number measurements. In practice, Alice and Bob can use the decoy-state method to estimate the fraction of the single-photon component. Furthermore, we assume the random phase is measured at the beginning, which is then equivalent to a random phase encoding on the emitted optical pulses. In this way, we remove the two ancillary qudits $\tilde{A}_1$ and $\tilde{A}_2$ in the encoding process.

After removing the phase encoding ancillary systems, we are now going to remove the four ancillary qubits $A_1'$, $A_1''$, $A_2'$ and $A_2''$. To do this, we move all the measurements on the ancillary qubits mentioned in Lemma~\ref{lem:fixedMPreduce} forward.
In Fig.~\ref{fig:miMDICombReduce}(a), we list the complete basis and key-assignment procedure of the entanglement-based MP scheme. The final measurement results $\sigma^a,\lambda^a,\zeta^a$ and $\tau^a$ are used in different ways, as shown in Fig.~\ref{fig:miMDICombReduce}(d). Note that the $Z$-basis measurements commute with all $CNOT$ gates. Therefore, moving all measurements before the control gates between locations $1$ and $2$, one can equivalently obtain the measurement results $\sigma^a,\lambda^a,\zeta^a$, and $\tau^a$ by classical post-processing, as shown in Fig.~\ref{fig:miMDICombReduce}(b).

\begin{figure}[htbp]
\includegraphics[width=18cm]{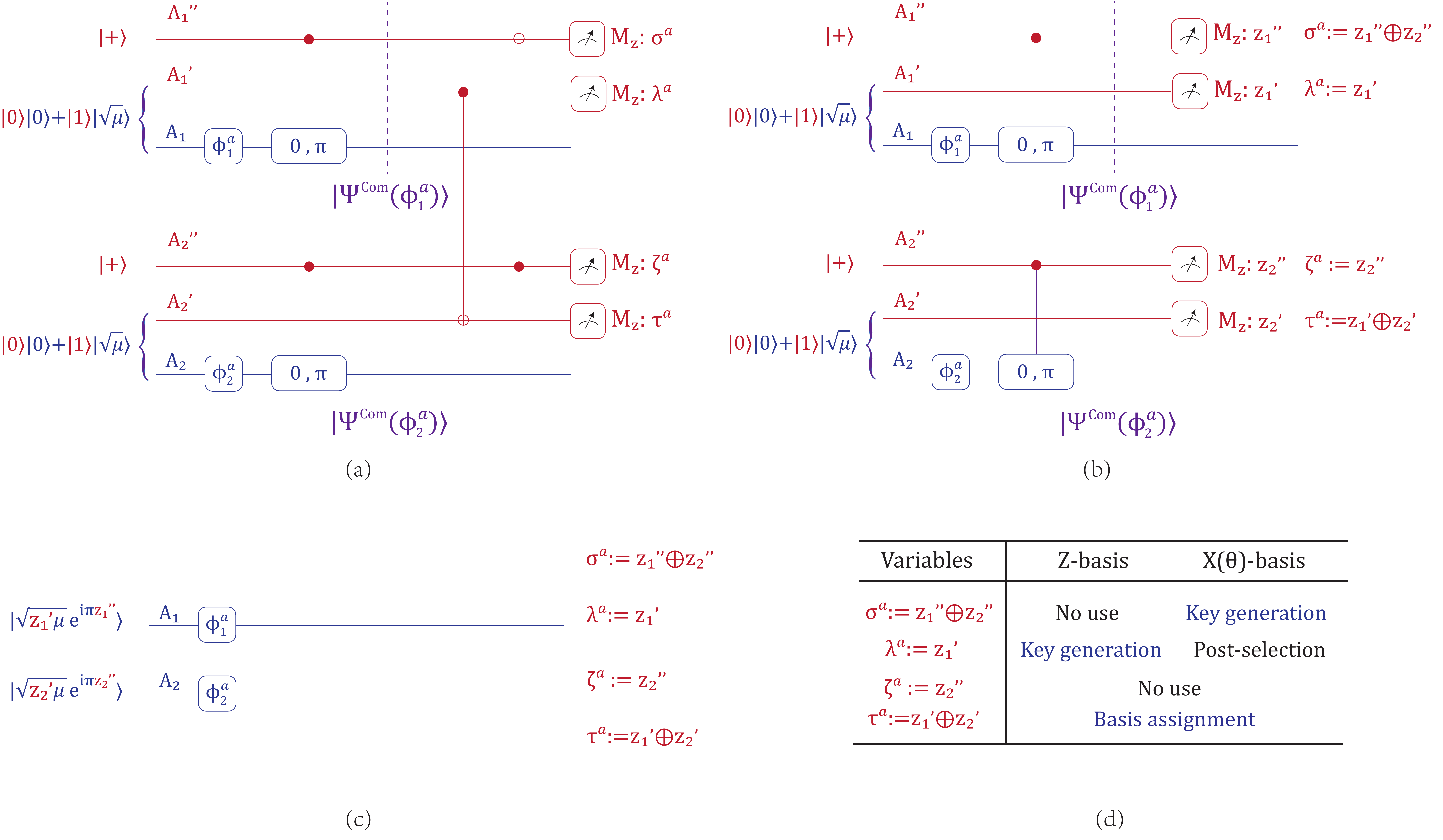}
\caption{(a) All of the operations that Alice performs on the ancillary qubits $A_1''$, $A_1'$, $A_2''$ and $A_2'$ to assign bases and generate raw key bits when the locations $1$ and $2$ are paired. $\tau^a$ is used for the basis assignment, $\lambda^a$ is used for key generation in the $Z$-basis and post-selection in the $X(\theta)$-basis and $\sigma^a$ is used for key generation in the $X(\theta)$-basis. (b) An equivalent process whereby Alice moves all the measurement before the global control gates between locations $1$ and $2$. The measurement results $z_1''$ and $z_1'$ on $A_1''$ and $A_1'$ reveal the encoded phase and the intensity on $A_1$, respectively. (c) Alice further measures all ancillaries in the beginning. This is equivalent to preparing the coherent states on $A_1$ and $A_2$ with modulated intensities and phases. (d) The values and usage of the variables $\sigma^a,\lambda^a,\zeta^a$ and $\tau^a$.} \label{fig:miMDICombReduce}
\end{figure}

In Fig.~\ref{fig:miMDICombReduce}(b), the measurement results $z_1'$ and $z_1''$ on $A_1'$ and $A_1''$ reveal the encoded intensity and phase on $A_1$, respectively. Alice further moves the measurements on the ancillary qubits to the beginning. It is then equivalent to preparing the coherent states $\ket{\sqrt{z_1' \mu} e^{i\pi z_1''}}, \ket{\sqrt{z_2' \mu} e^{i\pi z_2''}}$ on $A_1$ and $A_2$ with modulated intensities and phases. Thus, we reduce the entanglement-based scheme to a prepare-and-measure version.

\begin{Boxes}{Prepare-and-measure MP scheme}{PrepMeaMP}

\begin{enumerate}
\item 
State preparation: In the $i$-th round, Alice generates two random bits $z_i'$ and $z_i''$ and a random phase $\phi^a_i$ uniformly and discretely chosen from $\{\frac{2\pi}{D}j\}_{j\in[D]}$. She generates a coherent state $\ket{\sqrt{z_i'\mu}e^{i(\pi z_i''+\phi^a_i)}}$. Similarly, Bob generates two random bits $y_i'$ and $y_i''$, $\phi^b_i$, and a coherent state $\ket{\sqrt{y_i'\mu}e^{i(\pi y_i''+\phi^b_i)}}$.

\item 
Measurement: Alice and Bob send the coherent states to Charlie, who is supposed to perform a single-photon interference measurement. Charlie announces the clicks of the $L$ and/or $R$ detectors.

Alice and Bob perform the above steps over many rounds. Afterwards, they perform the following data post-processing steps.

\item
Mode pairing: Alice and Bob set a strategy to pair the clicked locations, i.e., to group each two detected rounds together based on Charlie's announcement.

\item 
Basis sifting: For each pair on locations $i$ and $j$, Alice sets $\tau^a := z_i'\oplus z_j'$, $\lambda^a:= z_i'$. Alice sets the basis of the pair to $Z$ if $\tau^a=1$, $X$ if $\tau^a=0$ and $\lambda^a=1$, and `0' if $\tau^a=\lambda^a=0$. Bob assigns the basis in the same way. Alice and Bob announce the bases. If the announced bases are $X,X$ or $Z,Z$, they maintain the signals.

\item
Key mapping: For each $Z$-pair on locations $i$ and $j$, Alice records $\lambda^a$ as the key bit. For each $X$-pair, Alice sets $\sigma^a:= z_i''\oplus z_j''$ as the key bit and records $\theta^a_{i,j}:= \phi^a_i - \phi^a_j$ for later usage. For all $X$-pairs, Alice and Bob further announce $\theta^a$ and $\theta^b$. If $\theta^a=\theta^b$, they keep the $X$-pair data; if $\theta^a -\theta^b = \pi$, they keep the $X$-pair data with Bob flipping the value of $\sigma^b$; otherwise, they discard the $X$-pair data. Moreover, if Charlie's announcement at locations $i$ and $j$ is $(L,R)$ or $(R,L)$, Bob flips the value of $\sigma^b$.

\item 
Parameter estimation: Alice and Bob estimate the fraction of clicked signals $q_{(1,1)}$ and the corresponding phase-error rate $e^X_{(1,1)}$ in the $Z$-pairs, where Alice and Bob both emit a single photon at locations $i$ and $j$. They also estimate the quantum bit-error rate $E^Z_{(\mu,\mu)}$ of the $Z$-pairs.

\item 
Key distillation: Alice and Bob reconcile the key string to $\kappa^a$ and perform privacy amplification using universal-2 hashing matrices. The sizes of hashing matrices are determined by the estimated single-photon yield $q_{(1,1)}$ and the $X(\theta)$-basis error rate $e^{X(\theta)}_{(1,1)}$ according to Eq.~\eqref{eq:keyrateDecoyMDI}.
\end{enumerate}
\end{Boxes}

To reduce the scheme to the MP scheme in the main text, we further simplify the announcement of the random phase in the $X$ basis. Note that the encoding phase in the $i$-th round is $\phi^a_i + \pi z_i''$, with $\phi^a_i \in [0,2\pi)$. The phase modulation of $z_i''$ is redundant in the experiment. We now try to absorb it into $\phi^a_i$.

In the following discussion, without loss of generality, we first assume that $\theta^a$ and $\theta^b$ are equal to $0$ or $\pi$. In this case, Alice and Bob's random phases on locations $i$ and $j$ are either $0$ or $\pi$. We denote them as
\begin{equation}
\begin{aligned}
&\phi^a_i = \pi \hat{\phi}^a_i, \quad \phi^a_j = \pi \hat{\phi}^a_j, \\
&\phi^b_i = \pi \hat{\phi}^b_i, \quad \phi^b_j = \pi \hat{\phi}^b_j,
\end{aligned}
\end{equation}
where $\hat{\phi}^a_i, \hat{\phi}^a_j, \hat{\phi}^b_i$ and $\hat{\phi}^b_j$ are four binary variables.

During the $X$-basis key mapping, Alice and Bob's raw key and random phase at locations $i$ and $j$ are
\begin{equation}
\begin{aligned}
&\sigma^a:= z_i''\oplus z_j'', \quad \theta^a:= \pi(\hat{\phi}^a_i \oplus \hat{\phi}^a_j), \\
&\sigma^b:= y_i''\oplus y_j''\oplus \delta_\theta \oplus \delta_d, \quad \theta^b:= \pi(\hat{\phi}^b_i \oplus \hat{\phi}^b_j).
\end{aligned}
\end{equation}
Here, the binary variable $\delta_\theta:= (\hat{\phi}^a_i \oplus \hat{\phi}^a_j \oplus \hat{\phi}^b_i \oplus \hat{\phi}^b_j)$ indicates whether the announcements of $\theta^a$ and $\theta^b$ are the same or different. $\delta_\theta$ describes the detection announcement in $i$ and $j$: if the announcements are $(L,L)$ and $(R,R)$, then $\delta_d = 0$; if the announcements are $(L,R)$ and $(R,L)$, $\delta_d = 1$.

To simplify the key-mapping strategy, Alice and Bob absorb the random phase $z_i'', z_j''; y_i'', y_j''$ into $\phi_i^a, \phi_j^a; \phi_i^a, \phi_j^b$. They revise the raw key announcement as follows:
\begin{equation}
\begin{aligned}
\bar{\sigma}^a := (\hat{\phi}^a_i \oplus z_i'') \oplus (\hat{\phi}^a_j \oplus z_j''), \\
\bar{\sigma}^b := (\hat{\phi}^b_i \oplus y_i'') \oplus (\hat{\phi}^b_j \oplus y_j'').
\end{aligned}
\end{equation}
Meanwhile, they set the announced basis angle to be $\bar{\theta}^a = \bar{\theta}^b = 0$. In this case, the announced information is less than the original scheme. Hence, the privacy of Alice's raw key $\bar{\sigma}^a$ will not worsen. The new raw key bit is related to the original raw key bit as
\begin{equation}
\begin{aligned}
\bar{\sigma}^a = \sigma^a \oplus (\hat{\phi}^a_i \oplus \hat{\phi}^a_j), \\
\bar{\sigma}^b = \sigma^b \oplus (\hat{\phi}^b_i \oplus \hat{\phi}^b_j).
\end{aligned}
\end{equation}
Therefore, $\bar{\sigma}^a = \bar{\sigma}^b$ iff $\sigma^a = \sigma^b$. Up to a random phase difference, the new raw key bits are equivalent to the original ones.

In general, $\theta^a$ and $\theta^b$ may not be $0$ or $\pi$. In this case, we revise the random phase announcement. Let
\begin{equation} \label{eq:thetar}
\theta^a_{r} := (\phi^a_i + \pi z_i'') - (\phi^a_j + \pi z_j'') = (\phi^a_i - \phi^a_j) + \pi(z_i'' - z_j'') = \theta^a + \pi(z_i'' \oplus z_j'').
\end{equation}
Now, we further split $\theta^a_{r}$ into two parts,
\begin{equation} \label{eq:thetarsplit}
\theta^a_r = \theta^a_0 + \pi \kappa^a,
\end{equation}
where $\theta^a_0 := \theta^a_r \text{ mod } \pi \in [0,\pi)$ denotes half of $\theta^a_r$ while $\kappa^a := (\theta^a_r/\pi) \text{ mod } 2$. Using a similar method, one can show that $\kappa^a = \kappa^b$ iff $\sigma^a = \sigma^b$.

Therefore, in an equivalent but simplified $X(\theta)$-basis key mapping procedure, Alice (Bob) regards the phase $(\phi^a_i+z_i'')$($(\phi^b_i+y_i'')$) as a whole. They calculate the basis angle as Eq.~\eqref{eq:thetar} and \eqref{eq:thetarsplit}. The corresponding revised scheme can now be reduced to the MP scheme in the main text.

\section{Decoy-state estimation} \label{Sec:decoyestimate}

Here, we apply decoy-state analysis~\cite{Zhang2017improved} to the mode-pairing scheme, with the typical three-intensity settings --- vacuum (zero intensity), weak (intensity $\nu$) and signal (intensity $\mu$) states. The discussion below can naturally be extended to the decoy-state method with an arbitrary number of intensities. In the mode-pairing scheme, Alice and Bob each independently generate coherent states with intensities $\mu_i^{a}$ and $\mu_i^{b}$, respectively, for the $i$-th round. Suppose the state intensities are randomly chosen from the set $\{0,\nu,\mu\}$ with probabilities $s_{0}, s_\nu$ and $s_\mu$, respectively, where, $s_{0} + s_\nu + s_\mu = 1$.

We consider the pairing setting $\vec{\chi}$ in Definition~\ref{def:pairing}, which defined on all the locations, including the locations with unsuccessful clicks. In the practical mode-pairing scheme, Alice and Bob decide the pairing setting based on Charlie's announcement $\vec{C}$, which is a $2N$-bit string. On each location, Charlie announces two-bit information $(L_i,R_i)$, indicating the detection results of the two detectors $L$ and $R$. The pairing strategy we mentioned in the main text is an algorithm to choose a specific pairing setting $\vec{\chi}$ based on Charlie's announcement $\vec{C}$. We simply pair all the unpaired locations (including the locations with unsuccessful detections and the ones discarded during the pairing process) in Algorithm~1 to construct a complete pairing setting $\vec{\chi}$.

In the following discussion, we consider the case that Alice and Bob choose a specific fixed pairing setting $\vec{\chi}$. Since the following decoy-state estimation methods holds for all pairing settings $\vec{\chi}$ from $\mathbf{X}$, it can also be applied to the pairing setting $\vec{\chi}$ chosen by the pairing strategy based on $\vec{C}$. A detailed discussion on the security of mode-pairing scheme with a pairing setting $\vec{\chi}$ chosen based on Charlie's announcement $\vec{C}$ can be found in Sec.~\ref{ssc:MPfree}.

Let us first study the source property under a pairing setting $\vec{\chi}$. Alice and Bob pair two locations, $i$ and $j$, based on $\vec{\chi}$ as one round of QKD.
The intensity vector of the $(i,j)$ pair is denoted as
\begin{equation}
\vec{\mu} = (\mu^a_i + \mu^a_j, \mu^b_i + \mu^b_j),
\end{equation}
where $\mu^a_i, \mu^a_j, \mu^b_i, \mu^b_j \in \{0,\nu,\mu\}$ and hence $\mu^a_i + \mu^a_j, \mu^b_i + \mu^b_j \in \{0,\nu,\mu,2\nu,\mu+\nu,2\mu\}$. The probability of Alice and Bob sending out intensities $\vec{\mu}$ for the $(i,j)$ pair can be derived from the single-round probabilities as
\begin{equation} \label{eq:qmu}
q^{\vec{\mu}} = \sum_{(\mu^a_i + \mu^a_j, \mu^b_i + \mu^b_j) = \vec{\mu}} s_{\mu^a_i}\, s_{\mu^a_j}\, s_{\mu^b_i}\, s_{\mu^b_j},
\end{equation}
which is independent of the pairing strategy.

In the entanglement-based mode-pairing scheme presented in Sec.~\ref{Sec:Security}, for each pair of locations $(i,j)$, Alice is able to perform the overall-photon-number measurement on the systems $\tilde{A}_i$ and $\tilde{A}_j$, respectively. We denote the photon-number measurement result to be $k^a$. Likewise, Bob can also perform the photon-number measurement on $\tilde{B}_i$ and $\tilde{B}_j$ and obtain $k^b$. Let $\vec{k}$ denote the photon numbers $(k^a,k^b)$ on the $(i,j)$ pair. Given the intensity setting $\vec{\mu}$ on the $(i,j)$ pair, the probability for Alice and Bob's photon-number measurement result to be $\vec{k}=(k^a,k^b)$ is given by
\begin{equation} \label{eq:twoPoisson}
\Pr(\vec{k}|\vec{\mu})= e^{-(\mu^a_i+\mu^a_j)-(\mu^b_i+\mu^b_j)}\frac{(\mu^a_i + \mu^a_j)^{k^a}(\mu^b_i + \mu^b_j)^{k^b}}{k^a! k^b!},
\end{equation}
which is a product of two Poisson distributions, since the intensity settings of Alice and Bob are independent.

In the above description, Alice and Bob first choose the intensity setting $\vec{\mu}$ with probability $q^{\vec{\mu}}$ and then measure the ancillary systems to obtain the photon number $\vec{k}$ with probability $\Pr(\vec{k}|\vec{\mu})$. From another viewpoint, the \emph{conditional probability} of the intensity setting being $\vec{\mu}$ if Alice and Bob measure the photon number on $(i,j)$-pair and obtain the result $\vec{k}$ is given by
\begin{equation} \label{eq:Prmuk}
\begin{aligned}
\Pr(\vec{\mu}|\vec{k}) 
&= \frac{q^{\vec{\mu}}\Pr(\vec{k}|\vec{\mu})}{\sum_{\vec{\mu}'}  q^{\vec{\mu}'}\Pr(\vec{k}|\vec{\mu}') },
\end{aligned}
\end{equation}
where the summation of $\vec{\mu}'=(\mu^a, \mu^b)$ takes over all possible source-intensity settings. 
The conditional probability $\Pr(\vec{\mu}|\vec{k})$ reflects an intrinsic property of the source, which is independent of the specific pairing setting $\vec{\chi}$. We remark that, the probabilities $\Pr(\vec{k}|\vec{\mu}), q^{\vec{\mu}}$ and $\Pr(\vec{\mu}|\vec{k})$ in Eq.~\eqref{eq:twoPoisson} and \eqref{eq:Prmuk} are \emph{a prior} probability distributions without the post-selection of the emitted pulses based on Charlie's announcement $\vec{C}$.

Now, let us investigate the detection side. For an $(i,j)$ pair, Alice and Bob determine the bases and raw key bits based on the encoded intensities and phases. If the intensity vector $\vec{\mu}$ satisfies
\begin{equation}
\mu^a_i \mu^a_j = \mu^b_i \mu^b_j = 0, \mu^a_i+\mu^a_j+\mu^b_i+\mu^b_j\neq 0,
\end{equation}
then it is a $Z$-pair. Below, we mainly consider decoy-state estimation in the $Z$ basis, but the analysis for $X$-basis estimation is similar.

After mode pairing and sifting, Alice and Bob obtain $M$ rounds of $Z$-pairs with successful detections; of these, $E$ rounds are erroneous. Let $M^{\vec{\mu}}$ denote the number of rounds detected with intensity setting $\vec{\mu}$, among which $E^{\vec{\mu}}$ rounds are wrongly detected. Let the total and erroneous numbers of rounds with photon numbers $\vec{k}$ be denoted by $M_{\vec{k}}$ and $E_{\vec{k}}$, respectively, among which $M^{\vec{\mu}}_{\vec{k}}$ and $E^{\vec{\mu}}_{\vec{k}}$ are the rounds with intensity setting $\vec{\mu}$. We have the following relations,
\begin{equation} \label{eq:MsumEsum}
\begin{aligned}
M &= \sum_{\vec{\mu}} M^{\vec{\mu}} = \sum_{\vec{k}} M_{\vec{k}} =\sum_{\vec{\mu}}\sum_{\vec{k}} M^{\vec{\mu}}_{\vec{k}} \\
E &= \sum_{\vec{\mu}} E^{\vec{\mu}} = \sum_{\vec{k}} E_{\vec{k}} =\sum_{\vec{\mu}}\sum_{\vec{k}} E^{\vec{\mu}}_{\vec{k}}.
\end{aligned}
\end{equation}
In the entire QKD process, the values $\{M^{\vec{\mu}}, E^{\vec{\mu}}\}_{\vec{\mu}}$ are known to Alice and Bob. The values $\{M_{\vec{k}}$ and $E_{\vec{k}}\}_{\vec{k}}$, however, are unknown to them but are fixed after Charlie's announcement.

We remark that, no matter what the intensity setting $\vec{\mu}$ is, if Alice and Bob measures the photon number on the ancillary systems and obtain the result $\vec{k}$, the post-selected state is independent of the intensity setting. As a result, \emph{conditional on Alice and Bob's specific photon-number measurement result $\vec{k}$, Charlie's announcement strategy is independent of the intensity setting $\vec{\mu}$}. This is an alternative description of the basic assumption that the original decoy-state method makes in its security proof~\cite{Lo2005Decoy}.

Based on this property, in all the generated $Z$-pairs with the photon-number measurement result to be $\vec{k}$, the expected ratio of different intensity settings should be the same as the ratio of the emitted states,
\begin{equation} \label{eq:MmukZ}
\frac{ M^{\vec{\mu}}_{\vec{k}} }{ M^{\vec{\mu}'}_{\vec{k}} } \approx \frac{ \Pr(\vec{\mu},\vec{k}) }{ \Pr(\vec{\mu}',\vec{k}) } = \frac{ q^{\vec{\mu}}\Pr(\vec{k}|\vec{\mu}) }{q^{\vec{\mu}'}\Pr(\vec{k}|\vec{\mu}') },
\end{equation}
where the approximation is due to statistical fluctuation. In the limit of infinite data size, the approximation becomes an equality.

In our experiment, we assume that Alice and Bob randomly determine the intensity $\vec{\mu}$ at each location by measuring an ancillary system. Measurement results are used to classically control the intensity of the signals to Charlie. In practice, Alice and Bob first measure the overall photon number and obtain $\vec{k}$. In this case, the intensity choice is determined only by $\vec{k}$, independent of Charlie's announcement.
Among the $M_{\vec{k}}$ pairs with photon number $\vec{k}$, the number of pairs with an intensity setting $\vec{\mu}$ is a variable $\mathcal{M}_{\vec{k}}^{\vec{\mu}}$ that is determined by Alice and Bob's ancillary systems. The expected ratio of the intensity setting $\vec{\mu}$ is given by
\begin{equation} \label{eq:conPrmuK2}
\mathbb{E}\left( \frac{ \mathcal{M}^{\vec{\mu}}_{\vec{k}}}{ M_{\vec{k}}} \right) = \Pr(\vec{\mu}|\vec{k}) = \frac{q^{\vec{\mu}}\Pr(\vec{k}|\vec{\mu})}{\sum_{\vec{\mu}'}  q^{\vec{\mu}'}\Pr(\vec{k}|\vec{\mu}') },
\end{equation}
where the expectation value is taken for the random variable $\mathcal{M}_{\vec{k}}^{\vec{\mu}}$ characterising Alice and Bob's intensity settings. Similarly, for erroneous detection among the $E_{\vec{k}}$ pairs with $\vec{k}$ photon numbers, $\mathcal{E}_{\vec{k}}^{\vec{\mu}}$ is the variable for describing the number of pairs with intensity setting $\vec{\mu}$,
\begin{equation} \label{eq:conPrmuK2E}
\mathbb{E}\left( \frac{ \mathcal{E}^{\vec{\mu}}_{\vec{k}}}{ E_{\vec{k}}} \right) = \Pr(\vec{\mu}|\vec{k}) = \frac{q^{\vec{\mu}}\Pr(\vec{k}|\vec{\mu})}{\sum_{\vec{\mu}'}  q^{\vec{\mu}'}\Pr(\vec{k}|\vec{\mu}') }.
\end{equation}
From Eqs.~\eqref{eq:conPrmuK2} and \eqref{eq:conPrmuK2E}, we have,
\begin{equation} \label{eq:EMeMmuk}
\begin{aligned}
\mathbb{E}[\mathcal{M}^{\vec{\mu}}_{\vec{k}}] &= \Pr(\vec{\mu}|\vec{k}) M_{\vec{k}}, \\
\mathbb{E}[\mathcal{E}^{\vec{\mu}}_{\vec{k}}] &= \Pr(\vec{\mu}|\vec{k}) E_{\vec{k}}. \\
\end{aligned}
\end{equation}
Let the composite variables $\mathcal{M}^{\vec{\mu}}$ and $\mathcal{E}^{\vec{\mu}}$ denote the overall numbers of pairs with intensity settings $\vec{\mu}$,
\begin{equation} \label{eq:barMmu}
\begin{aligned}
\mathcal{M}^{\vec{\mu}} &= \sum_{k} \mathcal{M}^{\vec{\mu}}_{\vec{k}}, \\
\mathcal{E}^{\vec{\mu}} &= \sum_{k} \mathcal{E}^{\vec{\mu}}_{\vec{k}}. \\
\end{aligned}
\end{equation}
From Eqs.~\eqref{eq:EMeMmuk} and \eqref{eq:barMmu}, we have
\begin{equation} \label{eq:EMeMmu}
\begin{aligned}
\mathbb{E}[\mathcal{M}^{\vec{\mu}}] &= \sum_{\vec{k}}\Pr(\vec{\mu}|\vec{k}) M_{\vec{k}} \\
\mathbb{E}[\mathcal{E}^{\vec{\mu}}] &= \sum_{\vec{k}}\Pr(\vec{\mu}|\vec{k}) E_{\vec{k}}. \\
\end{aligned}
\end{equation}

To summarise, decoy-state analysis can be formulated as follows. \emph{For unknown $\{M_k, E_k\}_k$ and known $\{\Pr(\vec{\mu}|\vec{k})\}_{\vec{\mu},\vec{k}}$, the lower bound of $\mathcal{M}^{(\mu,\mu)}_{(1,1)}$ and the upper bound of $\mathcal{E}^{(\mu,\mu)}_{(1,1)}$ are evaluated given the observed constraints of the detected pair numbers $\{ M^{\vec{\mu}}, E^{\vec{\mu}}\}_{\vec{\mu}}$}. This problem can be solved through the following steps.
\begin{enumerate}
\item
Record the number of clicked rounds $\{M^{\vec{\mu}}, E^{\vec{\mu}}\}_{\vec{\mu}}$.
\item
Evaluate the upper and lower bounds of the expectation values $\left\{ \mathbb{E}\left[\mathcal{M}^{\vec{\mu}}\right], \mathbb{E}\left[\mathcal{E}^{\vec{\mu}}\right] \right\}_{\vec{\mu}}$ using the (inverse) Chernoff bound for independent variables.
\item
Solve the following linear-programming problems:
\begin{equation} \label{eq:progM11}
\begin{aligned}
\min \quad & M_{(1,1)} \\
\text{s.t.} \quad & \mathbb{E}^L\left[\mathcal{M}^{\vec{\mu}}\right] \leq \sum_{\vec{k}}\Pr(\vec{\mu}|\vec{k}) M_{\vec{k}} \leq \mathbb{E}^U\left[\mathcal{M}^{\vec{\mu}}\right], \quad \forall \vec{\mu},\\
& 0 \leq M_{\vec{k}} \leq M,
\end{aligned}
\end{equation}
\begin{equation} \label{eq:progE11}
\begin{aligned}
\max \quad & E_{(1,1)} \\
\text{s.t.} \quad & \mathbb{E}^L\left[\mathcal{E}^{\vec{\mu}}\right] \leq \sum_{\vec{k}}\Pr(\vec{\mu}|\vec{k}) E_{\vec{k}} \leq \mathbb{E}^U\left[\mathcal{E}^{\vec{\mu}}\right], \quad \forall \vec{\mu},\\
& 0 \leq E_{\vec{k}} \leq E,
\end{aligned}
\end{equation}
where $\vec{\mu}\in \{ 0, \nu, \mu\}^{\otimes 2}$, $\vec{k} \in \mathbb{Z}_+^{\otimes 2}$ and $\mathbb{Z}_+ = \{0,1,2,...\}$ is the set of natural numbers. Denote the results for these two programmes as $M^{L}_{(1,1)}$ and $E^{U}_{(1,1)}$, respectively.

\item
To evaluate the upper bound of $M^{(\mu,\mu)}_{(1,1)}$ and the lower bound of $E^{(\mu,\mu)}_{(1,1)}$ using the values of $M^{L}_{(1,1)}$ and $E^{U}_{(1,1)}$ with Eq.~\eqref{eq:EMeMmuk}, one can apply the Chernoff bound again for independent variables.
\end{enumerate}

Here, we explain why Alice and Bob can use the Chernoff bound for independent variables in Step~2 and 4. As we mentioned before Eq.~\eqref{eq:MmukZ}, when the photon number $\vec{k}$ on the pair of location is given, the intensity $\vec{\mu}$ for each pair of locations is determined merely by the conditional probability $\Pr(\vec{\mu}|\vec{k})$ in Eq.~\eqref{eq:EMeMmuk}. As a result, the intensity choices in different pair of locations are independent when the photon number $\vec{k}$ of the pairs are given. Moreover, when the emitted state in each round is identical, the intensity choices in different pairs of locations are also identical.
To be more clear, we consider all the $M_{\vec{k}}$ clicked pairs for the pairs with specific measured photon number $\vec{k}$. If Alice and Bob check the intensity label $\vec{\mu}$ on these pairs, and treat these intensity labels as random variables, then these varibles are i.i.d.,
\begin{equation} \label{eq:inten_var}
\hat{\vec{\mu}}_i = \vec{\mu}, \text{  with probability  } \Pr(\vec{\mu}|\vec{k}).\quad  (i=1,2,...,M_{\vec{k}})
\end{equation}
Here, the value of the conditional probability $\Pr(\vec{\mu}|\vec{k})$ only depends on the property of the source.
In this case, Alice and Bob can evaluate the real number of pairs $M_{\vec{k}}^{\vec{\mu}}$ from its expectation value $M_{\vec{k}} \Pr(\vec{\mu}|\vec{k})$ using Chernoff bound of independent variables (Step 4 in the decoy estimation procedure above, cf. Eq.~(33) in Ref.~\cite{Zhang2017improved}); on the other hand, if Alice and Bob were awared of the photon number measurement results on the pairs, they can infer the expectation value from the real number of pairs $M_{\vec{k}}^{\vec{\mu}}$ by the inverse Chernoff bound (Step 2 in the decoy estimation procedure above, cf. Eq.~(30) in Ref.~\cite{Zhang2017improved}).

In the practical Step 2, however, Alice and Bob do not have the photon numer measurement results on the pairs. As an alternative way, instead of solving the inverse Chernoff problem for each photon number setting $\vec{k}$, they can combine the estimation problems together: to infer the expectation value of $\mathcal{M}^{\vec{\mu}}:= \sum_{\vec{k}} \mathcal{M}^{\vec{\mu}}_{\vec{k}}$ based on the observed value of $M^{\vec{\mu}}$. When the underlying values the photon numbers $\vec{k}$ for each pair of locations are determined, the intensity variables $\hat{\vec{\mu}}_i$ are independent, following Eq.~\eqref{eq:inten_var}, regardless of the photon number $\vec{k}$. For different $\vec{k}$, however, the variables $\hat{\vec{\mu}}_i$ are not identical anymore, as the conditional probabilities $\Pr(\vec{\mu}|\vec{k})$ differ. In this case, we can still apply the inverse Chernoff bound for the independent variables.

We have introduced the decoy-state procedure for the $Z$-basis. Likewise, Alice and Bob can perform decoy-state analysis for the sifted $X$-basis obtain the single-photon-pair numbers $M_{(1,1)}^{X,L}$ and $E_{(1,1)}^{X,U}$ for successful and erroneous detection, respectively. Let $M^{(\mu,\mu),Z,L}_{(1,1)}$ denote the estimated lower bound on the single-photon-pair number with intensity $(\mu^a,\mu^b) = (\mu,\mu)$ in the case where Alice and Bob emit $Z$-basis data. The finite key-length formula of the mode-pairing scheme is given by
\begin{equation}
\begin{aligned}
M_R &= M^{(\mu,\mu),Z}\left\{q^{(\mu,\mu),Z,L}_{(1,1)}\left[ 1 - H\left(e^{(\mu,\mu),Z,U}_{(1,1),ph}\right) \right] - f H(E^Z_{\mu\mu}) \right\}, \\
&= M^{(\mu,\mu),Z,L}_{(1,1)}\left[ 1 - H\left( \frac{E^{(\mu,\mu),Z,U}_{(1,1),ph}}{ M^{(\mu,\mu),Z,L}_{(1,1)} } \right) \right] - fM^{(\mu,\mu),Z} H(E^Z_{\mu\mu}),
\end{aligned}
\end{equation}
where $E^{(\mu,\mu),Z,U}_{(1,1),ph}$ is the upper bound of the phase-error number of the single-photon pair with the intensity $(\mu^a,\mu^b) = (\mu,\mu)$ when Alice and Bob emit $Z$-basis data. $E^{(\mu,\mu),Z,U}_{(1,1),ph}$ can be directly estimated from $M^{(\mu,\mu),Z,L}_{(1,1)}$, $M_{(1,1)}^{X,L}$ and $E_{(1,1)}^{X,U}$ based on the statistics of random sampling without a replacement problem~\cite{Fung2010practical}. The number of $(\mu,\mu)$-pairs in the $Z$-basis $M^{(\mu,\mu,Z)}$ and the quantum bit-error number $E^Z_{\mu\mu}$ can be directly obtained from the experimental results. The lower bound of the fraction $q^{(\mu,\mu),Z,L}_{(1,1)}$ and upper bound of the phase-error rate $e^{(\mu,\mu),Z,U}_{(1,1),ph}$ are defined as
\begin{equation} \label{eq:qmumuemumu}
\begin{aligned}
q^{(\mu,\mu),Z,L}_{(1,1)} &:= \frac{ M^{(\mu,\mu),Z,L}_{(1,1)} }{ M^{(\mu,\mu),Z} }, \\
e^{(\mu,\mu),Z,U}_{(1,1),ph} &:= \frac{E^{(\mu,\mu),Z,U}_{(1,1),ph}}{ M^{(\mu,\mu),Z,L}_{(1,1)} },
\end{aligned}
\end{equation}
In the asymptotic case, $q^{(\mu,\mu),Z,L}_{(1,1)}$ and $e^{(\mu,\mu),Z,U}_{(1,1),ph}$ become $q_{11}$ and $e^X_{11}$, respectively, defined in the main text.

In practice, the linear-programming problems in Eq.~\eqref{eq:progM11} and \eqref{eq:progE11} can be analytically solved by converting the detection number $M^{\vec{\mu}}, E^{\vec{\mu}}, M_{\vec{k}}$ and $E_{\vec{k}}$ into traditional ratio values (i.e. the gain and yield). To do this, we denote the expected total pair number to be $N_p = \lfloor N/2 \rfloor$.
Furthermore, we define $N^{\vec{\mu}} := q^{\vec{\mu}} N_p$ to be the expected number of rounds with intensity setting $\vec{\mu}$. We define
$N^\infty_{\vec{k}}:= \sum_{\vec{\mu}} \Pr(\vec{k}|\vec{\mu}) N^{\vec{\mu}}$ to be the expected number of rounds where the source sends out $\vec{k}$-photon states. Then the `gain' and `yield' are defined as
\begin{equation} \label{eq:QeQYeYdetected}
\begin{aligned}
\mathcal{Q}^{\vec{\mu}} := \frac{ \mathcal{M}^{\vec{\mu}} }{ N^{\vec{\mu}} }, & \quad \mathcal{Q}_E^{\vec{\mu}} := \frac{ \mathcal{E}^{\vec{\mu}} }{ N^{\vec{\mu}} }, \\
Y_{\vec{k}}^* := \frac{M_{\vec{k}}}{N^\infty_{\vec{k}} }, & \quad (Y_E)_{\vec{k}}^* := \frac{E_{\vec{k}}}{N^\infty_{\vec{k}} }.
\end{aligned}
\end{equation}

The relationship between $\mathcal{Q}^{\vec{\mu}}$ and $Y_{\vec{k}}^*$ is,
\begin{equation} \label{eq:normaldecoydetected}
\begin{aligned}
\mathbb{E}[\mathcal{Q}^{\vec{\mu}}] &= \mathbb{E}\left[\frac{ \mathcal{M}^{\vec{\mu}} }{ N^{\vec{\mu}} }\right] = \frac{ \mathbb{E}[\mathcal{M}^{\vec{\mu}} ]}{ N^{\vec{\mu}} }\\
&= \sum_{\vec{k}} \Pr(\vec{\mu}|\vec{k}) \frac{ M_{\vec{k}} }{ q^{\vec{\mu}} N_p } \\
&= \sum_{\vec{k}} \frac{ q^{\vec{\mu}} \Pr( \vec{k} | \vec{\mu} ) }{ \sum_{\vec{\mu}'} q^{\vec{\mu}'} \Pr(\vec{k}|\vec{\mu}')  } \frac{ M_{\vec{k}} }{q^{\vec{\mu}} N_p } \\
&= \sum_{\vec{k}} \Pr( \vec{k} | \vec{\mu} )  \frac{ M_{\vec{k}} }{ \sum_{\vec{\mu}'} \left(N_p q^{\vec{\mu}'}\right) \Pr(\vec{k}|\vec{\mu}') } \\
&= \sum_{\vec{k}} \Pr(\vec{k}|\vec{\mu}) Y_{\vec{k}}^*, \\
\mathbb{E}[\mathcal{Q}_E^{\vec{\mu}}] &= \sum_{\vec{k}} \Pr(\vec{k}|\vec{\mu}) (Y_E)_{\vec{k}}^*, \\
\end{aligned}
\end{equation}
where $\Pr(\vec{k}|\vec{\mu})$ is the product of two Poisson distributions, given by Eq.~\eqref{eq:twoPoisson}. In this way, one can directly apply the traditional decoy-state formulas for MDI-QKD as a solution to the linear-programming problems \cite{ma2012statistical,curty2014finite}. Note that the value of $N_p$ will not affect the correctness of Eq.~\eqref{eq:normaldecoydetected}. If we redefine $N_p$ as $N_p' = cN_p$, then the defined gain and yield will become $\frac{1}{c}\mathcal{Q}^{\vec{\mu}}$ and $\frac{1}{c}Y_{\vec{k}}^*$, respectively. From Eq.~\eqref{eq:QeQYeYdetected}, we can see that the final estimated value $M_{\vec{k}}$ is invariant only if the experimental data $\{M^{\vec{\mu}}\}$ are unchanged.

\section{Simulation formulas} \label{Sec:Simu}

Here, we list the formulas used for key-rate simulation of different QKD schemes. In Sec.~\ref{Ssec:MPsimu}, we derive the simulation formula for the free-round MP scheme with maximal pairing length $l$. In a similar fashion, we derive the simulation formula for the time-bin MDI-QKD scheme in Sec.~\ref{SSec:timebinMDIQKDsimu}. Finally, in Sec.~\ref{Ssec:othersimu}, we list the simulation formulas used in the main text for the decoy-BB84 and PM-QKD schemes.

\subsection{Mode-pairing scheme} \label{Ssec:MPsimu}

The key rate of the MP scheme can be calculated using Eqs.~\eqref{eq:keyrate},~\eqref{eq:rpSupp},~\eqref{eq:pulseQ}, \eqref{eq:pairQ}, \eqref{eq:pairEZ}, \eqref{eq:barq11} and \eqref{eq:e11X} below.

In the MP scheme, Alice and Bob use the $Z$-basis to generate keys. In the asymptotic case, we assume that Alice and Bob choose to randomly emit signals with intensity $\{0,\mu\}$ with a probability of nearly $1/2$ and a decoy intensity $\nu$ with negligible probability. Let us express the coherent pulse emitted by Alice in the $i$-th round as $\ket{\sqrt{z_i^a \mu} \exp(i\phi_i^a)}$, with $z_i^a\in \{0,1\}$ being a random variable indicating the intensity and $\phi_i^a$ being a random phase. Similarly, Bob emits $\ket{\sqrt{z_i^b \mu} \exp(i\phi_i^b)}$ in the $i$-th round. The $i$-th round intensity setting is then denoted by a $2$-bit vector $z_i := [z_i^a, z_i^b]$.

In the simulation, Alice and Bob are assumed to emit pulses to Charlie through a typical symmetric-loss channel. The single-side transmittance of the channel (with included detection efficiency $\eta_d$ included) is $\eta_s$, and the dark count of each detector is $p_d$. The simulation data is provided in Fig.~4 in the main text. In this case, the channel is i.i.d. for each round. Alice and Bob then pair the clicked pulses and set their bases. For the $(i,j)$-th pulses to be paired, let $\tau_{i,j} = [\tau^a_{i,j}, \tau^b_{i,j}] := [z_i^a\oplus z_j^a, z_i^a\oplus z_j^a]$, where $\oplus$ is the bit-wise addition modulo 2. When $\tau_{i,j}=[1,1]$, the $(i,j)$-pair is then set to be a signal pair.

The key rate of the MP scheme is expressed as
\begin{equation} \label{eq:keyrate}
R = r_p(p,l)\, r_s \left\{\bar{q}_{11}[ 1 - H(e^X_{(1,1)}) ] - f H(E^Z_{(\mu,\mu)}) \right\},
\end{equation}
where $l$ is the maximal pairing gap, $r_p(p,l)$ is the expected pair number generated during each round, and $r_s$ is the probability that a clicked pulse pair is a signal pair, $\bar{q}_{11}$ and $e^X_{(1,1)}$ are, respectively, the expected single-photon pair ratio in all signal pairs and the phase error of the single-photon pairs, which are estimated via decoy-state methods. $f$ is the error-correction efficiency, $E^Z_{(\mu,\mu)}$ is the bit-error rate of the signal pairs and $p$ is the probability of the $i$-th signal to be successfully clicked (as defined in Eq.~\eqref{eq:pulseQ}). The value of $r_p(p,l)$ has already been given in the Methods as
\begin{equation} \label{eq:rpSupp}
r_p(p, l) = \left[ \frac{1}{p [1 - (1-p)^l]} + \frac{1}{p} \right]^{-1}.
\end{equation}


Here, we first consider the calculation of $r_s, \bar{q}_{11}$ and $E^Z_{(\mu,\mu)}$. In the $i$-th round, we denote the click events of the left and right detectors at the $i$-th turn by two variables $(L_i, R_i)$. The successful click variable is $C_i = L_i \oplus R_i$. Only when $C_i = 1$ does an successful click occur. The detection probability $\Pr(C_i=1|z_i)$ is given by
\begin{equation} \label{eq:pulsexiQ}
\Pr(C_i=1|z_i) \approx 1 - (1-2p_d) \exp\left[-\eta_s\mu\left(z^a_i + z^b_i\right)\right].
\end{equation}
The expected click probability, i.e., the total transmittance of each round, is
\begin{equation} \label{eq:pulseQ}
p := \Pr(C_i = 1) = \sum_{z_i} \Pr(C_i = 1|z_i) \Pr(z_i) = \frac{1}{4} \sum_{z_i} \Pr(C_i = 1|z_i).
\end{equation}
The pairing rate $r_p(p, l)$ can then be calculated using Eq.~\eqref{eq:rpSupp}.

The phase-randomised coherent states emitted in the $i$-th round can also be regarded as the mixture of photon-number states. $\Pr(C_i=1|n_i)$ denotes the detection probability when Alice and Bob emit photon-number states $\ket{n_i^a}$ and $\ket{n_i^b}$, respectively; it is given by
\begin{equation} \label{eq:pulsexiY}
\Pr(C_i=1|n_i) \approx 1 - (1-2p_d) (1 - \eta_s)^{(n_i^a + n_i^b)}.
\end{equation}


Now, let us calculate the signal-pair ratio $r_s$. Without loss of generality, we consider the $i$-th and $j$-th turns as a paired group. For a general turn, the probability of a click caused by the intensity setting $z$ is
\begin{equation}
\Pr(z|C=1) = \frac{ \Pr(z, C=1) }{ \Pr(C = 1) } = \frac{ \Pr(C=1|z) }{ \sum_{z'} \Pr(C=1|z') }.
\end{equation}
Note that the subscript is omitted, since the detections in all rounds are identical and independently distributed in our simulation.

In the MP scheme, an successful click happens when $\tau_{i,j} = [1,1]$. Therefore, four possible configurations of $z_i$ and $z_j$ (which generate signal pairs) are
\begin{equation}
[z_i,z_j] \in \{[00,11],[01,10],[10,01],[11,00]\};
\end{equation}
of these, $Err := \{[00,11],[11,00]\}$ are the two configurations that cause bit errors. To simplify our notations, we introduce several events,
\begin{equation}
\begin{aligned}
&\Pr(C) = \Pr(\textit{Pair Clicked}) := \Pr(C_i=C_j=1)= p^2, \\
&\Pr(E) = \Pr(\textit{Pair Effective}) := \Pr(z_i\oplus z_j=11), \\
&\Pr(Err) = \Pr(\textit{Pair Erroneous}) := \Pr([z_i,z_j]\in Err), \\
&\Pr(S) = \Pr(\textit{Single-photon Pair}) := \Pr(n_i\oplus n_j = 11).
\end{aligned}
\end{equation}

The signal-pair ratio $r_s$ is then
\begin{equation} \label{eq:pairQ}
\begin{aligned}
r_s =&\Pr(E|C) = \Pr(z_i \oplus z_j = 11| C_i=1, C_j=1) \\
&= \sum_{z_i \oplus z_j =11} \, \Pr(z_i| C_i=1) \Pr(z_j | C_j = 1) \\
&= \sum_{z_i \oplus z_j = 11} \frac{\Pr(C_i = 1|z_i) \Pr(z_i) }{ \Pr(C_i = 1) } \frac{\Pr(C_j = 1|z_j) \Pr(z_j)}{\Pr(C_j = 1)} \\
&= \frac{1}{16}\frac{1}{p^2} \sum_{z_i \oplus z_j = 11} \Pr(C_i = 1|z_i) \Pr(C_j = 1|z_j), \\
\end{aligned}
\end{equation}

We remark that, when $\eta_s\mu \ll 1$, the signal-pair ratio $r_s$ is approximately
\begin{equation}
\begin{aligned}
r_s &\approx \frac{1}{16} \frac{1}{p^2} [2 (\eta_s\mu)^2 ] \\
&= \frac{1}{16} \frac{1}{(\eta_s\mu)^2} [2 (\eta_s\mu)^2 ]  = \frac{1}{8},
\end{aligned}
\end{equation}
which is nearly a constant independent of $\eta_s$ and $\mu$.

The expected error rate $E^Z_{(\mu,\mu)}$ of the $(i,j)$-pair is
\begin{equation} \label{eq:pairEZini}
\begin{aligned}
E^Z_{(\mu,\mu)} &=\Pr\left(Err|E,C\right) \\
&= \frac{ \Pr\left(Err,E|C \right) }{ \Pr\left(E|C\right) } = \frac{ \Pr\left(Err|C \right) }{ \Pr\left(E|C\right) }\\
&= r_s^{-1} \Pr\left(Err|C \right).
\end{aligned}
\end{equation}
Note that the erroneous pair condition is contained in the effective pair condition. The erroneous probability is given by
\begin{equation}
\begin{aligned}
&\Pr\left(Err|C \right) =\Pr \left([z_i,z_j] \in Err|C_i = C_j=1 \right) \\
&= \sum_{[z_i,z_j] \in Err} \, \Pr(z_i| C_i=1) \Pr(z_j | C_j = 1) \\
&= \sum_{[z_i,z_j] \in Err} \frac{\Pr(C_i = 1|z_i) \Pr(z_i) }{ \Pr(C_i = 1) } \frac{\Pr(C_j = 1|z_j) \Pr(z_j)}{\Pr(C_j = 1)} \\
&= \frac{1}{16}\frac{1}{p^2} \sum_{[z_i,z_j]\in Err } \Pr(C_i = 1|z_i) \Pr(C_j = 1|z_j).
\end{aligned}
\end{equation}
Therefore,
\begin{equation} \label{eq:pairEZ}
\begin{aligned}
E^Z_{(\mu,\mu)} &= \frac{1}{16}\frac{1}{r_s p^2} \sum_{[z_i,z_j] \in Err} \Pr(C_i = 1|z_i) \Pr(C_j = 1|z_j).
\end{aligned}
\end{equation}

Then, we calculate the expected single-photon pair ratio $\bar{q}_{11}$ in the effective signal-pairs:
\begin{equation} \label{eq:barq11}
\begin{aligned}
\bar{q}_{11} &= \Pr(S| E, C) = \frac{ \Pr(S, E, C) }{ \Pr(E, C) } \\
&= \frac{1}{r_s p^2}\Pr(S, E, C) \\
&= \frac{1}{r_s p^2} \sum_{z_i,z_j} \Pr(S,E,C| z_i,z_j) \Pr(z_i,z_j) \\
&= \frac{1}{16} \frac{1}{r_s p^2}  \sum_{z_i\oplus z_j=11} \Pr(S, C| z_i,z_j) \\
&= \frac{1}{16} \frac{1}{r_s p^2}  \sum_{z_i\oplus z_j=11} \Pr(C| S, z_i,z_j)  \Pr(S | z_i,z_j) \\
&= \frac{1}{16} \frac{P_\mu(1)^2}{r_s p^2}  \sum_{z_i\oplus z_j=11} \Pr(C_i=1| n_i = z_i)\Pr(C_j=1| n_j = z_j),
\end{aligned}
\end{equation}
where $P_\mu(k) = \exp(-\mu)\frac{\mu^k}{k!}$ is the Poisson distribution and $\Pr(C_i=1|n_i)$ is defined in Eq.~\eqref{eq:pulsexiY}.

The phase-error rate formula of the MP scheme cannot be calculated for single-pulse-based methods. In the MP scheme, if the decoy-state estimation is perfect, then the phase-error rate should be the same as that of the original two-mode MDI-QKD scheme. Here, we apply the formulas (A9) and (A11) in Ref.~\cite{Ma2012alternative}:
\begin{equation} \label{eq:e11X}
\begin{aligned}
Y_{11} & = (1-p_d)^2 \left[ \dfrac{\eta_a\eta_b}{2} + (2\eta_a + 2\eta_b -3\eta_a\eta_b)p_d + 4(1-\eta_a)(1-\eta_b)p_d^2 \right], \\
e_{11}^X Y_{11} & = e_0Y_{11} - (e_0- e_d)(1-p_d^2)\dfrac{\eta_a\eta_b}{2}.
\end{aligned}
\end{equation}
Here, $e_0=1/2$ is the error caused by vacuum signals, $e_d$ is the pre-set misalignment error, and $\eta_a(\eta_b)$ are the transmittances from Alice (Bob) to Charlie. In our case, $\eta_a = \eta_b = \eta_s$.

\subsection{Time-bin MDI-QKD} \label{SSec:timebinMDIQKDsimu}

The key rate of time-bin MDI-QKD is calculated as Eqs.~\eqref{eq:keyrateMDI},~\eqref{eq:MDIQmumu},~\eqref{eq:MDIEZmumu},~\eqref{eq:MDIq11}, and ~\eqref{eq:e11X2} below.

The time-bin MDI-QKD scheme can be regarded as a special case of the MP scheme with fixed-pairing setting and a predetermined effective-intensity setting $[z_i, z_j]$ at locations $i$ and $j$. That is, if the pairs are clicked, then the pair is effective with $z_i\oplus z_j = 11$. The derivation of the simulation formulas for time-bin MDI-QKD is similar to that of the MP scheme in Sec.~\ref{Ssec:MPsimu}.

In the time-bin MDI-QKD scheme, Alice and Bob use the $Z$-basis to generate keys. In the asymptotic case, we assume that they choose to emit the $X$-basis state with a negligible probability, such that the basis-sifting factor is almost $1$. The key rate of the time-bin MDI-QKD scheme is
\begin{equation} \label{eq:keyrateMDI}
R_{MDI} = \frac{1}{2} Q_{\mu\mu}\{\bar{q}_{11} [1 - H(e^X_{(1,1)})] - f H(E^Z_{(\mu,\mu)}) \},
\end{equation}
where the factor $\frac{1}{2}$ is for a comparison with the MP scheme, since the key-rate formula of $R_{MDI}$ is defined on a pair of rounds while the key-rate formula Eq.~\eqref{eq:keyrate} for the MP scheme is defined on each round. $Q_{\mu\mu}$ is the probability of the pair with successful detection. The parameters $\bar{q}_{11}, e^X_{(1,1)}$ and $E^Z_{(\mu,\mu)}$ are similar to those in the MP scheme.

Similar to Sec.~\ref{Ssec:MPsimu}, we also introduce $z_i = [\zeta^a_i, z^b_i]$ and $C_i$ to denote the intensity setting of the $i$-th location and detection results, respectively. The single-location detection probability of a given $z_i$ is the same as that in the MP scheme,
\begin{equation} \label{eq:pulsexiQ2}
\Pr(C_i=1|z_i) \approx 1 - (1-2p_d) \exp\left[-\eta_s\mu\left(z^a_i + z^b_i\right)\right].
\end{equation}

Meanwhile, we introduce $n_i=[n^a_i,n^b_i]$ to denote the emitted photon number at the $i$-th location. The conditional detection probability $\Pr(C_i=1|n_i)$ is still given by
\begin{equation} \label{eq:pulsexiY2}
\Pr(C_i=1|n_i) \approx 1 - (1-2p_d) (1 - \eta_s)^{(n_i^a + n_i^b)}.
\end{equation}

In time-bin MDI-QKD and for a fixed pair of locations $i$ and $j$, the configurations of $z_i$ and $z_j$ are predetermined to be chosen in the set
\begin{equation}
[z_i,z_j] \in \{[00,11],[01,10],[10,01],[11,00]\}.
\end{equation}
Among these $Err := \{[00,11],[11,00]\}$ are the two configurations causing bit errors. To simplify our notation, we also introduce several events,
\begin{equation}
\begin{aligned}
&\Pr(C) = \Pr(\textit{Pair Clicked}) := Q_{\mu\mu}, \\
&\Pr(Err) = \Pr(\textit{Pair Erroneous}) := \Pr([z_i,z_j]\in Err), \\
&\Pr(S) = \Pr(\textit{Single-photon Pair}) := \Pr(n_i\oplus n_j = 11).
\end{aligned}
\end{equation}

By definition we have
\begin{equation} \label{eq:MDIQmumu}
\begin{aligned}
Q_{\mu\mu} &= \Pr(C) = \sum_{z_i\oplus z_j = 11} \Pr(z_i,z_j) \Pr(C_i=1|z_i) \Pr(C_j=1|z_j) \\
&= \frac{1}{4} \sum_{z_i\oplus z_j = 11} \Pr(C_i=1|z_i) \Pr(C_j=1|z_j).
\end{aligned}
\end{equation}

The quantum bit-error rate $E^Z_{(\mu,\mu)}$, is
\begin{equation} \label{eq:MDIEZmumu}
\begin{aligned}
E^Z_{(\mu,\mu)} &= \Pr(Err|C) = \frac{\Pr(Err,C)}{\Pr(C)} \\
&= \frac{1}{Q_{\mu\mu}} \sum_{[z_i,z_j]\in Err} \Pr(z_i,z_j) \Pr(C_i=1|z_i) \Pr(C_j=1|z_j) \\
&= \frac{1}{4}\frac{1}{Q_{\mu\mu}} \sum_{[z_i,z_j]\in Err} \Pr(C_i=1|z_i) \Pr(C_j=1|z_j).
\end{aligned}
\end{equation}

The ratio of the single-photon pairs in the signals with successful detection $\bar{q}_{11}$ is
\begin{equation} \label{eq:MDIq11}
\begin{aligned}
\bar{q}_{11} &= \Pr(S|C) = \frac{\Pr(S,C)}{\Pr(C)} \\
&= \frac{1}{Q_{\mu\mu}} \sum_{z_i\oplus z_j =11} \Pr(z_i,z_j) \Pr(S,C|z_i,z_j) \\
&= \frac{1}{4} \frac{1}{Q_{\mu\mu}} \sum_{z_i\oplus z_j =11} \Pr(C|S,z_i,z_j) \Pr(S|z_i,z_j) \\
&= \frac{1}{4} \frac{P_\mu(1)^2}{Q_{\mu\mu}} \sum_{z_i\oplus z_j =11} \Pr(C_i=1|n_i=z_i)\Pr(C_j=1|n_j=z_j),
\end{aligned}
\end{equation}
where $P_\mu(k) = \exp(-\mu)\frac{\mu^k}{k!}$ is the Poisson distribution and $\Pr(C_i=1|n_i)$ is defined in Eq.~\eqref{eq:pulsexiY2}.

The phase-error-rate formula of time-bin MDI-QKD is given by the formulas (A9) and (A11) in Ref.~\cite{Ma2012alternative},
\begin{equation} \label{eq:e11X2}
\begin{aligned}
Y_{11} & = (1-p_d)^2 \left[ \dfrac{\eta_a\eta_b}{2} + (2\eta_a + 2\eta_b -3\eta_a\eta_b)p_d + 4(1-\eta_a)(1-\eta_b)p_d^2 \right], \\
e_{11}^X & = e_0Y_{11} - (e_0- e_d)(1-p_d^2)\dfrac{\eta_a\eta_b}{2}.
\end{aligned}
\end{equation}
Here, $e_0=1/2$ is the error caused by vacuum signals, $e_d$ is the pre-set misalignment error, and $\eta_a(\eta_b)$ are the transmittances from Alice (Bob) to Charlie. In our case, $\eta_a = \eta_b = \eta_s$.

\subsection{Other QKD schemes} \label{Ssec:othersimu}

Here, we list the simulation formulas for decoy-state BB84 and PM-QKD which are used in the key rate simulation programme. We also list the repeaterless rate-transmittance bound drawn in the main text.

The key rate of the decoy-state BB84 scheme is \cite{Lo2005Decoy}
\begin{equation}
R_{BB84} = \frac{1}{2} Q_\mu\{- f H(E_{\mu}^Z) + q_1[1 - H(e^X_1)] \}.
\end{equation}
Here, for a fair comparison with the time-bin MDI-QKD and the MP scheme, we consider a time-bin encoding BB84 scheme, where Alice's $Z$-basis encoding is undertaken by choosing to emit a phase-randomised coherent state in one of two time-bin optical modes. In this case, the $Z$-basis bit error is merely caused by the dark counts, regardless of the misalignment error.
We now consider an efficient basis choice with a sifting factor of $1$. For a fair comparison, we multiply the original key rate of decoy-state BB84 by a factor of $1/2$. The reason is that we are considering the averaged key rate generated by each pair of optical modes, $A_i$ and $B_i$, while the key rate for the original decoy-state BB84 is found for each round in which Alice (and Bob) emit two optical modes.

In the simulation, the yield and error rates of the $k$-photon component in the $X$-basis are \cite{Ma2008PhD}
\begin{equation} \label{eq:YemuBB84}
\begin{aligned}
Y_k &= 1-(1-Y_0)(1-\eta)^{k}, \\
e_k &= e_d + \frac{(e_0-e_d)Y_0}{Y_k}, \\
\end{aligned}
\end{equation}
where $e_d$ is the intrinsic misalignment-error rate caused by a phase-reference mismatch. The gain and QBER of the $Z$-basis pulse are given by
\begin{equation} \label{eq:QEmuBB84}
\begin{aligned}
Q_\mu &= \sum_{k=0}^{\infty} \frac{\mu^k e^{-\mu}}{k!} Y_k, \\
&= 1-(1-Y_0)e^{-\eta\mu}, \\
E_\mu 
&= \frac{e_d Y_0}{Q_\mu}, \\
\end{aligned}
\end{equation}
where $Y_0=2p_d$ and $e_0 = 1/2$.

For PM-QKD, the key-rate formula is \cite{Ma2018phase,Zeng2019Symmetryprotected},
\begin{equation} \label{eq:groupkey}
R_{PM} = \dfrac{2 Q_{\mu}}{D}\sum_j \left[ 1 - H(E_{\mu}^X) - f H(E_{\mu, j}^Z) \right],
\end{equation}
where $D=16$ is the number of slices. The gain, yield and error rates are
\begin{equation}
\begin{aligned}
Q_{\mu} & = 1-(1 - 2p_d)e^{-\eta\mu}, \\
Y_k & = 1 - (1 - 2p_d)(1-\eta)^k, \\
q_k & = Y_k\frac{e^{-\mu}\mu^k/k!}{Q_{\mu}}, \\
E_{\mu}^X & = 1 - \sum_k q_{2k+1} < 1 - q_1- q_3 - q_5, \\
E_{\mu, j}^Z & = [p_d+\eta\mu (e_{\delta, j} + e_d^{PM})]\frac{e^{-\eta\mu}}{Q_{\mu}}, \\
e_{\delta, j} & = \sin^2\left[\frac{\pi}{4}-\left|\frac{\pi}{4}-\frac{\pi j}{D}\right|\right],
\end{aligned}
\end{equation}
where $e_d^{PM}$ is the intrinsic misalignment error of the PM-QKD. We take these formulas from Eqs.~(8)-(12) and (D1)-(D4) in Ref.~\cite{Zeng2019Symmetryprotected}.

For SNS-TF-QKD, the key-rate formula is~\cite{wang2018twin,jiang2019unconditional},
\begin{equation}
R_{SNS} = 2 p_{z0} (1 - p_{z0}) \mu_z e^{-\mu_z} s_1 \left[ 1 - H(e_1^{ph}) \right] - f S_z H(E^Z),
\end{equation}
where $p_{z0}$ is Alice's (or Bob's) probability of sending vaccum state, $\mu_z$ is the light intensity for Alice or Bob's signal pulses, $s_1$ is the detecting probability of single-photon signals, $e_1^{ph}$ is the single-photon phase error rate, $S_z$ is the detecting probability of signal lights and $E^Z$ is the quantum bit error rate. $p_{z0}$ and $\mu_z$ need to be optimized. We use the following simulation formula, matching the results in Ref.~\cite{jiang2019unconditional},
\begin{equation}
\begin{aligned}
s_1 &= (1-\eta) 2p_d(1-p_d) + \eta(1-p_d), \\
e_1^{ph} &= (1-e_d^{SNS})\frac{(1-\eta)p_d(1-p_d)}{s_1} + e_d^{SNS}\left[1- \frac{(1-\eta)p_d(1-p_d)}{s_1} \right], \\
S_z &= p_{z0}^2 S_{00} + (1-p_{z0})^2 S_{22} + 2p_{z0}(1-p_{z0}) S_{02},\\
E^Z &= \left[p_{z0}^2 S_{00} + (1-p_{z0})^2 S_{22}\right]/S_z.
\end{aligned}
\end{equation}
Here, $e_d^{SNS}$ is the misalignment error. $S_{00}$, $S_{02}$, and $S_{22}$ indicate the detecting probabilities of the signals when Alice and Bob emit light with intensity $(0,0)$, $(0,\mu_z)$, and $(\mu_z,\mu_z)$, respectively, which are given by the following formulas,
\begin{equation}
\begin{aligned}
S_{00} &= 2p_d (1-p_d), \\
S_{02} &= 2\left[ (1-p_d)e^{-\eta\mu_z/2} - (1-p_d)^2 e^{-\eta\mu_z} \right], \\
S_{22} &= 2\left[ (1-p_d)e^{-\eta\mu_z}I_0(\eta\mu_z) - (1-p_d)^2 e^{-2\eta\mu_z} \right],
\end{aligned}
\end{equation}
where $I_0(x)$ is the modified Bessel function of the first kind.

The repeaterless rate-transmittance bound of the repeaterless point-to-point QKD scheme is the PLOB bound~\cite{pirandola2017fundamental}
\begin{equation}\label{eq:plob}
R_{PLOB} = - \log_2(1-\eta).
\end{equation}

\section{Numerical results for the optimal intensity} \label{Sec:optimalmu}

To further understand the key rate property of the MP scheme, we simulate this scheme's optimal intensity settings, as well as those of the other schemes mentioned  in the main text. We choose exactly the same simulation parameters as used in Fig.~4 in the main text. Fig.~\ref{fig:MDIfamIntComp} illustrates the optimal intensities of different QKD schemes; Fig.~\ref{fig:MPIntComp} illustrates those corresponding to different communication distances and pairing lengths $l$.

We first note that, for the time-bin-encoding version of BB84, the two-mode MDI-QKD and the MP scheme, the $Z$-basis bit error is almost $0$ and is irrelevant to the intensity $\mu$. In this case, the error-correction-related term has a negligible contribution in the key rate. For BB84, the optimal intensity satisfies (cf. Eq.~(12) in Ref.~\cite{Ma2005practical})
\begin{equation}
(1-\mu)e^{-\mu} = \frac{f H(E^Z_\mu)}{1 - H(e^X_1)}.
\end{equation}
When $f H(E^Z_\mu)\to 0$, we have $(1-\mu)e^{-\mu}\to 0$; hence, $\mu\to 1$. The optimal $\mu$ for MDI-QKD can be derived in a similar fashion. As shown in Fig.~\ref{fig:MDIfamIntComp}, before the error rates caused by the dark-count events come to dominate the bit-error rate $E^Z$, the optimal $\mu$ values for BB84 and MDI-QKD are always $1$.

The optimal $\mu$ for the MP scheme, however, cannot be directly derived from the key-rate formula due to an additional dependence of the key rate $R$ on the pairing rate $r_p$. If we set $p_d \to 0$ and assume that $\eta_s\mu\ll 1$, the key rate of the MP scheme in Eq.~\eqref{eq:keyrate} is approximately
\begin{equation}
\begin{aligned}
R &= r_p(p,l) r_s \{ q_{(1,1)} [1 - H(e^X_{(1,1)})] - fH(E^Z_{(\mu,\mu)})\} \\
&\approx p\left[\frac{1}{1 - (1-p)^l}\right]^{-1} \frac{1}{8} \{ e^{-2\mu} [1 - H(e^X_{(1,1)})] - fH(E^Z_{(\mu,\mu)}) \}.
\end{aligned}
\end{equation}
Here,
\begin{equation}
\begin{aligned}
r_s &\approx \frac{1}{16} \frac{1}{p^2} [2 (\eta_s\mu)^2 ] = \frac{1}{8},\\
q_{(1,1)} &\approx \frac{1}{16} \frac{ (\mu e^{-\mu})^2 }{r_s (\eta_s\mu)^2} 2\eta^2_s = e^{-2\mu}.
\end{aligned}
\end{equation}

Now, we consider two extreme cases: 1) $l$ is small, such that $pl\ll 1$; 2) $l$ is large, such that $pl \gg 1$. Note that, we always assume that $1\gg p\gg p_d$.

When $pl \ll 1$, we have $(1 - p)^l \approx 1-lp$. The key rate can be simplified to be
\begin{equation}
\begin{aligned}
R &\approx \frac{1}{8}lp^2\{ e^{-2\mu} [1 - H(e^X_{(1,1)})] - fH(E^Z_{(\mu,\mu)}) \}, \\
&= \frac{1}{8}l(\eta_s \mu)^2\{ e^{-2\mu} [1 - H(e^X_{(1,1)})] - fH(E^Z_{(\mu,\mu)}) \}.
\end{aligned}
\end{equation}
Taking the derivative of $R$ with respective to $\mu$, we obtain the following condition for optimal $\mu$:
\begin{equation}
(1-\mu) e^{-2\mu} = \frac{fH(E^Z_{(\mu,\mu)})}{1 - H(e^X_{(1,1)})}.
\end{equation}
Suppose that $E^Z_{(\mu,\mu)}\to 0$; then $\mu\to 1$. In the long-distance regime, $E^Z_{(\mu,\mu)}$ gets larger, resulting in a smaller $\mu$.

When $pl \gg 1$, we have $(1 - p)^l \approx 0$. The key rate can be simplified as
\begin{equation}
\begin{aligned}
R &\approx \frac{1}{16}p\{ e^{-2\mu} [1 - H(e^X_{(1,1)})] - fH(E^Z_{(\mu,\mu)}) \}, \\
&= \frac{1}{16}\eta_s\mu\{ e^{-2\mu} [1 - H(e^X_{(1,1)})] - fH(E^Z_{(\mu,\mu)}) \}.
\end{aligned}
\end{equation}
Taking the derivative of $R$ with respective to $\mu$, we obtain the following condition for optimal $\mu$:
\begin{equation}
(1-2\mu) e^{-2\mu} = \frac{fH(E^Z_{(\mu,\mu)})}{1 - H(e^X_{(1,1)})}.
\end{equation}
Suppose $E^Z_{(\mu,\mu)}\to 0$; then $\mu\to 1/2$. In the long-distance regime, $E^Z_{(\mu,\mu)}$ increases, resulting in a smaller $\mu$.

In the general case, for a fixed $l$, the value of $pl$ grows smaller as the communication distance becomes longer. The optimal value of $\mu$ may first be close to $0.5$ and then increased to nearly $1$.

\begin{figure}[htbp]
\includegraphics[width=10cm]{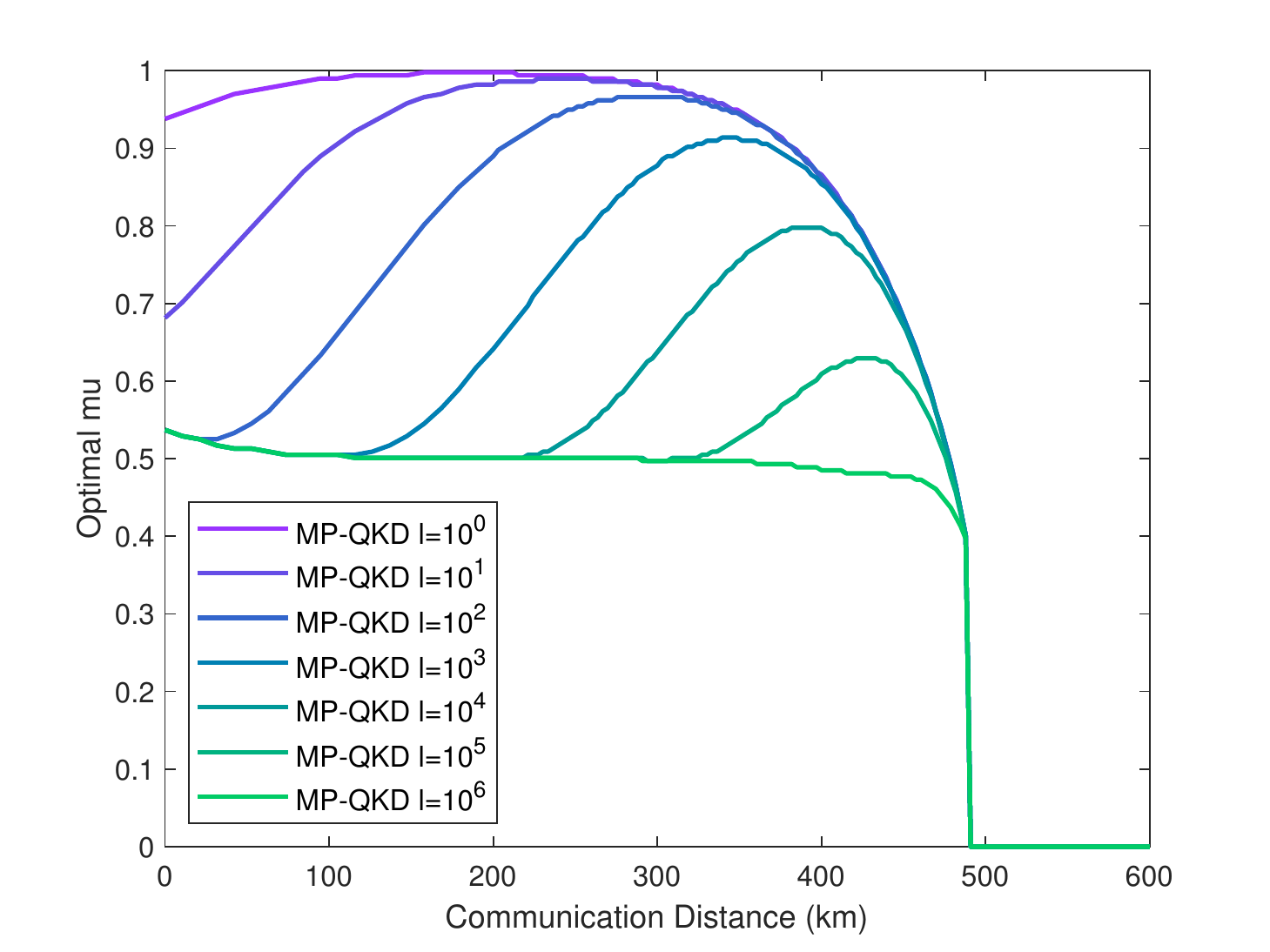}
\caption{Optimal intensity setting of the MP scheme with different pairing lengths $l$. Here, all of the simulation parameters are chosen to be the same as the ones in Fig.~4 in the main text.} \label{fig:MPIntComp}
\end{figure}

\begin{figure}[htbp]
\includegraphics[width=10cm]{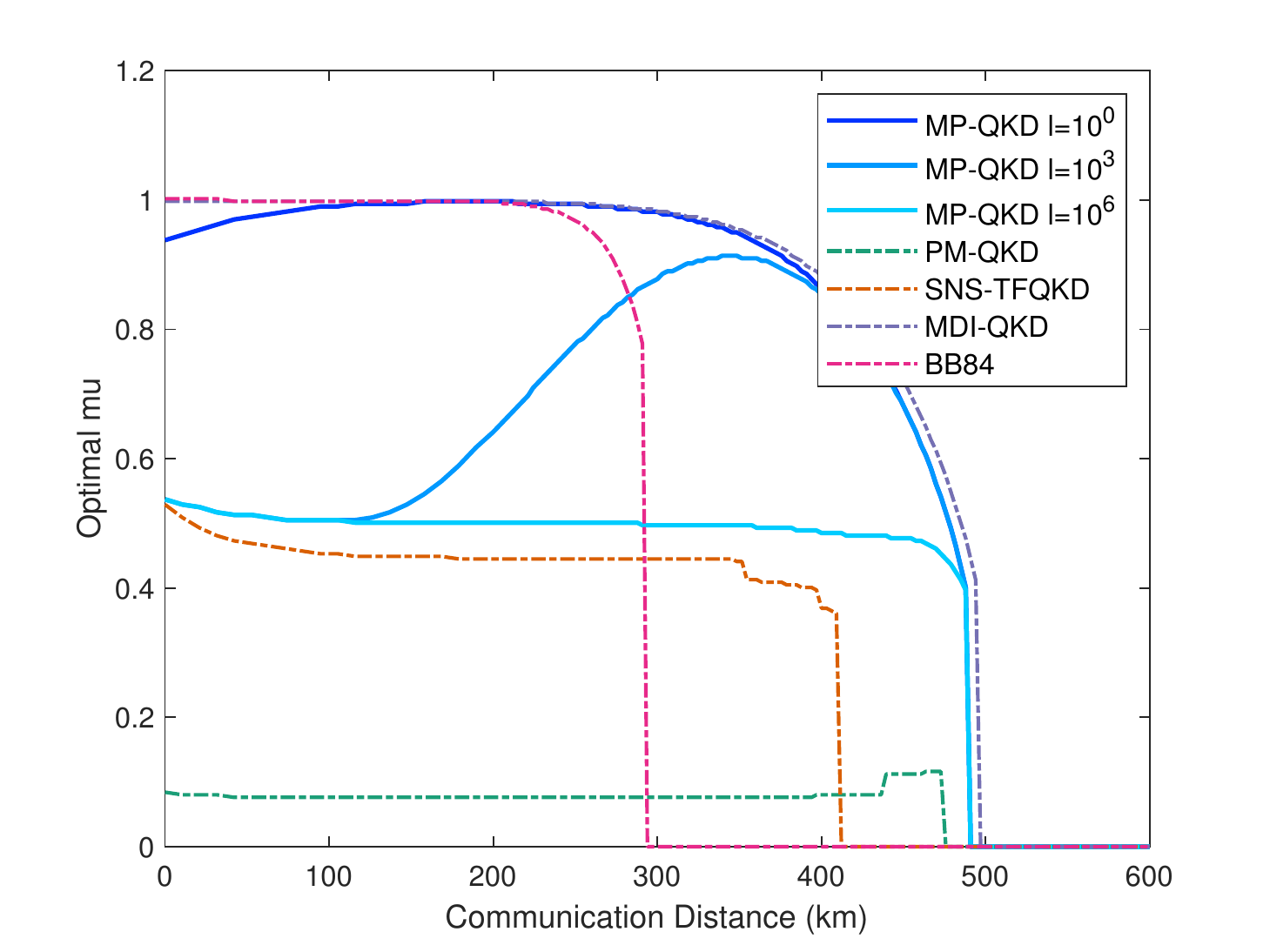}
\caption{Comparison of the optimal intensity settings of different QKD schemes. Here, all the simulation parameters are chosen to be the same as those in Fig.~4 in the main text. Here, the optimal $\mu$ for the BB84 scheme is defined to be the overall intensity of two emitted modes in each round. The optimal $\mu$ value for (two-mode) MDI-QKD is defined as the overall intensity of Alice's (or Bob's) two emitted modes in each round. The optimal $\mu$ values for MP-QKD, PM-QKD, and SNS-TFQKD are defined on each single round.} \label{fig:MDIfamIntComp}
\end{figure}

\section{Proof-of principle experimental demonstration}\label{sc:ExpDemo}

In the mode-pairing scheme, like other phase-encoding QKD schemes, certain ``local phase stabilization" is still required. In this section, we want to clarify the difference between this local phase stabilization and the ``global phase locking" over two remote independent lasers required for the one-mode MDI-QKD schemes. 

Suppose Alice and Bob hold two independent lasers emitting the optical pulses independently. One can track the phases of these coherent state pulses. For the pulses emitted by Alice and Bob at the time $t$, denote the global phases of the pulses as $\phi^a(t)$ and $\phi^b(t)$, respectively. Generally, the circular frequencies of Alice's and Bob's lasers, $\omega_a(t)$ and $\omega_b(t)$, are changing over time. We model the phases $\phi^a(t)$ and $\phi^b(t)$ as follows,
\begin{equation}\label{eq:phiab}
	\begin{aligned}
		\phi^a(t)=\int_0^t \omega_a(t)dt+ \psi_a(t), \\
		\phi^b(t)=\int_0^t \omega_b(t)dt+ \psi_b(t),
	\end{aligned}
\end{equation}
where $\psi_a(t)$ and $\psi_b(t)$ are two randomly fluctuated noise terms with a rate much slower than the one caused by circular frequencies and independent of $\omega_a$ and $\omega_b$. Here, the phase fluctuation caused by fibres is normally slow, based on the former experimental result. When the fibre length is around $500$ km, the phase drift velocity is less than $10$ rad/ms~\cite{fang2020implementation}. Then, we define global phase locking and local phase stabilization as following.

\begin{definition}[Global phase locking]\label{def:globallock}
After Alice and Bob perform global phase locking, they can achieve either of the following two tasks during the whole process of the experiment. 
\begin{itemize}
\item
They keep the relative phase between the two emitted pulses $\phi^b(t)-\phi^a(t)$ to be a constant, independent of the time $t$.
\item
They are able to estimate the phase difference $\phi^b(t)-\phi^a(t)$ for every moment $t$ accurately.
\end{itemize}
\end{definition}

If we want to perform global phase locking, we have to stabilize the phase difference $\phi^b(t)-\phi^a(t)$ for each moment $t$ or accurately estimate it. 
Due to the large frequency difference $\omega_b(t)-\omega_a(t)$ of two independent lasers, the value of $\phi^b(t)-\phi^a(t)$ changes rapidly with respect to $t$, making the stabilization or accurate estimation of $\phi^b(t)-\phi^a(t)$ challenging. 


\begin{definition}[Local phase stabilization]\label{def:localstable}
	After Alice and Bob perform local phase locking, they can achieve either of the following two tasks for any two moments $t_1$ and $t_2$, satisfying $|t_1-t_2|\le \Delta t$,
	\begin{itemize}
		\item
		Alice keeps the relative phase between her two emitted pulses $\phi^a(t_2)-\phi^a(t_1)$ to be a constant, independent of $t_1$ and $t_2$. Bob does the same thing for $\phi^b(t_2)-\phi^b(t_1)$.
		\item
		They are able to estimate the phase difference $\theta^{\delta}(t_1,t_2):=\phi^b(t_2)-\phi^b(t_1)-\phi^a(t_2)+\phi^a(t_1)$ accurately.
	\end{itemize}
\end{definition}

We remark that, the core difference of the local phase stabilization from the global phase locking is to estimate the \emph{difference} of the phase $\phi^b(t)-\phi^a(t)$ between two moments $t_1$ and $t_2$ instead of estimating the phase itself. We now discuss why this task will be easier than the global phase locking for a reasonable time period $\Delta t=t_2 - t_1$. 

In the task of local phase stabilization, the users try to stabilize
\begin{equation}\label{eq:relativephaseab}
	\begin{aligned}
\theta^{\delta}(t_1,t_2) = [\phi^b(t_2)-\phi^b(t_1)]-[\phi^a(t_2)-\phi^a(t_1)]=\int_{t_1}^{t_2}[\omega_b (t)-\omega_a (t)]dt+[\psi_b(t_2)-\psi_b(t_1)]-[\psi_a(t_2)-\psi_a(t_1)]
	\end{aligned}
\end{equation} 
within time $\Delta t=t_2-t_1$. 
For a short time period $\Delta t$, $\psi_a(t_2)-\psi_a(t_1)$ and $\psi_b(t_2)-\psi_b(t_1)$ will be relatively small.
To estimate the value of $\theta^{\delta}(t_1,t_2)$ in a short period $\Delta t$, Alice and Bob only need to track the frequency difference $\omega_b(t)-\omega_a(t)$, which is easier than directly measure the phase difference $\phi^b(t)-\phi^a(t)$. Then, they use the integration $\int_{t_1}^{t_2}[\omega_b (t)-\omega_a (t)]dt$ to estimate the phase difference $\theta^{\delta}(t_1,t_2)$.
On the other hand, when the time period $\Delta t$ gets longer, the accuracy of integration $\int_{t_1}^{t_2}[\omega_b (t)-\omega_a (t)]dt$ will be affected, and the slow phase fluctuation $\psi_a(t_2)-\psi_a(t_1)$ and $\psi_b(t_2)-\psi_b(t_1)$ would come into play, which would affect the final results of local phase stabilization. 
The hardness of local phase stabilization requirement depends on the time interval $\Delta t$. We now consider the following two extreme cases.

\begin{enumerate}
\item When $\Delta t = t_2-t_1 \ll 1/\Delta f$, where $\Delta f$ is the laser linwidth. For example, in the regular phase-encoding MDI-QKD, the $\Delta t$ of local phase stabilization requirement is the time between two adjacent pulses, which is usually in this scenario. In this case, the values of $\omega_b(t)-\omega_a(t)$ is stable, and the value of $\psi_a(t_2)-\psi_a(t_1)$ and $\psi_b(t_2)-\psi_b(t_1)$ will be close to zero. As a result, the local phase stabilization is easy to realize. 
\item When $\Delta t = t_2-t_1 \gg 1/\Delta f$. Then, the frequency difference $\omega_b (t)-\omega_a (t)$ fluctuates randomly, and the values of $\psi_a(t_2)-\psi_a(t_1)$ and $\psi_b(t_2)-\psi_b(t_1)$ become large. To fulfill the requirements in Definition~\ref{def:localstable}, we have to estimate both the phases $\phi^b(t_1) - \phi^a(t_1)$ and $\phi^b(t_2)- \phi^a(t_2)$ accurately. In this case, the local phase stabilization approaches the case of global phase locking. 
\end{enumerate}

In order to show the feasibility of realizing the local phase stabilization in Definition~\ref{def:localstable} with a reasonable large time period $\Delta t$, we perform an experimental demonstration using two independent lasers. This is the essentially the same setup as the phase-matching experiment \cite{fang2020implementation}, except for the removal of the strong master laser for global phase-locking. Alice and Bob each use a laser with the central frequency of $\omega_a\approx\omega_b\approx 193.533$ THz. The linewidth $\Delta f$ of the lasers is $2$ kHz. The frequency difference of two lasers is about $30$ MHz, which is much larger than the linewidth. The system frequency is $625$ MHz. 

In the experiment, Alice and Bob both send laser pulses of intensity $\mu$ to a measurement site in the middle. Depending on the transmission distance, they will choose a proper $\mu$ such that they will get a suitable number of detection results. Besides, they do not perform any phase modulation on these pulses for simplicity. The coherent state prepared by Alice and Bob can be denoted as $\ket{\sqrt{\mu}e^{i\phi^{a}(t)}}$ and $\ket{\sqrt{\mu}e^{i\phi^{b}(t)}}$, respectively. 
The two pulses are interfered at the measurement site and then are detected by single-photon detectors. The detection result is ``$L$'' if one of the detectors clicks, and ``$R$'' if the other detector clicks.

Now, they do not apply any phase-locking techniques in the demonstration, it is very challenging to track the global phases and predict the measurement result in single round. However, in mode-pairing scheme, we only care about whether the measurement results of the two paired locations are the same or not. Which is determined by the phase difference betweem two paired locations $i$ and $j$,
\begin{equation}\label{eq:phasediff}
	\theta^\delta = (\phi^b_j - \phi^b_i) - (\phi^a_j - \phi^a_i) = \int_{t_i}^{t_j} \omega_\delta(t) dt,
\end{equation}
where $\omega_\delta(t):=\omega_b(t) - \omega_a(t)$ is the frequency difference between Alice's and Bob's laser pulses. For a short period of time, we can assume that $\omega_\delta(t)$ drifts with a linear model,
\begin{equation} \label{eq:omegaDelta}
	\omega_\delta(t) = k t + \omega_\delta(0).
\end{equation}
After accumulating enough detection results, Alice and Bob can use them to estimate the frequency difference of their lasers during data post-processing. First, they pair the locations where the successful detection occurs. For a pair, there are four possible detection results, $(L,R)$, $(R,L)$, $(L,L)$ and $(R,R)$. One click on location $i$ and the other on $j$. With the paired locations and corresponding detection results, Alice and Bob can estimate $\omega_\delta(t)$ using a probabilistic model corresponding to this optical setting. Then they can calculate the phase difference $\theta^\delta$ using Eq.~\eqref{eq:phasediff}.

To test the accuracy of the estimation of $\omega_\delta(t)$ in Eq.~\eqref{eq:omegaDelta}, we further use the estimated $\theta^\delta$ to predict the detection results. When the estimated value of $\theta^\delta$ of a pair is in the region $[-\Delta/2, +\Delta/2)$ or $[\pi-\Delta/2, \pi+\Delta/2)$, this data pair can be further used for testing. Here, $\Delta = 2\pi/32$ is the width of the phase slice. In the former case, Alice and Bob assign the detection $(L,R)$ or $(R,L)$ to be error and $(L,L)$ or $(R,R)$ to be correct; while in the latter case, they assign the detection $(L,L)$ or $(R,R)$ to be error and $(L,R)$ or $(R,L)$ to be correct.

Below we draw the error rate of the paired strong pulses with respect to different pairing length $l=j-i$. We group the paired locations by different pairing length, as shown in Figure \ref{fig:ErrorVSL}. This error rate can reflect the $X$-basis error rate in mode-pairing QKD. Similar to the phase-encoding MDI-QKD, the $X$-basis error has a 25\% intrinsic error rate caused by the multi-photon components in the coherent states. From the demonstration we can see that, apart from the intrinsic error, our method only introduces a reasonable additional error rate. One can further reduce the error rate by post-selecting the good signals in post-processing of reference pulses. If we use the decoy method to estimate the phase error in a single photon pair, the error rate will be a reasonable value that we can obtain a positive final key rate. Therefore, the mode-pairing scheme is feasible without global phase locking.

\begin{figure}[htbp!]
	\centering \includegraphics[width=8cm]{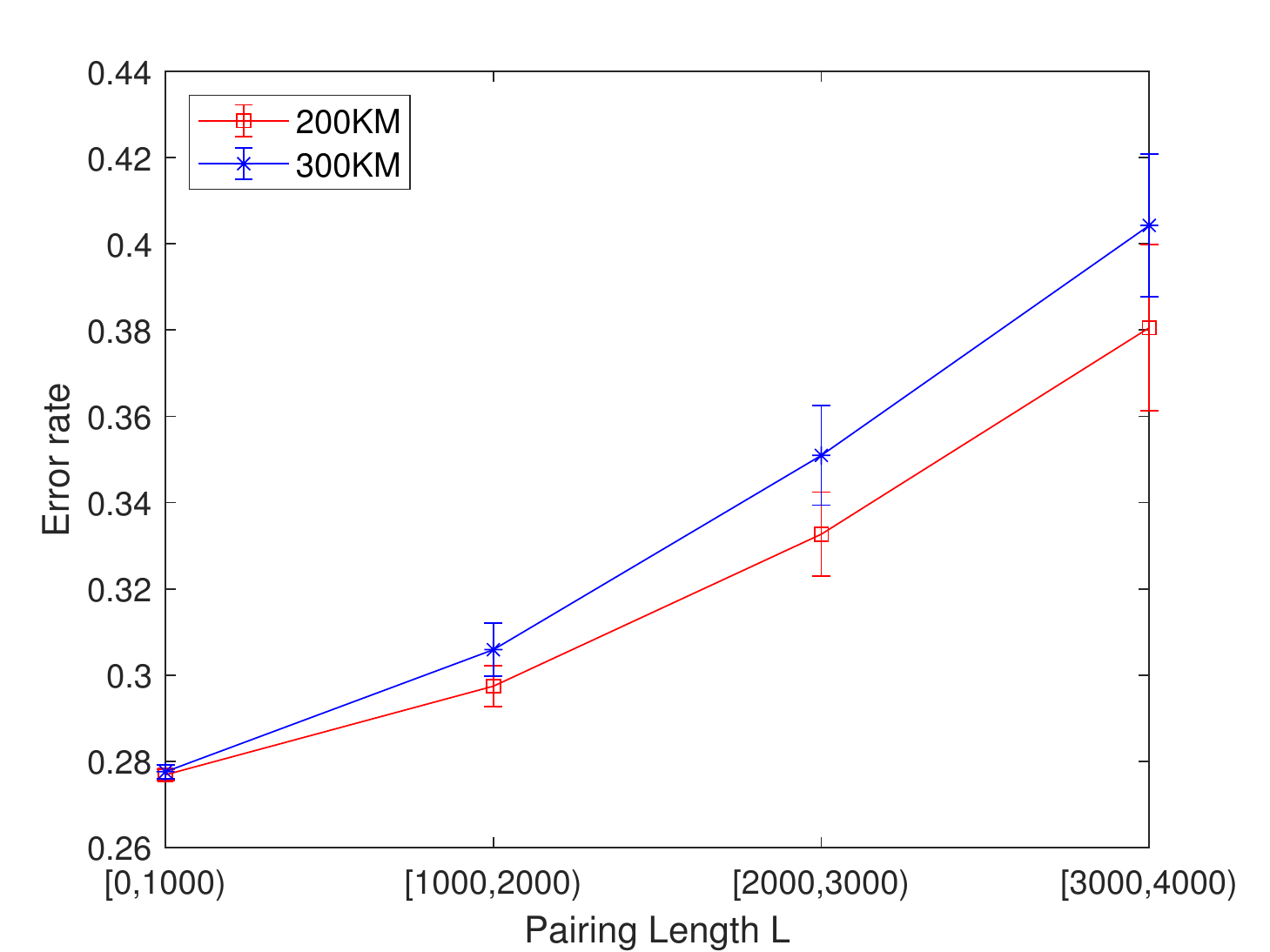}
	\caption{Error rates of the paired laser pulses under different pairing intervals and different transmission distances. The intensities of laser pulses are chosen such that the arrived signals on the measurement site are strong enough to provide enough clicks. Given a pairing length interval, say $l\in[1000,2000)$, we pair all detection clicks such that $1000\le j-i <2000$.} \label{fig:ErrorVSL}
\end{figure}

From Figure \ref{fig:ErrorVSL}, we can see that the error rate is stable when the pairing length is up to $3000\sim 4000$. If we further enhance the system frequency to $4$ GHz~\cite{wang2022twin}, we are able to achieve a maximal pairing length of 20000. In this way, the performance of the mode-pairing scheme can surpass the repeaterless key-rate bound. We leave a full demonstration of the mode-pairing scheme for future works.

\end{appendix}


\begin{thebibliography}{59}%
\makeatletter
\providecommand \@ifxundefined [1]{%
 \@ifx{#1\undefined}
}%
\providecommand \@ifnum [1]{%
 \ifnum #1\expandafter \@firstoftwo
 \else \expandafter \@secondoftwo
 \fi
}%
\providecommand \@ifx [1]{%
 \ifx #1\expandafter \@firstoftwo
 \else \expandafter \@secondoftwo
 \fi
}%
\providecommand \natexlab [1]{#1}%
\providecommand \enquote  [1]{``#1''}%
\providecommand \bibnamefont  [1]{#1}%
\providecommand \bibfnamefont [1]{#1}%
\providecommand \citenamefont [1]{#1}%
\providecommand \href@noop [0]{\@secondoftwo}%
\providecommand \href [0]{\begingroup \@sanitize@url \@href}%
\providecommand \@href[1]{\@@startlink{#1}\@@href}%
\providecommand \@@href[1]{\endgroup#1\@@endlink}%
\providecommand \@sanitize@url [0]{\catcode `\\12\catcode `\$12\catcode
  `\&12\catcode `\#12\catcode `\^12\catcode `\_12\catcode `\%12\relax}%
\providecommand \@@startlink[1]{}%
\providecommand \@@endlink[0]{}%
\providecommand \url  [0]{\begingroup\@sanitize@url \@url }%
\providecommand \@url [1]{\endgroup\@href {#1}{\urlprefix }}%
\providecommand \urlprefix  [0]{URL }%
\providecommand \Eprint [0]{\href }%
\providecommand \doibase [0]{https://doi.org/}%
\providecommand \selectlanguage [0]{\@gobble}%
\providecommand \bibinfo  [0]{\@secondoftwo}%
\providecommand \bibfield  [0]{\@secondoftwo}%
\providecommand \translation [1]{[#1]}%
\providecommand \BibitemOpen [0]{}%
\providecommand \bibitemStop [0]{}%
\providecommand \bibitemNoStop [0]{.\EOS\space}%
\providecommand \EOS [0]{\spacefactor3000\relax}%
\providecommand \BibitemShut  [1]{\csname bibitem#1\endcsname}%
\let\auto@bib@innerbib\@empty
\bibitem [{\citenamefont {Bennett}\ and\ \citenamefont
  {Brassard}(1984)}]{bennett1984quantum}%
  \BibitemOpen
  \bibfield  {author} {\bibinfo {author} {\bibfnamefont {C.~H.}\ \bibnamefont
  {Bennett}}\ and\ \bibinfo {author} {\bibfnamefont {G.}~\bibnamefont
  {Brassard}},\ }\bibfield  {title} {\bibinfo {title} {{Quantum Cryptography:
  Public Key Distribution and Coin Tossing}},\ }in\ \href
  {https://doi.org/10.1016/j.tcs.2014.05.025} {\emph {\bibinfo {booktitle}
  {Proceedings of the IEEE International Conference on Computers, Systems and
  Signal Processing}}}\ (\bibinfo  {publisher} {IEEE Press},\ \bibinfo
  {address} {New York},\ \bibinfo {year} {1984})\ pp.\ \bibinfo {pages}
  {175--179}\BibitemShut {NoStop}%
\bibitem [{\citenamefont {Ekert}(1991)}]{ekert1991Quantum}%
  \BibitemOpen
  \bibfield  {author} {\bibinfo {author} {\bibfnamefont {A.~K.}\ \bibnamefont
  {Ekert}},\ }\bibfield  {title} {\bibinfo {title} {Quantum cryptography based
  on bell's theorem},\ }\href {https://doi.org/10.1103/PhysRevLett.67.661}
  {\bibfield  {journal} {\bibinfo  {journal} {Phys. Rev. Lett.}\ }\textbf
  {\bibinfo {volume} {67}},\ \bibinfo {pages} {661} (\bibinfo {year}
  {1991})}\BibitemShut {NoStop}%
\bibitem [{\citenamefont {Chen}\ \emph {et~al.}(2021)\citenamefont {Chen},
  \citenamefont {Zhang}, \citenamefont {Chen}, \citenamefont {Cai},
  \citenamefont {Liao}, \citenamefont {Zhang}, \citenamefont {Chen},
  \citenamefont {Yin}, \citenamefont {Ren}, \citenamefont {Chen}, \citenamefont
  {Han}, \citenamefont {Yu}, \citenamefont {Liang}, \citenamefont {Zhou},
  \citenamefont {Yuan}, \citenamefont {Zhao}, \citenamefont {Wang},
  \citenamefont {Jiang}, \citenamefont {Zhang}, \citenamefont {Liu},
  \citenamefont {Li}, \citenamefont {Shen}, \citenamefont {Cao}, \citenamefont
  {Lu}, \citenamefont {Shu}, \citenamefont {Wang}, \citenamefont {Li},
  \citenamefont {Liu}, \citenamefont {Xu}, \citenamefont {Wang}, \citenamefont
  {Peng},\ and\ \citenamefont {Pan}}]{Chen2021integrated}%
  \BibitemOpen
  \bibfield  {author} {\bibinfo {author} {\bibfnamefont {Y.-A.}\ \bibnamefont
  {Chen}}, \bibinfo {author} {\bibfnamefont {Q.}~\bibnamefont {Zhang}},
  \bibinfo {author} {\bibfnamefont {T.-Y.}\ \bibnamefont {Chen}}, \bibinfo
  {author} {\bibfnamefont {W.-Q.}\ \bibnamefont {Cai}}, \bibinfo {author}
  {\bibfnamefont {S.-K.}\ \bibnamefont {Liao}}, \bibinfo {author}
  {\bibfnamefont {J.}~\bibnamefont {Zhang}}, \bibinfo {author} {\bibfnamefont
  {K.}~\bibnamefont {Chen}}, \bibinfo {author} {\bibfnamefont {J.}~\bibnamefont
  {Yin}}, \bibinfo {author} {\bibfnamefont {J.-G.}\ \bibnamefont {Ren}},
  \bibinfo {author} {\bibfnamefont {Z.}~\bibnamefont {Chen}}, \bibinfo {author}
  {\bibfnamefont {S.-L.}\ \bibnamefont {Han}}, \bibinfo {author} {\bibfnamefont
  {Q.}~\bibnamefont {Yu}}, \bibinfo {author} {\bibfnamefont {K.}~\bibnamefont
  {Liang}}, \bibinfo {author} {\bibfnamefont {F.}~\bibnamefont {Zhou}},
  \bibinfo {author} {\bibfnamefont {X.}~\bibnamefont {Yuan}}, \bibinfo {author}
  {\bibfnamefont {M.-S.}\ \bibnamefont {Zhao}}, \bibinfo {author}
  {\bibfnamefont {T.-Y.}\ \bibnamefont {Wang}}, \bibinfo {author}
  {\bibfnamefont {X.}~\bibnamefont {Jiang}}, \bibinfo {author} {\bibfnamefont
  {L.}~\bibnamefont {Zhang}}, \bibinfo {author} {\bibfnamefont {W.-Y.}\
  \bibnamefont {Liu}}, \bibinfo {author} {\bibfnamefont {Y.}~\bibnamefont
  {Li}}, \bibinfo {author} {\bibfnamefont {Q.}~\bibnamefont {Shen}}, \bibinfo
  {author} {\bibfnamefont {Y.}~\bibnamefont {Cao}}, \bibinfo {author}
  {\bibfnamefont {C.-Y.}\ \bibnamefont {Lu}}, \bibinfo {author} {\bibfnamefont
  {R.}~\bibnamefont {Shu}}, \bibinfo {author} {\bibfnamefont {J.-Y.}\
  \bibnamefont {Wang}}, \bibinfo {author} {\bibfnamefont {L.}~\bibnamefont
  {Li}}, \bibinfo {author} {\bibfnamefont {N.-L.}\ \bibnamefont {Liu}},
  \bibinfo {author} {\bibfnamefont {F.}~\bibnamefont {Xu}}, \bibinfo {author}
  {\bibfnamefont {X.-B.}\ \bibnamefont {Wang}}, \bibinfo {author}
  {\bibfnamefont {C.-Z.}\ \bibnamefont {Peng}},\ and\ \bibinfo {author}
  {\bibfnamefont {J.-W.}\ \bibnamefont {Pan}},\ }\bibfield  {title} {\bibinfo
  {title} {An integrated space-to-ground quantum communication network over
  4,600 kilometres},\ }\href {https://doi.org/10.1038/s41586-020-03093-8}
  {\bibfield  {journal} {\bibinfo  {journal} {Nature}\ }\textbf {\bibinfo
  {volume} {589}},\ \bibinfo {pages} {214} (\bibinfo {year}
  {2021})}\BibitemShut {NoStop}%
\bibitem [{\citenamefont {Takeoka}\ \emph {et~al.}(2014)\citenamefont
  {Takeoka}, \citenamefont {Guha},\ and\ \citenamefont
  {Wilde}}]{takeoka2014fundamental}%
  \BibitemOpen
  \bibfield  {author} {\bibinfo {author} {\bibfnamefont {M.}~\bibnamefont
  {Takeoka}}, \bibinfo {author} {\bibfnamefont {S.}~\bibnamefont {Guha}},\ and\
  \bibinfo {author} {\bibfnamefont {M.~M.}\ \bibnamefont {Wilde}},\ }\bibfield
  {title} {\bibinfo {title} {Fundamental rate-loss tradeoff for optical quantum
  key distribution},\ }\href {https://www.nature.com/articles/ncomms6235}
  {\bibfield  {journal} {\bibinfo  {journal} {Nat. Commun.}\ }\textbf {\bibinfo
  {volume} {5}},\ \bibinfo {pages} {5235} (\bibinfo {year} {2014})}\BibitemShut
  {NoStop}%
\bibitem [{\citenamefont {Pirandola}\ \emph {et~al.}(2017)\citenamefont
  {Pirandola}, \citenamefont {Laurenza}, \citenamefont {Ottaviani},\ and\
  \citenamefont {Banchi}}]{pirandola2017fundamental}%
  \BibitemOpen
  \bibfield  {author} {\bibinfo {author} {\bibfnamefont {S.}~\bibnamefont
  {Pirandola}}, \bibinfo {author} {\bibfnamefont {R.}~\bibnamefont {Laurenza}},
  \bibinfo {author} {\bibfnamefont {C.}~\bibnamefont {Ottaviani}},\ and\
  \bibinfo {author} {\bibfnamefont {L.}~\bibnamefont {Banchi}},\ }\bibfield
  {title} {\bibinfo {title} {Fundamental limits of repeaterless quantum
  communications},\ }\href {https://www.nature.com/articles/ncomms15043}
  {\bibfield  {journal} {\bibinfo  {journal} {Nat. Commun.}\ }\textbf {\bibinfo
  {volume} {8}},\ \bibinfo {pages} {15043} (\bibinfo {year}
  {2017})}\BibitemShut {NoStop}%
\bibitem [{\citenamefont {Zukowski}\ \emph {et~al.}(1993)\citenamefont
  {Zukowski}, \citenamefont {Zeilinger}, \citenamefont {Horne},\ and\
  \citenamefont {Ekert}}]{EntSwap1993}%
  \BibitemOpen
  \bibfield  {author} {\bibinfo {author} {\bibfnamefont {M.}~\bibnamefont
  {Zukowski}}, \bibinfo {author} {\bibfnamefont {A.}~\bibnamefont {Zeilinger}},
  \bibinfo {author} {\bibfnamefont {M.~A.}\ \bibnamefont {Horne}},\ and\
  \bibinfo {author} {\bibfnamefont {A.~K.}\ \bibnamefont {Ekert}},\ }\bibfield
  {title} {\bibinfo {title} {``event-ready-detectors'' bell experiment via
  entanglement swapping},\ }\href {https://doi.org/10.1103/PhysRevLett.71.4287}
  {\bibfield  {journal} {\bibinfo  {journal} {Phys. Rev. Lett.}\ }\textbf
  {\bibinfo {volume} {71}},\ \bibinfo {pages} {4287} (\bibinfo {year}
  {1993})}\BibitemShut {NoStop}%
\bibitem [{\citenamefont {Briegel}\ \emph {et~al.}(1998)\citenamefont
  {Briegel}, \citenamefont {D\"ur}, \citenamefont {Cirac},\ and\ \citenamefont
  {Zoller}}]{Briegel1998Repeater}%
  \BibitemOpen
  \bibfield  {author} {\bibinfo {author} {\bibfnamefont {H.-J.}\ \bibnamefont
  {Briegel}}, \bibinfo {author} {\bibfnamefont {W.}~\bibnamefont {D\"ur}},
  \bibinfo {author} {\bibfnamefont {J.~I.}\ \bibnamefont {Cirac}},\ and\
  \bibinfo {author} {\bibfnamefont {P.}~\bibnamefont {Zoller}},\ }\bibfield
  {title} {\bibinfo {title} {Quantum repeaters: The role of imperfect local
  operations in quantum communication},\ }\href
  {https://doi.org/10.1103/PhysRevLett.81.5932} {\bibfield  {journal} {\bibinfo
   {journal} {Phys. Rev. Lett.}\ }\textbf {\bibinfo {volume} {81}},\ \bibinfo
  {pages} {5932} (\bibinfo {year} {1998})}\BibitemShut {NoStop}%
\bibitem [{\citenamefont {Azuma}\ \emph
  {et~al.}(2015{\natexlab{a}})\citenamefont {Azuma}, \citenamefont {Tamaki},\
  and\ \citenamefont {Lo}}]{azuma2015all}%
  \BibitemOpen
  \bibfield  {author} {\bibinfo {author} {\bibfnamefont {K.}~\bibnamefont
  {Azuma}}, \bibinfo {author} {\bibfnamefont {K.}~\bibnamefont {Tamaki}},\ and\
  \bibinfo {author} {\bibfnamefont {H.-K.}\ \bibnamefont {Lo}},\ }\bibfield
  {title} {\bibinfo {title} {All-photonic quantum repeaters},\ }\href
  {https://www.nature.com/articles/ncomms7787} {\bibfield  {journal} {\bibinfo
  {journal} {Nat. Commun.}\ }\textbf {\bibinfo {volume} {6}},\ \bibinfo {pages}
  {6787} (\bibinfo {year} {2015}{\natexlab{a}})}\BibitemShut {NoStop}%
\bibitem [{\citenamefont {Xu}\ \emph {et~al.}(2020)\citenamefont {Xu},
  \citenamefont {Ma}, \citenamefont {Zhang}, \citenamefont {Lo},\ and\
  \citenamefont {Pan}}]{Xu2020Secure}%
  \BibitemOpen
  \bibfield  {author} {\bibinfo {author} {\bibfnamefont {F.}~\bibnamefont
  {Xu}}, \bibinfo {author} {\bibfnamefont {X.}~\bibnamefont {Ma}}, \bibinfo
  {author} {\bibfnamefont {Q.}~\bibnamefont {Zhang}}, \bibinfo {author}
  {\bibfnamefont {H.-K.}\ \bibnamefont {Lo}},\ and\ \bibinfo {author}
  {\bibfnamefont {J.-W.}\ \bibnamefont {Pan}},\ }\bibfield  {title} {\bibinfo
  {title} {Secure quantum key distribution with realistic devices},\ }\href
  {https://doi.org/10.1103/RevModPhys.92.025002} {\bibfield  {journal}
  {\bibinfo  {journal} {Rev. Mod. Phys.}\ }\textbf {\bibinfo {volume} {92}},\
  \bibinfo {pages} {025002} (\bibinfo {year} {2020})}\BibitemShut {NoStop}%
\bibitem [{\citenamefont {Lo}\ and\ \citenamefont
  {Chau}(1999)}]{lo1999Unconditional}%
  \BibitemOpen
  \bibfield  {author} {\bibinfo {author} {\bibfnamefont {H.~K.}\ \bibnamefont
  {Lo}}\ and\ \bibinfo {author} {\bibfnamefont {H.~F.}\ \bibnamefont {Chau}},\
  }\bibfield  {title} {\bibinfo {title} {Unconditional security of quantum key
  distribution over arbitrarily long distances},\ }\href
  {http://science.sciencemag.org/content/283/5410/2050} {\bibfield  {journal}
  {\bibinfo  {journal} {Science}\ }\textbf {\bibinfo {volume} {283}},\ \bibinfo
  {pages} {2050} (\bibinfo {year} {1999})}\BibitemShut {NoStop}%
\bibitem [{\citenamefont {Shor}\ and\ \citenamefont
  {Preskill}(2000)}]{shor2000Simple}%
  \BibitemOpen
  \bibfield  {author} {\bibinfo {author} {\bibfnamefont {P.~W.}\ \bibnamefont
  {Shor}}\ and\ \bibinfo {author} {\bibfnamefont {J.}~\bibnamefont
  {Preskill}},\ }\bibfield  {title} {\bibinfo {title} {Simple proof of security
  of the bb84 quantum key distribution protocol},\ }\href
  {https://journals.aps.org/prl/abstract/10.1103/PhysRevLett.85.441} {\bibfield
   {journal} {\bibinfo  {journal} {Phys. Rev. Lett.}\ }\textbf {\bibinfo
  {volume} {85}},\ \bibinfo {pages} {441} (\bibinfo {year} {2000})}\BibitemShut
  {NoStop}%
\bibitem [{\citenamefont {Koashi}(2009)}]{koashi2009simple}%
  \BibitemOpen
  \bibfield  {author} {\bibinfo {author} {\bibfnamefont {M.}~\bibnamefont
  {Koashi}},\ }\bibfield  {title} {\bibinfo {title} {Simple security proof of
  quantum key distribution based on complementarity},\ }\href@noop {}
  {\bibfield  {journal} {\bibinfo  {journal} {New Journal of Physics}\ }\textbf
  {\bibinfo {volume} {11}},\ \bibinfo {pages} {045018} (\bibinfo {year}
  {2009})}\BibitemShut {NoStop}%
\bibitem [{\citenamefont {Gottesman}\ \emph {et~al.}(2004)\citenamefont
  {Gottesman}, \citenamefont {Lo}, \citenamefont {L\"{u}tkenhaus},\ and\
  \citenamefont {Preskill}}]{gottesman2004security}%
  \BibitemOpen
  \bibfield  {author} {\bibinfo {author} {\bibfnamefont {D.}~\bibnamefont
  {Gottesman}}, \bibinfo {author} {\bibfnamefont {H.-K.}\ \bibnamefont {Lo}},
  \bibinfo {author} {\bibfnamefont {N.}~\bibnamefont {L\"{u}tkenhaus}},\ and\
  \bibinfo {author} {\bibfnamefont {J.}~\bibnamefont {Preskill}},\ }\bibfield
  {title} {\bibinfo {title} {Security of quantum key distribution with
  imperfect devices},\ }\href
  {http://dl.acm.org/citation.cfm?id=2011586.2011587} {\bibfield  {journal}
  {\bibinfo  {journal} {Quantum Info. Comput.}\ }\textbf {\bibinfo {volume}
  {4}},\ \bibinfo {pages} {325} (\bibinfo {year} {2004})}\BibitemShut {NoStop}%
\bibitem [{\citenamefont {Makarov}\ \emph {et~al.}(2006)\citenamefont
  {Makarov}, \citenamefont {Anisimov},\ and\ \citenamefont
  {Skaar}}]{makarov2008effects}%
  \BibitemOpen
  \bibfield  {author} {\bibinfo {author} {\bibfnamefont {V.}~\bibnamefont
  {Makarov}}, \bibinfo {author} {\bibfnamefont {A.}~\bibnamefont {Anisimov}},\
  and\ \bibinfo {author} {\bibfnamefont {J.}~\bibnamefont {Skaar}},\ }\bibfield
   {title} {\bibinfo {title} {Effects of detector efficiency mismatch on
  security of quantum cryptosystems},\ }\href
  {https://doi.org/10.1103/PhysRevA.74.022313} {\bibfield  {journal} {\bibinfo
  {journal} {Phys. Rev. A}\ }\textbf {\bibinfo {volume} {74}},\ \bibinfo
  {pages} {022313} (\bibinfo {year} {2006})}\BibitemShut {NoStop}%
\bibitem [{\citenamefont {Qi}\ \emph {et~al.}(2007)\citenamefont {Qi},
  \citenamefont {Fung}, \citenamefont {Lo},\ and\ \citenamefont
  {Ma}}]{qi2007time}%
  \BibitemOpen
  \bibfield  {author} {\bibinfo {author} {\bibfnamefont {B.}~\bibnamefont
  {Qi}}, \bibinfo {author} {\bibfnamefont {C.-H.~F.}\ \bibnamefont {Fung}},
  \bibinfo {author} {\bibfnamefont {H.-K.}\ \bibnamefont {Lo}},\ and\ \bibinfo
  {author} {\bibfnamefont {X.}~\bibnamefont {Ma}},\ }\bibfield  {title}
  {\bibinfo {title} {Time-shift attack in practical quantum cryptosystems},\
  }\href {https://dl.acm.org/doi/abs/10.5555/2011706.2011709} {\bibfield
  {journal} {\bibinfo  {journal} {Quantum Information \& Computation}\ }\textbf
  {\bibinfo {volume} {7}},\ \bibinfo {pages} {73} (\bibinfo {year}
  {2007})}\BibitemShut {NoStop}%
\bibitem [{\citenamefont {Lo}\ \emph {et~al.}(2012)\citenamefont {Lo},
  \citenamefont {Curty},\ and\ \citenamefont {Qi}}]{lo2012Measurement}%
  \BibitemOpen
  \bibfield  {author} {\bibinfo {author} {\bibfnamefont {H.-K.}\ \bibnamefont
  {Lo}}, \bibinfo {author} {\bibfnamefont {M.}~\bibnamefont {Curty}},\ and\
  \bibinfo {author} {\bibfnamefont {B.}~\bibnamefont {Qi}},\ }\bibfield
  {title} {\bibinfo {title} {Measurement-device-independent quantum key
  distribution},\ }\href {https://doi.org/10.1103/PhysRevLett.108.130503}
  {\bibfield  {journal} {\bibinfo  {journal} {Phys. Rev. Lett.}\ }\textbf
  {\bibinfo {volume} {108}},\ \bibinfo {pages} {130503} (\bibinfo {year}
  {2012})}\BibitemShut {NoStop}%
\bibitem [{\citenamefont {Rubenok}\ \emph {et~al.}(2013)\citenamefont
  {Rubenok}, \citenamefont {Slater}, \citenamefont {Chan}, \citenamefont
  {Lucio-Martinez},\ and\ \citenamefont {Tittel}}]{rubenok2013real}%
  \BibitemOpen
  \bibfield  {author} {\bibinfo {author} {\bibfnamefont {A.}~\bibnamefont
  {Rubenok}}, \bibinfo {author} {\bibfnamefont {J.~A.}\ \bibnamefont {Slater}},
  \bibinfo {author} {\bibfnamefont {P.}~\bibnamefont {Chan}}, \bibinfo {author}
  {\bibfnamefont {I.}~\bibnamefont {Lucio-Martinez}},\ and\ \bibinfo {author}
  {\bibfnamefont {W.}~\bibnamefont {Tittel}},\ }\bibfield  {title} {\bibinfo
  {title} {Real-world two-photon interference and proof-of-principle quantum
  key distribution immune to detector attacks},\ }\href
  {https://doi.org/10.1103/PhysRevLett.111.130501} {\bibfield  {journal}
  {\bibinfo  {journal} {Phys. Rev. Lett.}\ }\textbf {\bibinfo {volume} {111}},\
  \bibinfo {pages} {130501} (\bibinfo {year} {2013})}\BibitemShut {NoStop}%
\bibitem [{\citenamefont {Liu}\ \emph {et~al.}(2013)\citenamefont {Liu},
  \citenamefont {Chen}, \citenamefont {Wang}, \citenamefont {Liang},
  \citenamefont {Shentu}, \citenamefont {Wang}, \citenamefont {Cui},
  \citenamefont {Yin}, \citenamefont {Liu}, \citenamefont {Li}, \citenamefont
  {Ma}, \citenamefont {Pelc}, \citenamefont {Fejer}, \citenamefont {Peng},
  \citenamefont {Zhang},\ and\ \citenamefont {Pan}}]{liu2013experimental}%
  \BibitemOpen
  \bibfield  {author} {\bibinfo {author} {\bibfnamefont {Y.}~\bibnamefont
  {Liu}}, \bibinfo {author} {\bibfnamefont {T.-Y.}\ \bibnamefont {Chen}},
  \bibinfo {author} {\bibfnamefont {L.-J.}\ \bibnamefont {Wang}}, \bibinfo
  {author} {\bibfnamefont {H.}~\bibnamefont {Liang}}, \bibinfo {author}
  {\bibfnamefont {G.-L.}\ \bibnamefont {Shentu}}, \bibinfo {author}
  {\bibfnamefont {J.}~\bibnamefont {Wang}}, \bibinfo {author} {\bibfnamefont
  {K.}~\bibnamefont {Cui}}, \bibinfo {author} {\bibfnamefont {H.-L.}\
  \bibnamefont {Yin}}, \bibinfo {author} {\bibfnamefont {N.-L.}\ \bibnamefont
  {Liu}}, \bibinfo {author} {\bibfnamefont {L.}~\bibnamefont {Li}}, \bibinfo
  {author} {\bibfnamefont {X.}~\bibnamefont {Ma}}, \bibinfo {author}
  {\bibfnamefont {J.~S.}\ \bibnamefont {Pelc}}, \bibinfo {author}
  {\bibfnamefont {M.~M.}\ \bibnamefont {Fejer}}, \bibinfo {author}
  {\bibfnamefont {C.-Z.}\ \bibnamefont {Peng}}, \bibinfo {author}
  {\bibfnamefont {Q.}~\bibnamefont {Zhang}},\ and\ \bibinfo {author}
  {\bibfnamefont {J.-W.}\ \bibnamefont {Pan}},\ }\bibfield  {title} {\bibinfo
  {title} {Experimental measurement-device-independent quantum key
  distribution},\ }\href {https://doi.org/10.1103/PhysRevLett.111.130502}
  {\bibfield  {journal} {\bibinfo  {journal} {Phys. Rev. Lett.}\ }\textbf
  {\bibinfo {volume} {111}},\ \bibinfo {pages} {130502} (\bibinfo {year}
  {2013})}\BibitemShut {NoStop}%
\bibitem [{\citenamefont {Ferreira~da Silva}\ \emph {et~al.}(2013)\citenamefont
  {Ferreira~da Silva}, \citenamefont {Vitoreti}, \citenamefont {Xavier},
  \citenamefont {do~Amaral}, \citenamefont {Tempor\~ao},\ and\ \citenamefont
  {von~der Weid}}]{silva2013proof}%
  \BibitemOpen
  \bibfield  {author} {\bibinfo {author} {\bibfnamefont {T.}~\bibnamefont
  {Ferreira~da Silva}}, \bibinfo {author} {\bibfnamefont {D.}~\bibnamefont
  {Vitoreti}}, \bibinfo {author} {\bibfnamefont {G.~B.}\ \bibnamefont
  {Xavier}}, \bibinfo {author} {\bibfnamefont {G.~C.}\ \bibnamefont
  {do~Amaral}}, \bibinfo {author} {\bibfnamefont {G.~P.}\ \bibnamefont
  {Tempor\~ao}},\ and\ \bibinfo {author} {\bibfnamefont {J.~P.}\ \bibnamefont
  {von~der Weid}},\ }\bibfield  {title} {\bibinfo {title} {Proof-of-principle
  demonstration of measurement-device-independent quantum key distribution
  using polarization qubits},\ }\href
  {https://doi.org/10.1103/PhysRevA.88.052303} {\bibfield  {journal} {\bibinfo
  {journal} {Phys. Rev. A}\ }\textbf {\bibinfo {volume} {88}},\ \bibinfo
  {pages} {052303} (\bibinfo {year} {2013})}\BibitemShut {NoStop}%
\bibitem [{\citenamefont {Woodward}\ \emph {et~al.}(2021)\citenamefont
  {Woodward}, \citenamefont {Lo}, \citenamefont {Pittaluga}, \citenamefont
  {Minder}, \citenamefont {Paraiso}, \citenamefont {Lucamarini}, \citenamefont
  {Yuan},\ and\ \citenamefont {Shields}}]{woodward2021gigahertz}%
  \BibitemOpen
  \bibfield  {author} {\bibinfo {author} {\bibfnamefont {R.~I.}\ \bibnamefont
  {Woodward}}, \bibinfo {author} {\bibfnamefont {Y.}~\bibnamefont {Lo}},
  \bibinfo {author} {\bibfnamefont {M.}~\bibnamefont {Pittaluga}}, \bibinfo
  {author} {\bibfnamefont {M.}~\bibnamefont {Minder}}, \bibinfo {author}
  {\bibfnamefont {T.}~\bibnamefont {Paraiso}}, \bibinfo {author} {\bibfnamefont
  {M.}~\bibnamefont {Lucamarini}}, \bibinfo {author} {\bibfnamefont
  {Z.}~\bibnamefont {Yuan}},\ and\ \bibinfo {author} {\bibfnamefont
  {A.}~\bibnamefont {Shields}},\ }\bibfield  {title} {\bibinfo {title}
  {Gigahertz measurement-device-independent quantum key distribution using
  directly modulated lasers},\ }\href
  {https://www.nature.com/articles/s41534-021-00394-2} {\bibfield  {journal}
  {\bibinfo  {journal} {npj Quantum Information}\ }\textbf {\bibinfo {volume}
  {7}},\ \bibinfo {pages} {1} (\bibinfo {year} {2021})}\BibitemShut {NoStop}%
\bibitem [{\citenamefont {Tang}\ \emph {et~al.}(2016)\citenamefont {Tang},
  \citenamefont {Yin}, \citenamefont {Zhao}, \citenamefont {Liu}, \citenamefont
  {Sun}, \citenamefont {Huang}, \citenamefont {Zhang}, \citenamefont {Chen},
  \citenamefont {Zhang}, \citenamefont {You}, \citenamefont {Wang},
  \citenamefont {Liu}, \citenamefont {Lu}, \citenamefont {Jiang}, \citenamefont
  {Ma}, \citenamefont {Zhang}, \citenamefont {Chen},\ and\ \citenamefont
  {Pan}}]{Tang2016MDInet}%
  \BibitemOpen
  \bibfield  {author} {\bibinfo {author} {\bibfnamefont {Y.-L.}\ \bibnamefont
  {Tang}}, \bibinfo {author} {\bibfnamefont {H.-L.}\ \bibnamefont {Yin}},
  \bibinfo {author} {\bibfnamefont {Q.}~\bibnamefont {Zhao}}, \bibinfo {author}
  {\bibfnamefont {H.}~\bibnamefont {Liu}}, \bibinfo {author} {\bibfnamefont
  {X.-X.}\ \bibnamefont {Sun}}, \bibinfo {author} {\bibfnamefont {M.-Q.}\
  \bibnamefont {Huang}}, \bibinfo {author} {\bibfnamefont {W.-J.}\ \bibnamefont
  {Zhang}}, \bibinfo {author} {\bibfnamefont {S.-J.}\ \bibnamefont {Chen}},
  \bibinfo {author} {\bibfnamefont {L.}~\bibnamefont {Zhang}}, \bibinfo
  {author} {\bibfnamefont {L.-X.}\ \bibnamefont {You}}, \bibinfo {author}
  {\bibfnamefont {Z.}~\bibnamefont {Wang}}, \bibinfo {author} {\bibfnamefont
  {Y.}~\bibnamefont {Liu}}, \bibinfo {author} {\bibfnamefont {C.-Y.}\
  \bibnamefont {Lu}}, \bibinfo {author} {\bibfnamefont {X.}~\bibnamefont
  {Jiang}}, \bibinfo {author} {\bibfnamefont {X.}~\bibnamefont {Ma}}, \bibinfo
  {author} {\bibfnamefont {Q.}~\bibnamefont {Zhang}}, \bibinfo {author}
  {\bibfnamefont {T.-Y.}\ \bibnamefont {Chen}},\ and\ \bibinfo {author}
  {\bibfnamefont {J.-W.}\ \bibnamefont {Pan}},\ }\bibfield  {title} {\bibinfo
  {title} {Measurement-device-independent quantum key distribution over
  untrustful metropolitan network},\ }\href
  {https://doi.org/10.1103/PhysRevX.6.011024} {\bibfield  {journal} {\bibinfo
  {journal} {Phys. Rev. X}\ }\textbf {\bibinfo {volume} {6}},\ \bibinfo {pages}
  {011024} (\bibinfo {year} {2016})}\BibitemShut {NoStop}%
\bibitem [{\citenamefont {Tamaki}\ \emph {et~al.}(2012)\citenamefont {Tamaki},
  \citenamefont {Lo}, \citenamefont {Fung},\ and\ \citenamefont
  {Qi}}]{tamaki2012phase}%
  \BibitemOpen
  \bibfield  {author} {\bibinfo {author} {\bibfnamefont {K.}~\bibnamefont
  {Tamaki}}, \bibinfo {author} {\bibfnamefont {H.-K.}\ \bibnamefont {Lo}},
  \bibinfo {author} {\bibfnamefont {C.-H.~F.}\ \bibnamefont {Fung}},\ and\
  \bibinfo {author} {\bibfnamefont {B.}~\bibnamefont {Qi}},\ }\bibfield
  {title} {\bibinfo {title} {Phase encoding schemes for
  measurement-device-independent quantum key distribution with basis-dependent
  flaw},\ }\href {https://doi.org/10.1103/PhysRevA.85.042307} {\bibfield
  {journal} {\bibinfo  {journal} {Phys. Rev. A}\ }\textbf {\bibinfo {volume}
  {85}},\ \bibinfo {pages} {042307} (\bibinfo {year} {2012})}\BibitemShut
  {NoStop}%
\bibitem [{\citenamefont {Ma}\ and\ \citenamefont
  {Razavi}(2012)}]{Ma2012alternative}%
  \BibitemOpen
  \bibfield  {author} {\bibinfo {author} {\bibfnamefont {X.}~\bibnamefont
  {Ma}}\ and\ \bibinfo {author} {\bibfnamefont {M.}~\bibnamefont {Razavi}},\
  }\bibfield  {title} {\bibinfo {title} {Alternative schemes for
  measurement-device-independent quantum key distribution},\ }\href
  {https://doi.org/10.1103/PhysRevA.86.062319} {\bibfield  {journal} {\bibinfo
  {journal} {Phys. Rev. A}\ }\textbf {\bibinfo {volume} {86}},\ \bibinfo
  {pages} {062319} (\bibinfo {year} {2012})}\BibitemShut {NoStop}%
\bibitem [{\citenamefont {Lucamarini}\ \emph {et~al.}(2018)\citenamefont
  {Lucamarini}, \citenamefont {Yuan}, \citenamefont {Dynes},\ and\
  \citenamefont {Shields}}]{lucamarini2018overcoming}%
  \BibitemOpen
  \bibfield  {author} {\bibinfo {author} {\bibfnamefont {M.}~\bibnamefont
  {Lucamarini}}, \bibinfo {author} {\bibfnamefont {Z.}~\bibnamefont {Yuan}},
  \bibinfo {author} {\bibfnamefont {J.}~\bibnamefont {Dynes}},\ and\ \bibinfo
  {author} {\bibfnamefont {A.}~\bibnamefont {Shields}},\ }\bibfield  {title}
  {\bibinfo {title} {Overcoming the rate--distance limit of quantum key
  distribution without quantum repeaters},\ }\href
  {https://www.nature.com/articles/s41586-018-0066-6} {\bibfield  {journal}
  {\bibinfo  {journal} {Nature}\ }\textbf {\bibinfo {volume} {557}},\ \bibinfo
  {pages} {400} (\bibinfo {year} {2018})}\BibitemShut {NoStop}%
\bibitem [{\citenamefont {Ma}\ \emph {et~al.}(2018)\citenamefont {Ma},
  \citenamefont {Zeng},\ and\ \citenamefont {Zhou}}]{Ma2018phase}%
  \BibitemOpen
  \bibfield  {author} {\bibinfo {author} {\bibfnamefont {X.}~\bibnamefont
  {Ma}}, \bibinfo {author} {\bibfnamefont {P.}~\bibnamefont {Zeng}},\ and\
  \bibinfo {author} {\bibfnamefont {H.}~\bibnamefont {Zhou}},\ }\bibfield
  {title} {\bibinfo {title} {Phase-matching quantum key distribution},\ }\href
  {https://doi.org/10.1103/PhysRevX.8.031043} {\bibfield  {journal} {\bibinfo
  {journal} {Phys. Rev. X}\ }\textbf {\bibinfo {volume} {8}},\ \bibinfo {pages}
  {031043} (\bibinfo {year} {2018})}\BibitemShut {NoStop}%
\bibitem [{\citenamefont {Lin}\ and\ \citenamefont
  {L\"utkenhaus}(2018)}]{lin2018simple}%
  \BibitemOpen
  \bibfield  {author} {\bibinfo {author} {\bibfnamefont {J.}~\bibnamefont
  {Lin}}\ and\ \bibinfo {author} {\bibfnamefont {N.}~\bibnamefont
  {L\"utkenhaus}},\ }\bibfield  {title} {\bibinfo {title} {Simple security
  analysis of phase-matching measurement-device-independent quantum key
  distribution},\ }\href {https://doi.org/10.1103/PhysRevA.98.042332}
  {\bibfield  {journal} {\bibinfo  {journal} {Phys. Rev. A}\ }\textbf {\bibinfo
  {volume} {98}},\ \bibinfo {pages} {042332} (\bibinfo {year}
  {2018})}\BibitemShut {NoStop}%
\bibitem [{\citenamefont {Wang}\ \emph {et~al.}(2018)\citenamefont {Wang},
  \citenamefont {Yu},\ and\ \citenamefont {Hu}}]{wang2018twin}%
  \BibitemOpen
  \bibfield  {author} {\bibinfo {author} {\bibfnamefont {X.-B.}\ \bibnamefont
  {Wang}}, \bibinfo {author} {\bibfnamefont {Z.-W.}\ \bibnamefont {Yu}},\ and\
  \bibinfo {author} {\bibfnamefont {X.-L.}\ \bibnamefont {Hu}},\ }\bibfield
  {title} {\bibinfo {title} {Twin-field quantum key distribution with large
  misalignment error},\ }\href {https://doi.org/10.1103/PhysRevA.98.062323}
  {\bibfield  {journal} {\bibinfo  {journal} {Phys. Rev. A}\ }\textbf {\bibinfo
  {volume} {98}},\ \bibinfo {pages} {062323} (\bibinfo {year}
  {2018})}\BibitemShut {NoStop}%
\bibitem [{\citenamefont {Duan}\ \emph {et~al.}(2001)\citenamefont {Duan},
  \citenamefont {Lukin}, \citenamefont {Cirac},\ and\ \citenamefont
  {Zoller}}]{duan2001long}%
  \BibitemOpen
  \bibfield  {author} {\bibinfo {author} {\bibfnamefont {L.-M.}\ \bibnamefont
  {Duan}}, \bibinfo {author} {\bibfnamefont {M.}~\bibnamefont {Lukin}},
  \bibinfo {author} {\bibfnamefont {J.~I.}\ \bibnamefont {Cirac}},\ and\
  \bibinfo {author} {\bibfnamefont {P.}~\bibnamefont {Zoller}},\ }\bibfield
  {title} {\bibinfo {title} {Long-distance quantum communication with atomic
  ensembles and linear optics},\ }\href@noop {} {\bibfield  {journal} {\bibinfo
   {journal} {Nature}\ }\textbf {\bibinfo {volume} {414}},\ \bibinfo {pages}
  {413} (\bibinfo {year} {2001})}\BibitemShut {NoStop}%
\bibitem [{\citenamefont {Minder}\ \emph {et~al.}(2019)\citenamefont {Minder},
  \citenamefont {Pittaluga}, \citenamefont {Roberts}, \citenamefont
  {Lucamarini}, \citenamefont {Dynes}, \citenamefont {Yuan},\ and\
  \citenamefont {Shields}}]{minder2019experimental}%
  \BibitemOpen
  \bibfield  {author} {\bibinfo {author} {\bibfnamefont {M.}~\bibnamefont
  {Minder}}, \bibinfo {author} {\bibfnamefont {M.}~\bibnamefont {Pittaluga}},
  \bibinfo {author} {\bibfnamefont {G.}~\bibnamefont {Roberts}}, \bibinfo
  {author} {\bibfnamefont {M.}~\bibnamefont {Lucamarini}}, \bibinfo {author}
  {\bibfnamefont {J.}~\bibnamefont {Dynes}}, \bibinfo {author} {\bibfnamefont
  {Z.}~\bibnamefont {Yuan}},\ and\ \bibinfo {author} {\bibfnamefont
  {A.}~\bibnamefont {Shields}},\ }\bibfield  {title} {\bibinfo {title}
  {Experimental quantum key distribution beyond the repeaterless secret key
  capacity},\ }\href {https://www.nature.com/articles/s41566-019-0377-7}
  {\bibfield  {journal} {\bibinfo  {journal} {Nature Photonics}\ }\textbf
  {\bibinfo {volume} {13}},\ \bibinfo {pages} {334} (\bibinfo {year}
  {2019})}\BibitemShut {NoStop}%
\bibitem [{\citenamefont {Wang}\ \emph {et~al.}(2019)\citenamefont {Wang},
  \citenamefont {He}, \citenamefont {Yin}, \citenamefont {Lu}, \citenamefont
  {Cui}, \citenamefont {Chen}, \citenamefont {Zhou}, \citenamefont {Guo},\ and\
  \citenamefont {Han}}]{wang2019beating}%
  \BibitemOpen
  \bibfield  {author} {\bibinfo {author} {\bibfnamefont {S.}~\bibnamefont
  {Wang}}, \bibinfo {author} {\bibfnamefont {D.-Y.}\ \bibnamefont {He}},
  \bibinfo {author} {\bibfnamefont {Z.-Q.}\ \bibnamefont {Yin}}, \bibinfo
  {author} {\bibfnamefont {F.-Y.}\ \bibnamefont {Lu}}, \bibinfo {author}
  {\bibfnamefont {C.-H.}\ \bibnamefont {Cui}}, \bibinfo {author} {\bibfnamefont
  {W.}~\bibnamefont {Chen}}, \bibinfo {author} {\bibfnamefont {Z.}~\bibnamefont
  {Zhou}}, \bibinfo {author} {\bibfnamefont {G.-C.}\ \bibnamefont {Guo}},\ and\
  \bibinfo {author} {\bibfnamefont {Z.-F.}\ \bibnamefont {Han}},\ }\bibfield
  {title} {\bibinfo {title} {Beating the fundamental rate-distance limit in a
  proof-of-principle quantum key distribution system},\ }\href
  {https://doi.org/10.1103/PhysRevX.9.021046} {\bibfield  {journal} {\bibinfo
  {journal} {Phys. Rev. X}\ }\textbf {\bibinfo {volume} {9}},\ \bibinfo {pages}
  {021046} (\bibinfo {year} {2019})}\BibitemShut {NoStop}%
\bibitem [{\citenamefont {Fang}\ \emph {et~al.}(2020)\citenamefont {Fang},
  \citenamefont {Zeng}, \citenamefont {Liu}, \citenamefont {Zou}, \citenamefont
  {Wu}, \citenamefont {Tang}, \citenamefont {Sheng}, \citenamefont {Xiang},
  \citenamefont {Zhang}, \citenamefont {Li} \emph
  {et~al.}}]{fang2020implementation}%
  \BibitemOpen
  \bibfield  {author} {\bibinfo {author} {\bibfnamefont {X.-T.}\ \bibnamefont
  {Fang}}, \bibinfo {author} {\bibfnamefont {P.}~\bibnamefont {Zeng}}, \bibinfo
  {author} {\bibfnamefont {H.}~\bibnamefont {Liu}}, \bibinfo {author}
  {\bibfnamefont {M.}~\bibnamefont {Zou}}, \bibinfo {author} {\bibfnamefont
  {W.}~\bibnamefont {Wu}}, \bibinfo {author} {\bibfnamefont {Y.-L.}\
  \bibnamefont {Tang}}, \bibinfo {author} {\bibfnamefont {Y.-J.}\ \bibnamefont
  {Sheng}}, \bibinfo {author} {\bibfnamefont {Y.}~\bibnamefont {Xiang}},
  \bibinfo {author} {\bibfnamefont {W.}~\bibnamefont {Zhang}}, \bibinfo
  {author} {\bibfnamefont {H.}~\bibnamefont {Li}}, \emph {et~al.},\ }\bibfield
  {title} {\bibinfo {title} {Implementation of quantum key distribution
  surpassing the linear rate-transmittance bound},\ }\href
  {https://www.nature.com/articles/s41566-020-0599-8} {\bibfield  {journal}
  {\bibinfo  {journal} {Nature Photonics}\ ,\ \bibinfo {pages} {1}} (\bibinfo
  {year} {2020})}\BibitemShut {NoStop}%
\bibitem [{\citenamefont {Zhong}\ \emph {et~al.}(2019)\citenamefont {Zhong},
  \citenamefont {Hu}, \citenamefont {Curty}, \citenamefont {Qian},\ and\
  \citenamefont {Lo}}]{zhong2020proofofprinciple}%
  \BibitemOpen
  \bibfield  {author} {\bibinfo {author} {\bibfnamefont {X.}~\bibnamefont
  {Zhong}}, \bibinfo {author} {\bibfnamefont {J.}~\bibnamefont {Hu}}, \bibinfo
  {author} {\bibfnamefont {M.}~\bibnamefont {Curty}}, \bibinfo {author}
  {\bibfnamefont {L.}~\bibnamefont {Qian}},\ and\ \bibinfo {author}
  {\bibfnamefont {H.-K.}\ \bibnamefont {Lo}},\ }\bibfield  {title} {\bibinfo
  {title} {Proof-of-principle experimental demonstration of twin-field type
  quantum key distribution},\ }\href
  {https://doi.org/10.1103/PhysRevLett.123.100506} {\bibfield  {journal}
  {\bibinfo  {journal} {Phys. Rev. Lett.}\ }\textbf {\bibinfo {volume} {123}},\
  \bibinfo {pages} {100506} (\bibinfo {year} {2019})}\BibitemShut {NoStop}%
\bibitem [{\citenamefont {Chen}\ \emph {et~al.}(2020)\citenamefont {Chen},
  \citenamefont {Zhang}, \citenamefont {Liu}, \citenamefont {Jiang},
  \citenamefont {Zhang}, \citenamefont {Hu}, \citenamefont {Guan},
  \citenamefont {Yu}, \citenamefont {Xu}, \citenamefont {Lin}, \citenamefont
  {Li}, \citenamefont {Chen}, \citenamefont {Li}, \citenamefont {You},
  \citenamefont {Wang}, \citenamefont {Wang}, \citenamefont {Zhang},\ and\
  \citenamefont {Pan}}]{chen2020sending}%
  \BibitemOpen
  \bibfield  {author} {\bibinfo {author} {\bibfnamefont {J.-P.}\ \bibnamefont
  {Chen}}, \bibinfo {author} {\bibfnamefont {C.}~\bibnamefont {Zhang}},
  \bibinfo {author} {\bibfnamefont {Y.}~\bibnamefont {Liu}}, \bibinfo {author}
  {\bibfnamefont {C.}~\bibnamefont {Jiang}}, \bibinfo {author} {\bibfnamefont
  {W.}~\bibnamefont {Zhang}}, \bibinfo {author} {\bibfnamefont {X.-L.}\
  \bibnamefont {Hu}}, \bibinfo {author} {\bibfnamefont {J.-Y.}\ \bibnamefont
  {Guan}}, \bibinfo {author} {\bibfnamefont {Z.-W.}\ \bibnamefont {Yu}},
  \bibinfo {author} {\bibfnamefont {H.}~\bibnamefont {Xu}}, \bibinfo {author}
  {\bibfnamefont {J.}~\bibnamefont {Lin}}, \bibinfo {author} {\bibfnamefont
  {M.-J.}\ \bibnamefont {Li}}, \bibinfo {author} {\bibfnamefont
  {H.}~\bibnamefont {Chen}}, \bibinfo {author} {\bibfnamefont {H.}~\bibnamefont
  {Li}}, \bibinfo {author} {\bibfnamefont {L.}~\bibnamefont {You}}, \bibinfo
  {author} {\bibfnamefont {Z.}~\bibnamefont {Wang}}, \bibinfo {author}
  {\bibfnamefont {X.-B.}\ \bibnamefont {Wang}}, \bibinfo {author}
  {\bibfnamefont {Q.}~\bibnamefont {Zhang}},\ and\ \bibinfo {author}
  {\bibfnamefont {J.-W.}\ \bibnamefont {Pan}},\ }\bibfield  {title} {\bibinfo
  {title} {Sending-or-not-sending with independent lasers: Secure twin-field
  quantum key distribution over 509 km},\ }\href
  {https://doi.org/10.1103/PhysRevLett.124.070501} {\bibfield  {journal}
  {\bibinfo  {journal} {Phys. Rev. Lett.}\ }\textbf {\bibinfo {volume} {124}},\
  \bibinfo {pages} {070501} (\bibinfo {year} {2020})}\BibitemShut {NoStop}%
\bibitem [{\citenamefont {Pittaluga}\ \emph {et~al.}(2021)\citenamefont
  {Pittaluga}, \citenamefont {Minder}, \citenamefont {Lucamarini},
  \citenamefont {Sanzaro}, \citenamefont {Woodward}, \citenamefont {Li},
  \citenamefont {Yuan},\ and\ \citenamefont {Shields}}]{pittaluga2021600}%
  \BibitemOpen
  \bibfield  {author} {\bibinfo {author} {\bibfnamefont {M.}~\bibnamefont
  {Pittaluga}}, \bibinfo {author} {\bibfnamefont {M.}~\bibnamefont {Minder}},
  \bibinfo {author} {\bibfnamefont {M.}~\bibnamefont {Lucamarini}}, \bibinfo
  {author} {\bibfnamefont {M.}~\bibnamefont {Sanzaro}}, \bibinfo {author}
  {\bibfnamefont {R.~I.}\ \bibnamefont {Woodward}}, \bibinfo {author}
  {\bibfnamefont {M.-J.}\ \bibnamefont {Li}}, \bibinfo {author} {\bibfnamefont
  {Z.}~\bibnamefont {Yuan}},\ and\ \bibinfo {author} {\bibfnamefont {A.~J.}\
  \bibnamefont {Shields}},\ }\bibfield  {title} {\bibinfo {title} {600-km
  repeater-like quantum communications with dual-band stabilization},\ }\href
  {https://www.nature.com/articles/s41566-021-00811-0} {\bibfield  {journal}
  {\bibinfo  {journal} {Nature Photonics}\ ,\ \bibinfo {pages} {1}} (\bibinfo
  {year} {2021})}\BibitemShut {NoStop}%
\bibitem [{\citenamefont {Clivati}\ \emph {et~al.}(2020)\citenamefont
  {Clivati}, \citenamefont {Meda}, \citenamefont {Donadello}, \citenamefont
  {Virzi}, \citenamefont {Genovese}, \citenamefont {Levi}, \citenamefont
  {Mura}, \citenamefont {Pittaluga}, \citenamefont {Yuan}, \citenamefont
  {Shields}, \citenamefont {Lucamarini}, \citenamefont {Degiovanni},\ and\
  \citenamefont {Calonico}}]{clivati2020coherent}%
  \BibitemOpen
  \bibfield  {author} {\bibinfo {author} {\bibfnamefont {C.}~\bibnamefont
  {Clivati}}, \bibinfo {author} {\bibfnamefont {A.}~\bibnamefont {Meda}},
  \bibinfo {author} {\bibfnamefont {S.}~\bibnamefont {Donadello}}, \bibinfo
  {author} {\bibfnamefont {S.}~\bibnamefont {Virzi}}, \bibinfo {author}
  {\bibfnamefont {M.}~\bibnamefont {Genovese}}, \bibinfo {author}
  {\bibfnamefont {F.}~\bibnamefont {Levi}}, \bibinfo {author} {\bibfnamefont
  {A.}~\bibnamefont {Mura}}, \bibinfo {author} {\bibfnamefont {M.}~\bibnamefont
  {Pittaluga}}, \bibinfo {author} {\bibfnamefont {Z.~L.}\ \bibnamefont {Yuan}},
  \bibinfo {author} {\bibfnamefont {A.~J.}\ \bibnamefont {Shields}}, \bibinfo
  {author} {\bibfnamefont {M.}~\bibnamefont {Lucamarini}}, \bibinfo {author}
  {\bibfnamefont {I.~P.}\ \bibnamefont {Degiovanni}},\ and\ \bibinfo {author}
  {\bibfnamefont {D.}~\bibnamefont {Calonico}},\ }\href
  {https://arxiv.org/abs/2012.15199} {\bibinfo {title} {Coherent phase transfer
  for real-world twin-field quantum key distribution}} (\bibinfo {year}
  {2020}),\ \Eprint {https://arxiv.org/abs/2012.15199} {arXiv:2012.15199
  [quant-ph]} \BibitemShut {NoStop}%
\bibitem [{\citenamefont {Wang}\ \emph {et~al.}(2022)\citenamefont {Wang},
  \citenamefont {Yin}, \citenamefont {He}, \citenamefont {Chen}, \citenamefont
  {Wang}, \citenamefont {Ye}, \citenamefont {Zhou}, \citenamefont {Fan-Yuan},
  \citenamefont {Wang}, \citenamefont {Zhu}, \citenamefont {Morozov},
  \citenamefont {Divochiy}, \citenamefont {Zhou}, \citenamefont {Guo},\ and\
  \citenamefont {Han}}]{wang2022twin}%
  \BibitemOpen
  \bibfield  {author} {\bibinfo {author} {\bibfnamefont {S.}~\bibnamefont
  {Wang}}, \bibinfo {author} {\bibfnamefont {Z.-Q.}\ \bibnamefont {Yin}},
  \bibinfo {author} {\bibfnamefont {D.-Y.}\ \bibnamefont {He}}, \bibinfo
  {author} {\bibfnamefont {W.}~\bibnamefont {Chen}}, \bibinfo {author}
  {\bibfnamefont {R.-Q.}\ \bibnamefont {Wang}}, \bibinfo {author}
  {\bibfnamefont {P.}~\bibnamefont {Ye}}, \bibinfo {author} {\bibfnamefont
  {Y.}~\bibnamefont {Zhou}}, \bibinfo {author} {\bibfnamefont {G.-J.}\
  \bibnamefont {Fan-Yuan}}, \bibinfo {author} {\bibfnamefont {F.-X.}\
  \bibnamefont {Wang}}, \bibinfo {author} {\bibfnamefont {Y.-G.}\ \bibnamefont
  {Zhu}}, \bibinfo {author} {\bibfnamefont {P.~V.}\ \bibnamefont {Morozov}},
  \bibinfo {author} {\bibfnamefont {A.~V.}\ \bibnamefont {Divochiy}}, \bibinfo
  {author} {\bibfnamefont {Z.}~\bibnamefont {Zhou}}, \bibinfo {author}
  {\bibfnamefont {G.-C.}\ \bibnamefont {Guo}},\ and\ \bibinfo {author}
  {\bibfnamefont {Z.-F.}\ \bibnamefont {Han}},\ }\bibfield  {title} {\bibinfo
  {title} {Twin-field quantum key distribution over 830-km fibre},\ }\bibfield
  {journal} {\bibinfo  {journal} {Nature Photonics}\ }\href
  {https://doi.org/10.1038/s41566-021-00928-2} {10.1038/s41566-021-00928-2}
  (\bibinfo {year} {2022})\BibitemShut {NoStop}%
\bibitem [{\citenamefont {Tang}\ \emph
  {et~al.}(2014{\natexlab{a}})\citenamefont {Tang}, \citenamefont {Liao},
  \citenamefont {Xu}, \citenamefont {Qi}, \citenamefont {Qian},\ and\
  \citenamefont {Lo}}]{Tang2014Experimental}%
  \BibitemOpen
  \bibfield  {author} {\bibinfo {author} {\bibfnamefont {Z.}~\bibnamefont
  {Tang}}, \bibinfo {author} {\bibfnamefont {Z.}~\bibnamefont {Liao}}, \bibinfo
  {author} {\bibfnamefont {F.}~\bibnamefont {Xu}}, \bibinfo {author}
  {\bibfnamefont {B.}~\bibnamefont {Qi}}, \bibinfo {author} {\bibfnamefont
  {L.}~\bibnamefont {Qian}},\ and\ \bibinfo {author} {\bibfnamefont {H.-K.}\
  \bibnamefont {Lo}},\ }\bibfield  {title} {\bibinfo {title} {Experimental
  demonstration of polarization encoding measurement-device-independent quantum
  key distribution},\ }\href {https://doi.org/10.1103/PhysRevLett.112.190503}
  {\bibfield  {journal} {\bibinfo  {journal} {Phys. Rev. Lett.}\ }\textbf
  {\bibinfo {volume} {112}},\ \bibinfo {pages} {190503} (\bibinfo {year}
  {2014}{\natexlab{a}})}\BibitemShut {NoStop}%
\bibitem [{\citenamefont {Tang}\ \emph
  {et~al.}(2014{\natexlab{b}})\citenamefont {Tang}, \citenamefont {Yin},
  \citenamefont {Chen}, \citenamefont {Liu}, \citenamefont {Zhang},
  \citenamefont {Jiang}, \citenamefont {Zhang}, \citenamefont {Wang},
  \citenamefont {You}, \citenamefont {Guan} \emph {et~al.}}]{Tang2014field}%
  \BibitemOpen
  \bibfield  {author} {\bibinfo {author} {\bibfnamefont {Y.-L.}\ \bibnamefont
  {Tang}}, \bibinfo {author} {\bibfnamefont {H.-L.}\ \bibnamefont {Yin}},
  \bibinfo {author} {\bibfnamefont {S.-J.}\ \bibnamefont {Chen}}, \bibinfo
  {author} {\bibfnamefont {Y.}~\bibnamefont {Liu}}, \bibinfo {author}
  {\bibfnamefont {W.-J.}\ \bibnamefont {Zhang}}, \bibinfo {author}
  {\bibfnamefont {X.}~\bibnamefont {Jiang}}, \bibinfo {author} {\bibfnamefont
  {L.}~\bibnamefont {Zhang}}, \bibinfo {author} {\bibfnamefont
  {J.}~\bibnamefont {Wang}}, \bibinfo {author} {\bibfnamefont {L.-X.}\
  \bibnamefont {You}}, \bibinfo {author} {\bibfnamefont {J.-Y.}\ \bibnamefont
  {Guan}}, \emph {et~al.},\ }\bibfield  {title} {\bibinfo {title} {Field test
  of measurement-device-independent quantum key distribution},\ }\href
  {https://ieeexplore.ieee.org/abstract/document/6920009} {\bibfield  {journal}
  {\bibinfo  {journal} {IEEE Journal of Selected Topics in Quantum
  Electronics}\ }\textbf {\bibinfo {volume} {21}},\ \bibinfo {pages} {116}
  (\bibinfo {year} {2014}{\natexlab{b}})}\BibitemShut {NoStop}%
\bibitem [{\citenamefont {Lo}\ \emph {et~al.}(2005)\citenamefont {Lo},
  \citenamefont {Ma},\ and\ \citenamefont {Chen}}]{Lo2005Decoy}%
  \BibitemOpen
  \bibfield  {author} {\bibinfo {author} {\bibfnamefont {H.-K.}\ \bibnamefont
  {Lo}}, \bibinfo {author} {\bibfnamefont {X.}~\bibnamefont {Ma}},\ and\
  \bibinfo {author} {\bibfnamefont {K.}~\bibnamefont {Chen}},\ }\bibfield
  {title} {\bibinfo {title} {Decoy state quantum key distribution},\ }\href
  {https://doi.org/10.1103/PhysRevLett.94.230504} {\bibfield  {journal}
  {\bibinfo  {journal} {Phys. Rev. Lett.}\ }\textbf {\bibinfo {volume} {94}},\
  \bibinfo {pages} {230504} (\bibinfo {year} {2005})}\BibitemShut {NoStop}%
\bibitem [{\citenamefont {Inoue}\ \emph {et~al.}(2003)\citenamefont {Inoue},
  \citenamefont {Waks},\ and\ \citenamefont
  {Yamamoto}}]{inoue2003differential}%
  \BibitemOpen
  \bibfield  {author} {\bibinfo {author} {\bibfnamefont {K.}~\bibnamefont
  {Inoue}}, \bibinfo {author} {\bibfnamefont {E.}~\bibnamefont {Waks}},\ and\
  \bibinfo {author} {\bibfnamefont {Y.}~\bibnamefont {Yamamoto}},\ }\bibfield
  {title} {\bibinfo {title} {Differential-phase-shift quantum key distribution
  using coherent light},\ }\href {https://doi.org/10.1103/PhysRevA.68.022317}
  {\bibfield  {journal} {\bibinfo  {journal} {Phys. Rev. A}\ }\textbf {\bibinfo
  {volume} {68}},\ \bibinfo {pages} {022317} (\bibinfo {year}
  {2003})}\BibitemShut {NoStop}%
\bibitem [{\citenamefont {Sasaki}\ \emph {et~al.}(2014)\citenamefont {Sasaki},
  \citenamefont {Yamamoto},\ and\ \citenamefont
  {Koashi}}]{sasaki2014practical}%
  \BibitemOpen
  \bibfield  {author} {\bibinfo {author} {\bibfnamefont {T.}~\bibnamefont
  {Sasaki}}, \bibinfo {author} {\bibfnamefont {Y.}~\bibnamefont {Yamamoto}},\
  and\ \bibinfo {author} {\bibfnamefont {M.}~\bibnamefont {Koashi}},\
  }\bibfield  {title} {\bibinfo {title} {Practical quantum key distribution
  protocol without monitoring signal disturbance},\ }\href
  {https://www.nature.com/articles/nature13303} {\bibfield  {journal} {\bibinfo
   {journal} {Nature}\ }\textbf {\bibinfo {volume} {509}},\ \bibinfo {pages}
  {475} (\bibinfo {year} {2014})}\BibitemShut {NoStop}%
\bibitem [{\citenamefont {Hwang}(2003)}]{hwang2003decoy}%
  \BibitemOpen
  \bibfield  {author} {\bibinfo {author} {\bibfnamefont {W.-Y.}\ \bibnamefont
  {Hwang}},\ }\bibfield  {title} {\bibinfo {title} {Quantum key distribution
  with high loss: Toward global secure communication},\ }\href
  {https://doi.org/10.1103/PhysRevLett.91.057901} {\bibfield  {journal}
  {\bibinfo  {journal} {Phys. Rev. Lett.}\ }\textbf {\bibinfo {volume} {91}},\
  \bibinfo {pages} {057901} (\bibinfo {year} {2003})}\BibitemShut {NoStop}%
\bibitem [{\citenamefont {Wang}(2005)}]{wang2005decoy}%
  \BibitemOpen
  \bibfield  {author} {\bibinfo {author} {\bibfnamefont {X.-B.}\ \bibnamefont
  {Wang}},\ }\bibfield  {title} {\bibinfo {title} {Beating the
  photon-number-splitting attack in practical quantum cryptography},\ }\href
  {https://doi.org/10.1103/PhysRevLett.94.230503} {\bibfield  {journal}
  {\bibinfo  {journal} {Phys. Rev. Lett.}\ }\textbf {\bibinfo {volume} {94}},\
  \bibinfo {pages} {230503} (\bibinfo {year} {2005})}\BibitemShut {NoStop}%
\bibitem [{\citenamefont {Cao}\ \emph {et~al.}(2015)\citenamefont {Cao},
  \citenamefont {Zhang}, \citenamefont {Lo},\ and\ \citenamefont
  {Ma}}]{Cao2015discrete}%
  \BibitemOpen
  \bibfield  {author} {\bibinfo {author} {\bibfnamefont {Z.}~\bibnamefont
  {Cao}}, \bibinfo {author} {\bibfnamefont {Z.}~\bibnamefont {Zhang}}, \bibinfo
  {author} {\bibfnamefont {H.-K.}\ \bibnamefont {Lo}},\ and\ \bibinfo {author}
  {\bibfnamefont {X.}~\bibnamefont {Ma}},\ }\bibfield  {title} {\bibinfo
  {title} {Discrete-phase-randomized coherent state source and its application
  in quantum key distribution},\ }\href
  {http://stacks.iop.org/1367-2630/17/i=5/a=053014} {\bibfield  {journal}
  {\bibinfo  {journal} {New J. Phys.}\ }\textbf {\bibinfo {volume} {17}},\
  \bibinfo {pages} {053014} (\bibinfo {year} {2015})}\BibitemShut {NoStop}%
\bibitem [{\citenamefont {Zeng}\ \emph {et~al.}(2020)\citenamefont {Zeng},
  \citenamefont {Wu},\ and\ \citenamefont {Ma}}]{Zeng2019Symmetryprotected}%
  \BibitemOpen
  \bibfield  {author} {\bibinfo {author} {\bibfnamefont {P.}~\bibnamefont
  {Zeng}}, \bibinfo {author} {\bibfnamefont {W.}~\bibnamefont {Wu}},\ and\
  \bibinfo {author} {\bibfnamefont {X.}~\bibnamefont {Ma}},\ }\bibfield
  {title} {\bibinfo {title} {Symmetry-protected privacy: Beating the
  rate-distance linear bound over a noisy channel},\ }\href
  {https://doi.org/10.1103/PhysRevApplied.13.064013} {\bibfield  {journal}
  {\bibinfo  {journal} {Phys. Rev. Applied}\ }\textbf {\bibinfo {volume}
  {13}},\ \bibinfo {pages} {064013} (\bibinfo {year} {2020})}\BibitemShut
  {NoStop}%
\bibitem [{\citenamefont {Jiang}\ \emph {et~al.}(2019)\citenamefont {Jiang},
  \citenamefont {Yu}, \citenamefont {Hu},\ and\ \citenamefont
  {Wang}}]{jiang2019unconditional}%
  \BibitemOpen
  \bibfield  {author} {\bibinfo {author} {\bibfnamefont {C.}~\bibnamefont
  {Jiang}}, \bibinfo {author} {\bibfnamefont {Z.-W.}\ \bibnamefont {Yu}},
  \bibinfo {author} {\bibfnamefont {X.-L.}\ \bibnamefont {Hu}},\ and\ \bibinfo
  {author} {\bibfnamefont {X.-B.}\ \bibnamefont {Wang}},\ }\bibfield  {title}
  {\bibinfo {title} {Unconditional security of sending or not sending
  twin-field quantum key distribution with finite pulses},\ }\href
  {https://doi.org/10.1103/PhysRevApplied.12.024061} {\bibfield  {journal}
  {\bibinfo  {journal} {Phys. Rev. Applied}\ }\textbf {\bibinfo {volume}
  {12}},\ \bibinfo {pages} {024061} (\bibinfo {year} {2019})}\BibitemShut
  {NoStop}%
\bibitem [{\citenamefont {Ma}\ \emph {et~al.}(2012)\citenamefont {Ma},
  \citenamefont {Fung},\ and\ \citenamefont {Razavi}}]{ma2012statistical}%
  \BibitemOpen
  \bibfield  {author} {\bibinfo {author} {\bibfnamefont {X.}~\bibnamefont
  {Ma}}, \bibinfo {author} {\bibfnamefont {C.-H.~F.}\ \bibnamefont {Fung}},\
  and\ \bibinfo {author} {\bibfnamefont {M.}~\bibnamefont {Razavi}},\
  }\bibfield  {title} {\bibinfo {title} {Statistical fluctuation analysis for
  measurement-device-independent quantum key distribution},\ }\href
  {https://journals.aps.org/pra/abstract/10.1103/PhysRevA.86.052305} {\bibfield
   {journal} {\bibinfo  {journal} {Phys. Rev. A}\ }\textbf {\bibinfo {volume}
  {86}},\ \bibinfo {pages} {052305} (\bibinfo {year} {2012})}\BibitemShut
  {NoStop}%
\bibitem [{\citenamefont {Curty}\ \emph {et~al.}(2014)\citenamefont {Curty},
  \citenamefont {Xu}, \citenamefont {Cui}, \citenamefont {Lim}, \citenamefont
  {Tamaki},\ and\ \citenamefont {Lo}}]{curty2014finite}%
  \BibitemOpen
  \bibfield  {author} {\bibinfo {author} {\bibfnamefont {M.}~\bibnamefont
  {Curty}}, \bibinfo {author} {\bibfnamefont {F.}~\bibnamefont {Xu}}, \bibinfo
  {author} {\bibfnamefont {W.}~\bibnamefont {Cui}}, \bibinfo {author}
  {\bibfnamefont {C.~C.~W.}\ \bibnamefont {Lim}}, \bibinfo {author}
  {\bibfnamefont {K.}~\bibnamefont {Tamaki}},\ and\ \bibinfo {author}
  {\bibfnamefont {H.-K.}\ \bibnamefont {Lo}},\ }\bibfield  {title} {\bibinfo
  {title} {Finite-key analysis for measurement-device-independent quantum key
  distribution},\ }\href@noop {} {\bibfield  {journal} {\bibinfo  {journal}
  {Nature communications}\ }\textbf {\bibinfo {volume} {5}},\ \bibinfo {pages}
  {3732} (\bibinfo {year} {2014})}\BibitemShut {NoStop}%
\bibitem [{\citenamefont {Xu}\ \emph {et~al.}(2014)\citenamefont {Xu},
  \citenamefont {Xu},\ and\ \citenamefont {Lo}}]{xi2014protocol}%
  \BibitemOpen
  \bibfield  {author} {\bibinfo {author} {\bibfnamefont {F.}~\bibnamefont
  {Xu}}, \bibinfo {author} {\bibfnamefont {H.}~\bibnamefont {Xu}},\ and\
  \bibinfo {author} {\bibfnamefont {H.-K.}\ \bibnamefont {Lo}},\ }\bibfield
  {title} {\bibinfo {title} {Protocol choice and parameter optimization in
  decoy-state measurement-device-independent quantum key distribution},\ }\href
  {https://doi.org/10.1103/PhysRevA.89.052333} {\bibfield  {journal} {\bibinfo
  {journal} {Phys. Rev. A}\ }\textbf {\bibinfo {volume} {89}},\ \bibinfo
  {pages} {052333} (\bibinfo {year} {2014})}\BibitemShut {NoStop}%
\bibitem [{\citenamefont {Panayi}\ \emph {et~al.}(2014)\citenamefont {Panayi},
  \citenamefont {Razavi}, \citenamefont {Ma},\ and\ \citenamefont
  {L{\"u}tkenhaus}}]{panayi2014memory}%
  \BibitemOpen
  \bibfield  {author} {\bibinfo {author} {\bibfnamefont {C.}~\bibnamefont
  {Panayi}}, \bibinfo {author} {\bibfnamefont {M.}~\bibnamefont {Razavi}},
  \bibinfo {author} {\bibfnamefont {X.}~\bibnamefont {Ma}},\ and\ \bibinfo
  {author} {\bibfnamefont {N.}~\bibnamefont {L{\"u}tkenhaus}},\ }\bibfield
  {title} {\bibinfo {title} {Memory-assisted measurement-device-independent
  quantum key distribution},\ }\href
  {https://iopscience.iop.org/article/10.1088/1367-2630/16/4/043005/meta}
  {\bibfield  {journal} {\bibinfo  {journal} {New Journal of Physics}\ }\textbf
  {\bibinfo {volume} {16}},\ \bibinfo {pages} {043005} (\bibinfo {year}
  {2014})}\BibitemShut {NoStop}%
\bibitem [{\citenamefont {Azuma}\ \emph
  {et~al.}(2015{\natexlab{b}})\citenamefont {Azuma}, \citenamefont {Tamaki},\
  and\ \citenamefont {Munro}}]{azuma2015intercity}%
  \BibitemOpen
  \bibfield  {author} {\bibinfo {author} {\bibfnamefont {K.}~\bibnamefont
  {Azuma}}, \bibinfo {author} {\bibfnamefont {K.}~\bibnamefont {Tamaki}},\ and\
  \bibinfo {author} {\bibfnamefont {W.~J.}\ \bibnamefont {Munro}},\ }\bibfield
  {title} {\bibinfo {title} {All-photonic intercity quantum key distribution},\
  }\href {https://www.nature.com/articles/ncomms10171/} {\bibfield  {journal}
  {\bibinfo  {journal} {Nature communications}\ }\textbf {\bibinfo {volume}
  {6}},\ \bibinfo {pages} {1} (\bibinfo {year}
  {2015}{\natexlab{b}})}\BibitemShut {NoStop}%
\bibitem [{\citenamefont {Scarani}\ \emph {et~al.}(2009)\citenamefont
  {Scarani}, \citenamefont {Bechmann-Pasquinucci}, \citenamefont {Cerf},
  \citenamefont {Du\ifmmode~\check{s}\else \v{s}\fi{}ek}, \citenamefont
  {L\"utkenhaus},\ and\ \citenamefont {Peev}}]{Scarani2009security}%
  \BibitemOpen
  \bibfield  {author} {\bibinfo {author} {\bibfnamefont {V.}~\bibnamefont
  {Scarani}}, \bibinfo {author} {\bibfnamefont {H.}~\bibnamefont
  {Bechmann-Pasquinucci}}, \bibinfo {author} {\bibfnamefont {N.~J.}\
  \bibnamefont {Cerf}}, \bibinfo {author} {\bibfnamefont {M.}~\bibnamefont
  {Du\ifmmode~\check{s}\else \v{s}\fi{}ek}}, \bibinfo {author} {\bibfnamefont
  {N.}~\bibnamefont {L\"utkenhaus}},\ and\ \bibinfo {author} {\bibfnamefont
  {M.}~\bibnamefont {Peev}},\ }\bibfield  {title} {\bibinfo {title} {The
  security of practical quantum key distribution},\ }\href
  {https://doi.org/10.1103/RevModPhys.81.1301} {\bibfield  {journal} {\bibinfo
  {journal} {Rev. Mod. Phys.}\ }\textbf {\bibinfo {volume} {81}},\ \bibinfo
  {pages} {1301} (\bibinfo {year} {2009})}\BibitemShut {NoStop}%
\bibitem [{\citenamefont {Ferenczi}\ and\ \citenamefont
  {L\"utkenhaus}(2012)}]{Ferenczi2012symmetries}%
  \BibitemOpen
  \bibfield  {author} {\bibinfo {author} {\bibfnamefont {A.}~\bibnamefont
  {Ferenczi}}\ and\ \bibinfo {author} {\bibfnamefont {N.}~\bibnamefont
  {L\"utkenhaus}},\ }\bibfield  {title} {\bibinfo {title} {Symmetries in
  quantum key distribution and the connection between optimal attacks and
  optimal cloning},\ }\href {https://doi.org/10.1103/PhysRevA.85.052310}
  {\bibfield  {journal} {\bibinfo  {journal} {Phys. Rev. A}\ }\textbf {\bibinfo
  {volume} {85}},\ \bibinfo {pages} {052310} (\bibinfo {year}
  {2012})}\BibitemShut {NoStop}%
\bibitem [{\citenamefont {Xie}\ \emph {et~al.}(2021)\citenamefont {Xie},
  \citenamefont {Lu}, \citenamefont {Weng}, \citenamefont {Cao}, \citenamefont
  {Jia}, \citenamefont {Bao}, \citenamefont {Wang}, \citenamefont {Fu},
  \citenamefont {Yin},\ and\ \citenamefont {Chen}}]{xie2021breaking}%
  \BibitemOpen
  \bibfield  {author} {\bibinfo {author} {\bibfnamefont {Y.-M.}\ \bibnamefont
  {Xie}}, \bibinfo {author} {\bibfnamefont {Y.-S.}\ \bibnamefont {Lu}},
  \bibinfo {author} {\bibfnamefont {C.-X.}\ \bibnamefont {Weng}}, \bibinfo
  {author} {\bibfnamefont {X.-Y.}\ \bibnamefont {Cao}}, \bibinfo {author}
  {\bibfnamefont {Z.-Y.}\ \bibnamefont {Jia}}, \bibinfo {author} {\bibfnamefont
  {Y.}~\bibnamefont {Bao}}, \bibinfo {author} {\bibfnamefont {Y.}~\bibnamefont
  {Wang}}, \bibinfo {author} {\bibfnamefont {Y.}~\bibnamefont {Fu}}, \bibinfo
  {author} {\bibfnamefont {H.-L.}\ \bibnamefont {Yin}},\ and\ \bibinfo {author}
  {\bibfnamefont {Z.-B.}\ \bibnamefont {Chen}},\ }\href
  {https://arxiv.org/abs/2112.11635} {\bibinfo {title} {Breaking the rate-loss
  relationship of quantum key distribution with asynchronous two-photon
  interference}} (\bibinfo {year} {2021}),\ \Eprint
  {https://arxiv.org/abs/2112.11635} {arXiv:2112.11635 [quant-ph]} \BibitemShut
  {NoStop}%
\bibitem [{\citenamefont {Bennett}\ \emph {et~al.}(1992)\citenamefont
  {Bennett}, \citenamefont {Brassard},\ and\ \citenamefont
  {Mermin}}]{bbm1992quantum}%
  \BibitemOpen
  \bibfield  {author} {\bibinfo {author} {\bibfnamefont {C.~H.}\ \bibnamefont
  {Bennett}}, \bibinfo {author} {\bibfnamefont {G.}~\bibnamefont {Brassard}},\
  and\ \bibinfo {author} {\bibfnamefont {N.~D.}\ \bibnamefont {Mermin}},\
  }\bibfield  {title} {\bibinfo {title} {Quantum cryptography without bell's
  theorem},\ }\href {https://doi.org/10.1103/PhysRevLett.68.557} {\bibfield
  {journal} {\bibinfo  {journal} {Phys. Rev. Lett.}\ }\textbf {\bibinfo
  {volume} {68}},\ \bibinfo {pages} {557} (\bibinfo {year} {1992})}\BibitemShut
  {NoStop}%
\bibitem [{\citenamefont {Zhang}\ \emph {et~al.}(2017)\citenamefont {Zhang},
  \citenamefont {Zhao}, \citenamefont {Razavi},\ and\ \citenamefont
  {Ma}}]{Zhang2017improved}%
  \BibitemOpen
  \bibfield  {author} {\bibinfo {author} {\bibfnamefont {Z.}~\bibnamefont
  {Zhang}}, \bibinfo {author} {\bibfnamefont {Q.}~\bibnamefont {Zhao}},
  \bibinfo {author} {\bibfnamefont {M.}~\bibnamefont {Razavi}},\ and\ \bibinfo
  {author} {\bibfnamefont {X.}~\bibnamefont {Ma}},\ }\bibfield  {title}
  {\bibinfo {title} {Improved key-rate bounds for practical decoy-state
  quantum-key-distribution systems},\ }\href
  {https://journals.aps.org/pra/abstract/10.1103/PhysRevA.95.012333} {\bibfield
   {journal} {\bibinfo  {journal} {Phys. Rev. A}\ }\textbf {\bibinfo {volume}
  {95}},\ \bibinfo {pages} {012333} (\bibinfo {year} {2017})}\BibitemShut
  {NoStop}%
\bibitem [{\citenamefont {Fung}\ \emph {et~al.}(2010)\citenamefont {Fung},
  \citenamefont {Ma},\ and\ \citenamefont {Chau}}]{Fung2010practical}%
  \BibitemOpen
  \bibfield  {author} {\bibinfo {author} {\bibfnamefont {C.-H.~F.}\
  \bibnamefont {Fung}}, \bibinfo {author} {\bibfnamefont {X.}~\bibnamefont
  {Ma}},\ and\ \bibinfo {author} {\bibfnamefont {H.~F.}\ \bibnamefont {Chau}},\
  }\bibfield  {title} {\bibinfo {title} {Practical issues in
  quantum-key-distribution postprocessing},\ }\href
  {https://doi.org/10.1103/PhysRevA.81.012318} {\bibfield  {journal} {\bibinfo
  {journal} {Phys. Rev. A}\ }\textbf {\bibinfo {volume} {81}},\ \bibinfo
  {pages} {012318} (\bibinfo {year} {2010})}\BibitemShut {NoStop}%
\bibitem [{\citenamefont {Ma}(2008)}]{Ma2008PhD}%
  \BibitemOpen
  \bibfield  {author} {\bibinfo {author} {\bibfnamefont {X.}~\bibnamefont
  {Ma}},\ }\emph {\bibinfo {title} {Quantum cryptography: from theory to
  practice}},\ \href {https://arxiv.org/abs/0808.1385} {Ph.D. thesis},\
  \bibinfo  {school} {University of Toronto} (\bibinfo {year} {2008}),\
  \bibinfo {note} {also available in arXiv:0808.1385}\BibitemShut {NoStop}%
\bibitem [{\citenamefont {Ma}\ \emph {et~al.}(2005)\citenamefont {Ma},
  \citenamefont {Qi}, \citenamefont {Zhao},\ and\ \citenamefont
  {Lo}}]{Ma2005practical}%
  \BibitemOpen
  \bibfield  {author} {\bibinfo {author} {\bibfnamefont {X.}~\bibnamefont
  {Ma}}, \bibinfo {author} {\bibfnamefont {B.}~\bibnamefont {Qi}}, \bibinfo
  {author} {\bibfnamefont {Y.}~\bibnamefont {Zhao}},\ and\ \bibinfo {author}
  {\bibfnamefont {H.-K.}\ \bibnamefont {Lo}},\ }\bibfield  {title} {\bibinfo
  {title} {Practical decoy state for quantum key distribution},\ }\href
  {https://doi.org/10.1103/PhysRevA.72.012326} {\bibfield  {journal} {\bibinfo
  {journal} {Phys. Rev. A}\ }\textbf {\bibinfo {volume} {72}},\ \bibinfo
  {pages} {012326} (\bibinfo {year} {2005})}\BibitemShut {NoStop}%
\end{thebibliography}
\end{document}